\documentclass[11pt]{article}

\newsavebox{\foobox}
\newcommand{\slantbox}[2][0]{\mbox{%
        \sbox{\foobox}{#2}%
        \hskip\wd\foobox
        \pdfsave
        \pdfsetmatrix{1 0 #1 1}%
        \llap{\usebox{\foobox}}%
        \pdfrestore
}}
\newcommand\unslant[2][-.25]{\slantbox[#1]{$#2$}}

\newcommand{\mpi}{\text{\unslant[-.18]\pi}}

\newcommand{\CB}{\mathcal{B}}

\newcommand{\CH}{\mathcal{H}}
\newcommand{\CI}{\mathcal{I}}
\newcommand{\CJ}{\mathcal{J}}

\newcommand{\CL}{\mathcal{L}}

\newcommand{\CO}{\mathcal{O}}

\newcommand{\BP}{\mathbb{P}}

\newcommand{\lV}{\lVert}
\newcommand{\rV}{\rVert}
\newcommand{\vertiii}[1]{{\left\vert\kern-0.25ex\left\vert\kern-0.25ex\left\vert #1 \right\vert\kern-0.25ex\right\vert\kern-0.25ex\right\vert}}
\newcommand{\norm}[1]{\Vert {#1} \Vert}

\newcommand{\normp}[2]{\norm{#1}_{#2}}
\newcommand{\fnorm}[1]{\norm{#1}_{\mathrm{F}}}

\newcommand{\e}{\mathrm{e}}

\newcommand*{\tr}{\mathrm{Tr}}

\newcommand{\tx}[1]{\text{#1}}
\newcommand{\AC}[1]{
{\color{black}#1}
}
\usepackage[left=2cm, right=2cm, top=2.5cm, bottom=2.5cm]{geometry}
\usepackage[x11names]{xcolor}
\usepackage{fancyhdr, amssymb, cancel, amsmath, graphicx, pgfplots, tikz}
\usepackage{isomath}

\usetikzlibrary{shadows}

\allowdisplaybreaks

\newcommand{\stylecolor}{IndianRed3}

\usepackage[labelfont={bf,sf, color=\stylecolor}, margin={1.5cm,0cm}]{caption}

\usepackage[colorlinks=true, urlcolor=\stylecolor, linkcolor=\stylecolor, citecolor=\stylecolor, hyperindex=true, linktocpage=true]{hyperref}

\usepackage{amsthm}
\usepackage{comment}

\usepackage[explicit]{titlesec}

\newcommand*\sectionlabel{}
\titleformat{\section}
  {\gdef\sectionlabel{}
   \Large\bfseries\scshape}
  {\gdef\sectionlabel{\thesection }}{0pt}
  {\begin{tikzpicture}[remember picture,overlay]
       \end{tikzpicture}
  }
\titlespacing*{\section}{0pt}{0pt}{0pt}

\newcommand*\subsectionlabel{}
\titleformat{\subsection}
  {\gdef\subsectionlabel{}
   \large\bfseries\scshape}
  {\gdef\subsectionlabel{\thesubsection  }}{0pt}
  {\begin{tikzpicture}[remember picture]
    	\draw (-0.15, 0) node[left] {\color{\stylecolor} \textsf{\subsectionlabel}};
	\draw (0.15, 0) node[right] {\color{\stylecolor} \textsf{#1}};
	\fill[color=\stylecolor] (-0.05, -0.23) rectangle (0.05, 0.23);
       \end{tikzpicture}
  }
\titlespacing*{\subsection}{-4pt}{10pt}{0pt}

\newcommand*\subsubsectionlabel{}
\titleformat{\subsubsection}
  {\gdef\subsubsectionlabel{}
   \bfseries\scshape}
  {\gdef\subsubsectionlabel{\thesubsubsection.\ \  }}{0pt}
  {\begin{tikzpicture}[remember picture]
    	\draw (0, 0) node[left] {\color{\stylecolor} \textsf{\subsubsectionlabel}};
	\draw (0, 0) node[right] {\color{\stylecolor} \textsf{#1}};
       \end{tikzpicture}
  }
\titlespacing*{\subsubsection}{-4pt}{7pt}{0pt}

\pgfplotsset{every axis legend/.append style={at={(1.02,1)},anchor=north west}}

\usepackage[framemethod=tikz]{mdframed}

\pgfplotsset{every axis legend/.append style={at={(1.02,1)},anchor=north west}}

   \newcounter{exam}
   \renewcommand{\theexam}{\arabic{exam}}

\usetikzlibrary{calc, arrows}

      \newcounter{theor}
   \renewcommand{\thetheor}{\arabic{theor}}
\newenvironment{theor}[1][]{%
 \refstepcounter{theor}%
  \ifstrempty{#1}%
  {\mdfsetup{%
     frametitle={%
        {\strut \color{white} \textsf{Theorem~\thetheor}}}}%
   }%
  {\mdfsetup{%
     frametitle={%
        {\strut \color{white} \textsf{Theorem~\thetheor:~#1}}}}%
   }%
   \mdfsetup{innertopmargin=\topskip,linecolor=\stylecolor, frametitlerule=true,
             linewidth=0pt, backgroundcolor=\stylecolor!15!white, topline=false, frametitlebackgroundcolor=\stylecolor}
   \begin{mdframed}[]\relax%
   }{\end{mdframed}}
   
   \newenvironment{corol}[1][]{%
 \refstepcounter{theor}%
  \ifstrempty{#1}%
  {\mdfsetup{%
     frametitle={%
        {\strut \color{white} \textsf{Corollary~\thetheor}}}}%
   }%
  {\mdfsetup{%
     frametitle={%
        {\strut \color{white} \textsf{Corollary~\thetheor:~#1}}}}%
   }%
   \mdfsetup{innertopmargin=\topskip,linecolor=\stylecolor, frametitlerule=true,
             linewidth=0pt, backgroundcolor=\stylecolor!15!white, topline=false, frametitlebackgroundcolor=\stylecolor}
   \begin{mdframed}[]\relax%
   }{\end{mdframed}}
   
   \newenvironment{prop}[1][]{%
 \refstepcounter{theor}%
  \ifstrempty{#1}%
  {\mdfsetup{%
     frametitle={%
        {\strut \color{white} \textsf{Proposition~\thetheor}}}}%
   }%
  {\mdfsetup{%
     frametitle={%
        {\strut \color{white} \textsf{Proposition~\thetheor:~#1}}}}%
   }%
   \mdfsetup{innertopmargin=\topskip,linecolor=\stylecolor, frametitlerule=true,
             linewidth=0pt, backgroundcolor=\stylecolor!15!white, topline=false, frametitlebackgroundcolor=\stylecolor}
   \begin{mdframed}[]\relax%
   }{\end{mdframed}}
   
      \newenvironment{lma}[1][]{%
 \refstepcounter{theor}%
  \ifstrempty{#1}%
  {\mdfsetup{%
     frametitle={%
        {\strut \color{white} \textsf{Lemma~\thetheor}}}}%
   }%
  {\mdfsetup{%
     frametitle={%
        {\strut \color{white} \textsf{Lemma~\thetheor:~#1}}}}%
   }%
   \mdfsetup{innertopmargin=\topskip,linecolor=\stylecolor, frametitlerule=true,
             linewidth=0pt, backgroundcolor=\stylecolor!15!white, topline=false, frametitlebackgroundcolor=\stylecolor}
   \begin{mdframed}[]\relax%
   }{\end{mdframed}}
   
   \newenvironment{theorNB}[1][]{%
 \refstepcounter{theor}%
  \ifstrempty{#1}%
  {\mdfsetup{%
     frametitle={%
        {\strut \color{white} \textsf{Theorem~\thetheor}}}}%
   }%
  {\mdfsetup{%
     frametitle={%
        {\strut \color{white} \textsf{Theorem~\thetheor:~#1}}}}%
   }%
   \mdfsetup{innertopmargin=\topskip,linecolor=\stylecolor, frametitlerule=true,
             linewidth=0pt, backgroundcolor=\stylecolor!15!white, topline=false, nobreak=true, frametitlebackgroundcolor=\stylecolor}
   \begin{mdframed}[]\relax%
   }{\end{mdframed}}
   
   \newenvironment{corolNB}[1][]{%
 \refstepcounter{theor}%
  \ifstrempty{#1}%
  {\mdfsetup{%
     frametitle={%
        {\strut \color{white} \textsf{Corollary~\thetheor}}}}%
   }%
  {\mdfsetup{%
     frametitle={%
        {\strut \color{white} \textsf{Corollary~\thetheor:~#1}}}}%
   }%
   \mdfsetup{innertopmargin=\topskip,linecolor=\stylecolor, frametitlerule=true,
             linewidth=0pt, backgroundcolor=\stylecolor!15!white, topline=false, nobreak=true, frametitlebackgroundcolor=\stylecolor}
   \begin{mdframed}[]\relax%
   }{\end{mdframed}}
   
   \newenvironment{propNB}[1][]{%
 \refstepcounter{theor}%
  \ifstrempty{#1}%
  {\mdfsetup{%
     frametitle={%
        {\strut \color{white} \textsf{Proposition~\thetheor}}}}%
   }%
  {\mdfsetup{%
     frametitle={%
        {\strut \color{white} \textsf{Proposition~\thetheor:~#1}}}}%
   }%
   \mdfsetup{innertopmargin=\topskip,linecolor=\stylecolor, frametitlerule=true,
             linewidth=0pt, backgroundcolor=\stylecolor!15!white, topline=false, nobreak=true, frametitlebackgroundcolor=\stylecolor}
   \begin{mdframed}[]\relax%
   }{\end{mdframed}}
   
      \newenvironment{lmaNB}[1][]{%
 \refstepcounter{theor}%
  \ifstrempty{#1}%
  {\mdfsetup{%
     frametitle={%
        {\strut \color{white} \textsf{Lemma~\thetheor}}}}%
   }%
  {\mdfsetup{%
     frametitle={%
        {\strut \color{white} \textsf{Lemma~\thetheor:~#1}}}}%
   }%
   \mdfsetup{innertopmargin=\topskip,linecolor=\stylecolor, frametitlerule=true,
             linewidth=0pt, backgroundcolor=\stylecolor!15!white, topline=false, nobreak=true, frametitlebackgroundcolor=\stylecolor}
   \begin{mdframed}[]\relax%
   }{\end{mdframed}}

\begin{document}


\pagestyle{fancy}
\renewcommand{\headrulewidth}{0pt}
\fancyhead{}

\fancyfoot{}
\fancyfoot[C] {\textsf{\textbf{\thepage}}}

\begin{equation*}
\begin{tikzpicture}
\draw (\textwidth, 0) node[text width = \textwidth, right] {\color{white} easter egg};
\end{tikzpicture}
\end{equation*}

\begin{equation*}
\begin{tikzpicture}
\draw (0.5\textwidth, -3) node[text width = \textwidth] {\huge  \textsf{\textbf{Operator growth bounds from graph theory}} };
\end{tikzpicture}
\end{equation*}
\begin{equation*}
\begin{tikzpicture}
\draw (0.5\textwidth, 0.1) node[text width=\textwidth] {\large \color{black} \textsf{Chi-Fang Chen$^{\color{\stylecolor} a,b}$ and Andrew Lucas$^{\color{\stylecolor} a,c,d}$}};
\draw (0.5\textwidth, -0.5) node[text width=\textwidth] {\small \textsf{\; $^{\color{\stylecolor} a}$ Department of Physics, Stanford University, Stanford, CA 94305, USA}};
\draw (0.5\textwidth, -1) node[text width=\textwidth] {\small \textsf{\; $^{\color{\stylecolor} b}$ Institute for Quantum Information and Matter,
California Institute of Technology, Pasadena, CA, USA}};
\draw (0.5\textwidth, -1.5) node[text width=\textwidth] {\small \textsf{\; $^{\color{\stylecolor} c}$ Department of Physics, University of Colorado, Boulder CO 80309, USA}};
\draw (0.5\textwidth, -2) node[text width=\textwidth] {\small \textsf{\; $^{\color{\stylecolor} d}$ Center for Theory of Quantum Matter, University of Colorado, Boulder CO 80309, USA}};
\end{tikzpicture}
\end{equation*}
\begin{equation*}
\begin{tikzpicture}
\draw (0, -13.15) node[right, text width=0.5\paperwidth] { \texttt{chifang@caltech.edu, andrew.j.lucas@colorado.edu}};
\draw (\textwidth, -13.1) node[left] {\textsf{\today}};
\end{tikzpicture}
\end{equation*}
\begin{equation*}
\begin{tikzpicture}
\draw[very thick, color=\stylecolor] (0.0\textwidth, -5.75) -- (0.99\textwidth, -5.75);
\draw (0.12\textwidth, -6.25) node[left] {\color{\stylecolor}  \textsf{\textbf{Abstract:}}};
\draw (0.53\textwidth, -6) node[below, text width=0.8\textwidth, text justified] {\small  
Let $A$ and $B$ be local operators in Hamiltonian quantum systems with $N $ degrees of freedom  and finite-dimensional Hilbert space.  We prove that the commutator norm $\lVert [A(t),B]\rVert$ is upper bounded by a topological combinatorial problem: counting irreducible weighted paths between two points on the Hamiltonian's factor graph. \AC{Our bounds sharpen existing Lieb-Robinson bounds by removing extraneous growth.} In quantum systems drawn from zero-mean random ensembles with few-body interactions, we prove stronger bounds on the ensemble-averaged out-of-time-ordered correlator $\mathbb{E}\left[ \fnorm{ [A(t),B]}^2\right]$.   In such quantum systems on Erd\"os-R\'enyi factor graphs, we prove that the scrambling time $t_{\mathrm{s}}$, at which $\fnorm{ [A(t),B]}=\mathrm{\Theta}(1)$, is almost surely $t_{\mathrm{s}}=\mathrm{\Omega}(\sqrt{\log N})$; we further prove $t_{\mathrm{s}}=\mathrm{\Omega}(\log N) $ to high order in perturbation theory in $1/N$.  We constrain infinite temperature quantum chaos in the $q$-local Sachdev-Ye-Kitaev model at any order in $1/N$; at leading order, our upper bound on the Lyapunov exponent is within a factor of 2 of the known result at any $q>2$.  We also speculate on the implications of our theorems for conjectured holographic descriptions of quantum gravity.
 };
\end{tikzpicture}
\end{equation*}

\tableofcontents

\begin{equation*}
\begin{tikzpicture}
\draw[very thick, color=\stylecolor] (0.0\textwidth, -5.75) -- (0.99\textwidth, -5.75);
\end{tikzpicture}
\end{equation*}

\titleformat{\section}
  {\gdef\sectionlabel{}
   \Large\bfseries\scshape}
  {\gdef\sectionlabel{\thesection }}{0pt}
  {\begin{tikzpicture}[remember picture]
	\draw (0.2, 0) node[right] {\color{\stylecolor} \textsf{#1}};
	\fill[color=\stylecolor]  (0,0.37) rectangle (-0.7, -0.37);
	\draw (0.0, 0) node[left, fill=\stylecolor] {\color{white} \textsf{\sectionlabel}};
       \end{tikzpicture}
  }
\titlespacing*{\section}{0pt}{20pt}{5pt}

\section{Introduction}
\subsection{Operator Growth in Quantum Many-Body Systems}
Consider a many-body quantum system of $N\gg 1$ degrees of freedom which interact via few-body interactions.   How quickly is quantum information lost?   It has been conjectured by Sekino and Susskind \cite{susskind08} that the \emph{scrambling time} $t_{\mathrm{s}}$ for quantum information obeys a universal bound
 \begin{equation}
t_{\mathrm{s}} = \mathrm{\Omega}(\log N);  \label{eq:susskind}
\end{equation}
namely, any physical quantum system cannot lose information in a time that does not grow at least logarithmically with the number of degrees of freedom.   In the literature, (\ref{eq:susskind}) is coined the \emph{fast scrambling conjecture}.   

One formal definition for scrambling of quantum information is the time at which two halves of a quantum many-body system, prepared in a tensor product state, become nearly maximally entangled.   With this definition, the inspiration for (\ref{eq:susskind}) is that quantum black holes appear to clone states \cite{hayden07} (which is impossible \cite{dieks, zurek}) unless (\ref{eq:susskind}) holds.   It was proved \cite{lucas1805} that the time to generate nearly maximal entanglement between two halves of a many-body quantum system obeys (\ref{eq:susskind}) whenever generic local two-point correlation functions decay exponentially with time. 

A weaker notion of ``scrambling", which is more accessible experimentally \cite{Garttner2017, Li2017}, is the spreading of operators from local operators into non-local operators under Heisenberg time evolution.  This is an early time manifestation of what is called ``many-body quantum chaos".   One practical way of measuring operator spreading is with out-of-time-ordered correlators (OTOCs) \cite{shenker13}.   Since infinite temperature OTOCs are equivalent to the 2-norm of an operator commutator, generalizations \cite{lucas1805, lashkari} of the Lieb-Robinson theorem \cite{liebrobinson} have been used to rigorously bound OTOCs and scrambling.  Unfortunately, these bounds are too weak to imply (\ref{eq:susskind}) in many physically relevant models.

\subsection{Summary of this Paper}
We now present a heuristic overview of our key results and their implications.  (The rest of the paper will be mathematically precise.)

Consider a Hamiltonian quantum many-body system with $N\gg1$ degrees of freedom, interacting via few-body interactions.   We write the Hamiltonian as \begin{equation}
H = \sum_X H_X,
\end{equation}
where $X$ correspond to subsets of the degrees of freedom with O(1) elements.  Let $A_i$ and $B_j$ be local operators acting on degrees of freedom $i$ and $j$.   Our goal is to bound the commutator norm $\lVert [A_i(t),B_j]\rVert$.  This object tells us what ``fraction" of the time evolved operator $A_i(t)$ acts non-trivially on the Hilbert space of $j$.  This was historically done following Lieb and Robinson \cite{liebrobinson}, who repeatedly invoke the triangle inequality on the Heisenberg evolved $[A_i(t),B_j] = [A_i + \mathrm{i}t[H,A_i]   -\frac{1}{2}t^2[H,[H,A_i]] + \mathrm{O}(t^3),B_j]$.
As we will explain, this perspective does not give sharp bounds on operator growth. \AC{Perhaps a more useful intuitive starting point, which we develop in this paper, is to consider instead the operator $A_i(t)$ as a \emph{vector} in the real vector space of Hermitian operators \cite{poulin}.   The commutator $[A_i(t),B_j]$ is only non-zero if the vector $A_i(t)$ ``points" in certain directions.  Clearly, we should try to separate out terms in the time evolution that ``rotate" our vector in the directions of interest, while ignoring rotations in planes orthogonal to directions of interest.   In particular, this ``rotational analogy" suggests that we might look for the \emph{critical steps} in a sequence of $[H,\cdots]$ commutators which rotate our vector in the right direction, and resum the remainder (this resummation should not contribute to our commutator bound).  To emphasize this alternative intuition, we will often replace the commutator norm bounds with the following bounds on the \emph{projected operator} $\mathbb{P}_j A_i(t)$:}
\begin{equation}
\lVert \mathbb{P}_j A_i(t) \rVert \sim \sup_{B_j} \frac{\lVert [A_i, B_j]\rVert }{\lVert B_j \rVert}.
\end{equation}   
This notation helps us to clarify that most terms in the nested commutators of Lieb and Robinson are ``rotations" of $\mathbb{P}_j A_i(t)$ which leave $\lVert \mathbb{P}_j A_i(t)\rVert$ fixed.  We formalize which terms are important and which terms are not using the 
 \emph{factor graph} structure of the Hamiltonian.  Even without knowing each $H_X$ explicitly, we still restrict how operators grow, classifying the terms in the Taylor expanded $\mathbb{P}_jA_i(t)$ into different topological classes of operator growth.  Remarkably, each topological class is resummable, using a generalization of the interaction picture of quantum mechanics.  This exact resummation leads to a simple bound on $\lVert \mathbb{P}_jA_i(t) \rVert$ (Theorem~\ref{theor3}).  \AC{And moreover, this approach works for both the conventional operator norm of a commutator, as well as other norms including Frobenius norms or Schatten $p$-norms.}


The classic Lieb-Robinson bound \cite{liebrobinson} is a simple corollary of our bound (Corollary \ref{corolLR}), and it is instructive to compare our results to theirs.  The key historical insight of the Lieb-Robinson bound was that \begin{equation}
\lVert [A_i(t),B_j]\rVert \lesssim \mathrm{e}^{\lambda_*(t-d(i,j)/v_{\mathrm{LR}})},  \label{eq:introLR}
\end{equation}
where $d(i,j)$ is the distance between degrees of freedom $i$ and $j$, $\lambda_*$ is a \emph{quantum Lyapunov exponent} and $v_{\mathrm{LR}}$ is the \emph{Lieb-Robinson velocity} which bounds the spatial growth of operators.   (\ref{eq:introLR}) suggests that the domain in which operators have support expands out at the velocity $v_{\mathrm{LR}}$, which is readily computed following \cite{hastings}.  We prove that $v_{\mathrm{LR}}$, as computed by applying the triangle inequality to nested commutators, overestimates the speed at which this domain grows by (at least) a factor of 2  (Proposition \ref{prop8}).   We further show that the speed is overestimated by (at least) a factor of  4 in a one dimensional lattice model.

In many large-$N$ quantum systems which involve all-to-all couplings between degrees of freedom, one expects that \cite{shenker13}
\AC{\begin{equation}
\lVert [A_i(t),B_j]\rVert_{\mathrm{F}} \lesssim \frac{1}{N} \mathrm{e}^{\lambda_*t},  \label{eq:introLR}
\end{equation}}
where $\lambda_*$ is independent of $N$.  If we define $t_{\mathrm{s}}^*$ as the time when the commutator norm is O(1), then we find $t_{\mathrm{s}}^* \gtrsim \log N$, in agreement with (\ref{eq:susskind}).   (\ref{eq:introLR}) can be proven from the Lieb-Robinson bound in models where $\sum \lVert H_X\rVert = \mathrm{O}(N)$  \cite{lucas1805, lashkari}  (see also Proposition \ref{propFS3}).

However, the scaling of individual terms in the Hamiltonian is often far too large to invoke the Lieb-Robinson bound or Theorem \ref{theor3} directly:  in many physical models of large-$N$ systems, $\sum \lVert H_X\rVert^2 = \mathrm{O}(N)$.   We formally study random ensembles of Hamiltonians where it is natural to impose this scaling of $\lVert H_X\rVert$.  Bounding the ensemble-averaged $\mathbb{E}\left[\lVert\mathbb{P}_jA_i(t)\rVert^2\right]$ differs from our previous bound in two ways: (\emph{1}) the square allows for the operator to grow in ``two different ways" which need not be the same; (\emph{2}) randomness demands that each $H_X$ show up at least twice in the nested commutators (of either side).  The appropriate classification of operator growth, while retaining its topological character, is now more complicated (Theorem~\ref{theor4}).

We then study random ensembles of quantum systems on Erd\"os-R\'enyi hypergraphs where each coupling constant (on average) takes a similar value.   We prove that (on average) \AC{\begin{equation}
\lVert \mathbb{P}_j A_i(t)\rVert^2_{\mathrm{F}} \lesssim c_1 \frac{\mathrm{e}^{\lambda_*t}}{N} \sum_{g=0}^{N-1} g! \left(\frac{c_2 t^2 \mathrm{e}^{\lambda_*t}}{N}\right)^{g} \label{eq:introasymptotic}
\end{equation}}
where $c_{1,2}$ and $\lambda_*$ are $N$-independent constants (Theorem \ref{theorFS}).  In the above formula, $g$ corresponds to a topological graph genus.   Using (\ref{eq:introasymptotic}), we can perturbatively prove (\ref{eq:susskind}) to any order $g=\mathrm{O}(N^{1-\delta})$ with $\delta>0$.   Perhaps more interestingly, we also find that contributions at order $g\ge N$ cannot grow the commutator any further.   This truncation of the series at order $g=N-1$ implies that it is absolutely convergent in the $N\rightarrow \infty$ limit  when $\lambda^*_{\mathrm{L}}t \lesssim 1$.  So while we cannot formally prove the fast scrambling conjecture (\ref{eq:susskind}), we also did not find an asymptotic series with zero radius of convergence, as one might expect from a generic quantum many-body system \cite{dyson}.  

We are also able to find an alternative bound (Theorem \ref{theorsqrtlogN}): \AC{\begin{equation}
\lVert \mathbb{P}_j A_i(t)\rVert^2_{\mathrm{F}} \lesssim \frac{c_1}{N}\left(\mathrm{e}^{\lambda_*t}-1\right) + \frac{c_2}{N^2} \left(\mathrm{e}^{at^2}-1\right)
\end{equation}}
with $a=\mathrm{O}(N^0)$.   This allows us to prove that the scrambling time $t_{\mathrm{s}} = \mathrm{\Omega}(\sqrt{\log N})$, almost surely on any Erd\"os-R\'enyi hypergraph.  While not quite as strong as (\ref{eq:susskind}), this result shows that the scrambling time diverges in the thermodynamic limit.   We find both Theorem~\ref{theorFS} and Theorem~\ref{theorsqrtlogN} to be of value, as the former bound is much stronger to any finite order in the perturbative expansion, whereas the latter bound gives a sharper non-perturbative constraint on the scrambling time. 


An example of a random regular model for which Theorem~\ref{theorFS} and Theorem~\ref{theorsqrtlogN} are non-trivial is the $q$-local Sachdev-Ye-Kitaev (SYK) model \cite{sachdevye, maldacena2016remarks, suh}.   This model is conjectured to be holographically dual to some theory of quantum gravity in two dimensions in an asymptotically $\mathrm{AdS}_2$ spacetime.   Applying Theorem \ref{theorFS} to the SYK model, we find that our bound on the Lyapunov exponent is within a factor of 2 of the known result \cite{stanford1802} at leading order in $\frac{1}{N}$ at all values of $q>2$ for which the model is chaotic.  The origin of the discrepancy will be explained.   We also provide new analytic constraints on infinite temperature OTOCs at any $q$ and any order in $\frac{1}{N}$.

Since our bounds describe qualitatively correct physics in the SYK model, it is tempting to speculate that our theorems may have interesting implications for theories of quantum gravity.   \AC{While we could not say anything rigorous, we do make a few comments below.   We also note (see the Epilogue section at the end of this paper) that in subsequent work to this paper, we were able to prove the fast scrambling conjecture in the SYK model, using generalizations of the ideas which we had originally developed in this work.}

Firstly, the asymptotic nature of our bound (\ref{eq:introasymptotic}) is natural in the context of diagrammatic quantum many-body theory: the factor $g!$ is a consequence of the $\mathrm{O}(g!)$ distinct graph topologies of genus $g$ \cite{riddell}.  We will find that the $g!$ in (\ref{eq:introasymptotic}) also has this origin.  In the more specific context of string theory, these asymptotic series can signal the presence of non-perturbative effects -- branes \cite{polchinski, stanford1903}. Our results, which are of an entirely combinatorial nature, also suggest such non-perturbative effects could be generic in any model with sufficient spatial non-locality to be fast scramblers.

On the other hand, Theorem~\ref{theorsqrtlogN} also implies that these non-perturbative effects are always ``mild", since even as $N\rightarrow \infty$, $C_{ij}(t)\rightarrow 0$ almost surely for any finite time $t=\mathrm{O}(N^0)$.  \AC{As noted above, this ``mildness" was indeed later proven to be the case.} 

In the conjectured holographic duality between certain large-$N$ quantum theories and gravity \cite{maldacena}, the small parameter $\frac{1}{N}$ which controls (\ref{eq:introasymptotic}) is proportional to Newton's gravitational constant $G_{\mathrm{N}}$ in the dual quantum theory.  $G_{\mathrm{N}}$ is the coupling constant for quantum gravity, and as $G_{\mathrm{N}} \rightarrow 0$, perturbation theory suggests that quantum fluctuations of spacetime become negligible and gravity may be treated semiclassically.  As our work strongly suggests that (at least for certain correlation functions) the $\frac{1}{N}$ expansion is well-behaved, even accounting for any possible non-perturbative effects, our work could then further imply that semiclassical gravity is a sensible limit of a true, ultraviolet-complete and non-perturbative quantum theory of gravity.   While the validity of the semiclassical approximation is often assumed in the physics literature, it would be interesting if this result follows from elementary mathematical aspects of graph theory and $q$-local many-body quantum mechanics, and was not a necessary postulate in physics.

\section{Preparatory Formalism} \label{sec:formalism} 

\subsection{Heisenberg Operator Evolution}
\AC{
Having informally introduced the subject of the paper, and our main results, we now precisely define the relevant terms introduced above.   Let $V$ be a discrete set whose elements denote quantum degrees of freedom; define $|V|:=N$.  We study many-body quantum systems whose Hilbert space $\mathcal{H}$ is a finite dimensional complex vector space, expressed as a tensor product \begin{equation}
\mathcal{H} := \bigotimes_{i\in V} \mathcal{H}_i \label{eq:calHdef}
\end{equation}
We assume that \begin{equation}
\dim(\mathcal{H}_i) = d_i = \mathrm{\Theta}(N^0)
\end{equation}
for every $i$, and that $d_i > 1$ for every $i$.

Among the norms one can equip with $O \in \CB = \tx{End}(\CH)$, the community studying many-body chaos largely studies the \textit{Hilbert-Schmidt norm} or \textit{Frobenius norm}, which is conveniently also an inner-product on the vector space of operators:
\begin{align}
    \fnorm{\CO} := (\mathcal{O}|\mathcal{O}),
\end{align}
where
\begin{equation}
(\mathcal{O}_1|\mathcal{O}_2) = \frac{\mathrm{tr}(\mathcal{O}_1\mathcal{O}_2)}{\dim(\mathcal{H})}.  \label{eq:innerproduct}
\end{equation}
For example, an orthonormal basis of $\mathcal{B}$ may be constructed as \begin{equation}
|T^{a_i}_i) := \bigotimes_{i=1}^N  T^{a_i}_i,
\end{equation} where $T^0_i = 1_i$ is the identity and $T^a_i$ for $a_i=0,\ldots,d_i^2-1$ are some orthogonal traceless operators on site $i$ that
\begin{align}
    \mathrm{tr}(T^a_i T^b_i) = \mathbb{I}(a=b) d_i.
\end{align}
Here $\mathbb{I}(\cdots)$ denotes the indicator function which returns 1 if its argument is true, and 0 if false.   

We may define orthogonal projectors $\mathbb{P}_i$ onto all operators which act non-trivially on $\mathcal{H}_i$:\begin{equation}
\mathbb{P}_i |T^{a_i}_i) := \mathbb{I}(a_i\ne0) |T^{a_i}_i).
\end{equation}
Note that these projectors are well-defined whether or not we are concerned with the Frobenius inner product above.
Throughout this work, we will denote an operator as $\mathcal{O}_j$ whenever $\mathbb{P}_i|\mathcal{O}_j) = \mathbb{I}(i=j) |\mathcal{O}_j)$.   The space of all such $\mathcal{O}_j$ is $\mathcal{B}_j \subset \mathcal{B}$.   If $X\subseteq \lbrace 1,\ldots,N\rbrace$, we define an $X$-local operator by $\mathbb{P}_j |\mathcal{O}) = \mathbb{I}(j\in X)|\mathcal{O})$.   We denote $\mathcal{B}_X\subset \mathcal{B}$ to be the space of all $X$-local operators.   We define a $q$-local operator to be a sum of $X_\alpha$-local operators, where each set $X_\alpha$ has at most $q$ elements:  $|X_\alpha|\le q$.   As an example, in the notation above, $\mathcal{O}_j$ is $\{j\}$-local and more generally 1-local. 

We study Hamiltonains which take the form\begin{equation}
H := \sum_{X \in F} H_X, \label{eq:HX}
\end{equation}
with $F\subseteq \mathbb{Z}_2^V$ a set of subsets of $V$, and $H_X \in \mathcal{B}_X$. We assume that each $H_X$ is nonzero, $\lVert H_X\rVert > 0$; by construction, $\mathrm{tr}(H_X)=0$.   In this formula, $\lVert H_X\rVert$ corresponds to the operator norm of $H_X$, i.e. the maximal singular value of $H_X$.  In this paper, we are interested in $q$-local Hamiltonians, with $q = \mathrm{O}(N^0)$.    Hence if $X\in F$, $|X|\le q$.   
Heisenberg evolution can then be regarded as a rotation\footnote{Note that $\mathcal{L}$ are antisymmetric, given the inner product (\ref{eq:innerproduct}).  Clearly, $\mathrm{e}^{\mathcal{L}t} \in \mathrm{SO}(\dim(\mathcal{B}))$ is an orthogonal transformation. }
\begin{equation}
|\mathcal{O}(t)) := \mathrm{e}^{\mathcal{L}t} |\mathcal{O}),
\end{equation}
which is generated by the Liouvillian $\mathcal{L} \in \mathrm{End}(\mathcal{B})$
\begin{equation}
\mathcal{L}|\mathcal{O}) := |\mathrm{i}[H,\mathcal{O}]).
\end{equation}
We will also define \begin{equation}
\mathcal{L}_X|\mathcal{O}) = |\mathrm{i}[H_X,\mathcal{O}]).
\end{equation}

We can now write out-of-time-ordered correlators (OTOCs), the main objects of interest in this work, in the prevailing notation 
\begin{equation}
C_{ij}(t):=\frac{\fnorm{\mathbb{P}_j |A_i(t))}}{\fnorm{|A_i(t))}} :=  \sqrt{\frac{(A_i(t)|\mathbb{P}_j |A_i(t))}{(A_i|A_i)}}.   \label{eq:Cijdef}
\end{equation}
We will always assume that $(A_i|A_i)=1$ and ignore the denominator, as it is otherwise a distraction: all bounds developed in this paper do not depend on the choice of $A_i \in \mathcal{B}_i$. 

To interpret $\fnorm{\mathbb{P}_j |A_i(t))}$, it is useful to consider $|\mathcal{O})$ as analogous to the wave function in ordinary quantum mechanics, and $\mathrm{i} \mathcal{L}$ as a ``Hamiltonian" on ``operator space". Now, $\fnorm{\mathbb{P}_j |A_i(t))}^2$ represents the fraction of the growing operator $|A_j(t))$ that acts non-trivially on the tensor factor $\mathcal{H}_j$.   Indeed, it is obvious from (\ref{eq:Cijdef}) that $0\le C_{ij}(t) \le 1$.   It is more natural to talk about $C_{ij}(t)$ because, as in quantum mechanics, it is the probability amplitude that most naturally evolves under automorphisms including time evolution $\mathrm{e}^{\mathcal{L}t}$.

We must warn the reader that some of the properties and interpretations above are unique to inner-product spaces, i.e. the Frobenius norm. Another natural choice of norms for operators $O \in \tx{End}(\CH)$ is one of the Schatten p-norms for $1\le p\le \infty$:
\begin{align}
    \lV O \rV_p: = \mathrm{tr}[(O^\dagger O)^{p/2}]^{1/p}.
\end{align}
Only when $p=2$ does this (Frobenius) norm lead to an inner product space.  Nevertheless, the key result of this paper (Theorem \ref{theor3}) will hold both for Frobenius norms, and for $p$-norms.  Indeed, the $p\rightarrow\infty$ limit of the Schatten norm gives the conventional operator norm, which is typically used in studies of Lieb-Robinson bounds in the mathematical physics literature.
So we will also describe how to bound
\begin{equation}
\hat{C}_{ij}^{(p)}(t):= \sup_{ B_j \in \mathcal{B}_j} \frac{\lVert [A_i(t),B_j]\rVert_p}{2\lVert A_i\rVert_p \lVert B_j\rVert_\infty}.  \label{eq:hatCijdef}
\end{equation}
In the limit $p\rightarrow\infty$, a bound on $\lVert [A_i(t),B_j]\rVert$ should be understood as a ``worst case bound" since the operator norm $\lVert \cdots \rVert$ is the largest singular value. Note that the Frobenius norm is upper bounded by the operator norm
\begin{align}
    \fnorm{\CO} =\frac{\normp{\CO}{2}}{\sqrt{\tr[I]}} \le \normp{\CO}{\infty}.
\end{align}
Because the operator norm leads to a ``worst case" bound, it is often not as useful as Frobenius norms for understanding many-body chaos, or for understanding random systems.   As discussed in the Epilogue, subsequent work to this paper has uncovered important subtleties and differences between how quickly Frobenius norms can grow, versus operator norms, in certain classes of many-body systems.
}

\subsection{(Factor) Graph Theory}

We will find it useful to organize the information about $H_X$ into a \emph{factor graph} $G=(V,F,E)$: see Figure~\ref{fig:factorgraph}. The \emph{node set} $V$ and \emph{factor set} $F$ have already been defined above; the \emph{edge set} $E \subseteq V\times F$ is defined as 
\begin{figure}[t]
\centering
\includegraphics[width=3.5in]{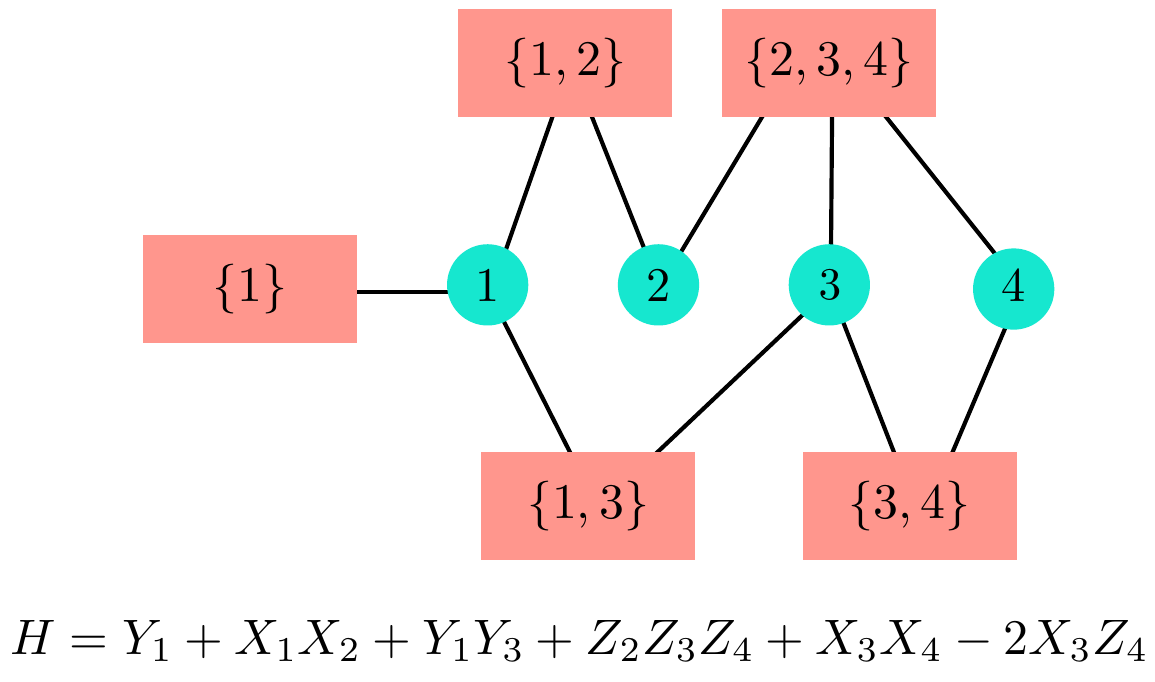}
\caption{A factor graph $G=(V,F,E)$ for the Hamiltonian printed above.   Throughout this paper we will denote elements of $V$ as light blue circles, elements of $F$ as light red squares/rectangles and elements of $E$ as solid black lines.}
\label{fig:factorgraph}
\end{figure}
\begin{equation}
E = \lbrace (i,X) : i\in X\text{ and } X\in F\rbrace.  \label{eq:Edef}
\end{equation}
While factor graphs were introduced in the study of classical coding theory \cite{loeliger}, they are natural for studying $q$-local quantum systems as well \cite{lucas1805}.  We will often abuse the shorthand notation $i\in G$ or $X\in G$ to mean $i\in V$ or $X\in F$ respectively, whenever the intent is clear from context.   The boundary operation $\partial : V \rightarrow F$ is defined as $\partial v := \lbrace X\in F : (v,X)\in E\rbrace$;  similarly, $\partial : F\rightarrow V$ is defined by $\partial X := \lbrace v\in V : (v,X) \in E\rbrace$.  We may write $\partial_V$ or $\partial_F$ if it is important to emphasize which definition is being used.  The boundary of a set is the union of the boundary of subsets: e.g. $\partial_V \lbrace i_1,\ldots, i_n\rbrace = \partial_V i_1 \cup \cdots \cup \partial_V i_n$.\footnote{This is because a factor graph is ``bipartite" when interpreted as an ordinary graph.}   The degree of a node or factor $\deg_G : V\cup F \rightarrow \mathbb{Z}^+$ is defined as $\deg_G(v) = |\partial v|$.   We define the distance function $d_G: V\cup F \times V\cup F \rightarrow \mathbb{Z}^+$ as follows: \begin{equation}
d_G(a,b) = \min \left[|E^\prime| \; : \; (V^\prime, F^\prime, E^\prime) \subseteq G \text{ is connected, and } \lbrace a,b\rbrace \subseteq V^\prime \cup F^\prime\right] \label{eq:defdistfactor}
\end{equation}
i.e. it is the number of edges one traverses on the factor graph $G$ to get from $a$ to $b$;  note that $a$ and $b$ can be either nodes or factors.   We assume that $G$ is connected, so $d_G(a,b) \le 2(N-1)$.   When the context is clear, we will usually drop the subscript from $d_G$.

\section{Non-Random Systems}
We now introduce a formalism suitable for studying operator growth in general Hamiltonian systems, especially those on factor graphs of low degree.  We begin by writing (for $j\ne i$) \begin{align}
\mathbb{P}_j\mathrm{e}^{\CL t}|\CO_i) &= \mathbb{P}_j\exp\left[ \sum_{X\in F} \mathcal{L}_X t\right]|\CO_i)  = \mathbb{P}_j\sum^\infty_{n=1} \frac{t^n}{n!}\sum_{X_1, X_2, X_n\in F} \CL_{X_n}\cdots \CL_{X_1}|\CO_i). \label{eq:taylorseries3}
\end{align}
Not every sequence above is interesting.  In particular, we note the following three identities, which follow immediately from the definitions in Section \ref{sec:formalism}:
\begin{subequations}\label{eq:3simpleidentities}\begin{align}
\mathcal{L}_X|\mathcal{O}_i) &= 0 \;\;\; \text{if} \;\;\; i\notin X,  \label{eq:LXOi} \\
[\mathbb{P}_j, \mathcal{L}_X] &= 0\;\;\; \text{if and only if} \;\;\; j\notin X,  \label{eq:PjLX} \\ 
[\mathcal{L}_X,\mathcal{L}_Y] &= 0 \;\;\; \text{if} \;\;\; X\cap Y = \emptyset.  \label{eq:LXLYcommute}
\end{align}\end{subequations}
The purpose of the next subsection is to use these simple identities to organize sequences in (\ref{eq:taylorseries3}).   We emphasize that the main result of this section, Theorem~\ref{theor3}, follows from (\ref{eq:taylorseries3}) and (\ref{eq:3simpleidentities}) -- no further information about the Hamiltonian is needed.

\subsection{Causal Trees}\label{sec:causaltrees}
Our first goal is to develop a topological classification for sequences of Liouvillians acting on operators $\mathcal{L}_{X_n}\cdots \mathcal{L}_{X_1}|\mathcal{O}_i)$.   For convenience, we write this ordered sequence using only the graph theoretic information:  $\mathcal{M}:=(i,X_1,\ldots, X_n)$.   

Using (\ref{eq:LXLYcommute}), it is clear that some ordered sequences are very similar and will lead to the same operator:  if $\mathcal{M}$ contains two factors $X_k$ and $X_{k+1}$ with $X_k \cap X_{k+1}=\emptyset$, we should not care which came first.   Rather than keeping track of the entire sequence $\mathcal{M}$, we will only keep track of pairs of couplings which must occur in a specific order.  

Given a sequence $\mathcal{M}$, we define an ordered sequence of \emph{causal forests} $T_0,\ldots, T_n$.  Each $T_k$ can be thought of as an undirected graph on the vertex set $\lbrace i \rbrace \cup F$.   We recursively construct $T_n$ as follows: \begin{itemize}
\item $T_0 = \lbrace i \rbrace$
\item Given $T_{n-1}$, we construct $T_n$ as follows: \begin{itemize}
\item If there is an integer $1\le k < n$ with $X_k=X_n$, $T_n=T_{n-1}$;
\item else if $i\in X_n$, $T_n = T_{n-1} \cup (X_n, (i,X_n))$;
\item else if $1\le k<n$ is the smallest value of $k$ such that $X_k \cap X_n \ne \emptyset$, $T_n = T_{n-1} \cup (X_n, (X_k,X_n))$;
\item else $T_n = T_{n-1} \cup (X_n,\emptyset)$.
\end{itemize}
\end{itemize}
We say that $T_k$ is a \emph{causal tree} if it is simply connected: namely, it has a single connected component.   If a sequence has $n$ elements, we denote the final causal tree $T_n$ as $T(\mathcal{M})$.

\begin{propNB}
$\mathcal{L}_{X_n}\cdots \mathcal{L}_{X_1}|\mathcal{O}_i) = 0$ if $T((i,X_1,\ldots, X_n))$ is not a causal tree.
\label{propcreeping}
\end{propNB}
\begin{proof}
This follows immediately from (\ref{eq:LXOi}) and (\ref{eq:LXLYcommute}).
\end{proof}

 Intuitively, the causal tree stores information about which factors in $\mathcal{M}$ are the first to grow $|\mathcal{O}_i)$ into a larger operator (acting on a larger subset of $V$).     We say that if $T(\mathcal{M})$ is simply connected,  the sequence $\mathcal{M}$ exhibits \emph{creeping order} and that $\mathcal{L}_{X_n}\cdots \mathcal{L}_{X_1}$ is \emph{creeping}.  

It is often helpful to ``embed" causal trees $T(\mathcal{M})$ as subtrees of the Hamilonian's factor graph $G=(V,F,E)$.  To do this is straightforward: since two connected factors in $T(\mathcal{M})$ can only be connected by an edge if they share a node in common, we embed $T(\mathcal{M})$ such that if edge $(X_1,X_2)\in T(\mathcal{M})$, we choose a vertex $v\in X_1\cap X_2$ and add connect $X_1$ and $X_2$ to $v$.  Slightly abusing notation, we will often refer to this subgraph of $G$ as the causal tree $T(\mathcal{M})$.  Figure \ref{fig:causaltree} shows a sequence $\mathcal{M}$ together with its causal tree as a subtree in the factor graph.  
  
\begin{figure}[t]
\centering
\includegraphics[width=0.45\textwidth]{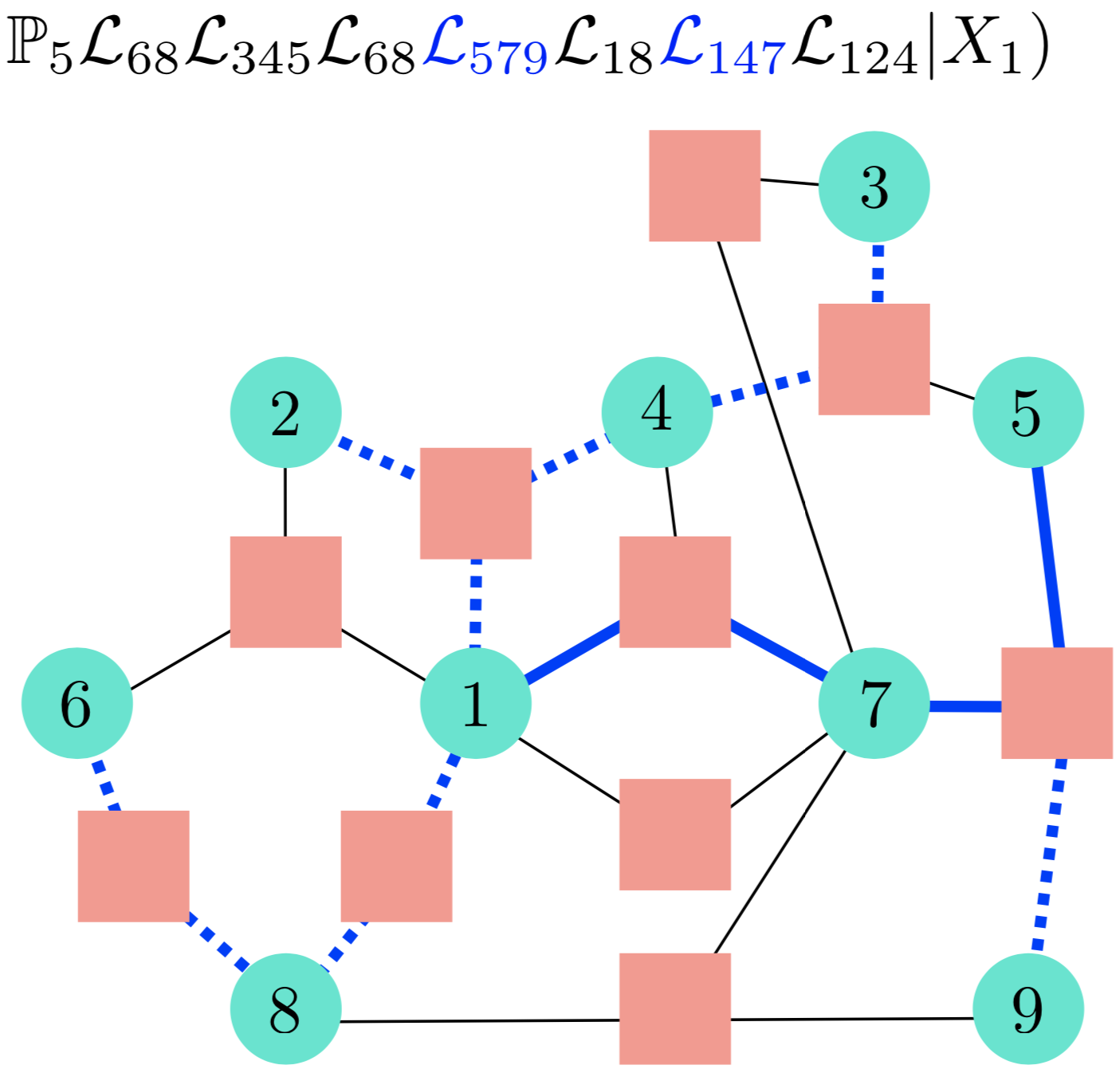}
\caption{The causal tree consists of all blue paths.   The irreducible path is solid blue, and the rest of the causal tree is dashed blue.   Blue $\mathcal{L}$s are part of the irreducible path.}
\label{fig:causaltree}
\end{figure}

An important subtlety when thinking of causal trees as subtrees of $G$ is that they need not be unique if the Hamiltonian is $q$-local with $q>2$.   A simple example of a non-unique $T(\mathcal{M})\subseteq G$ is shown in Figure \ref{fig:indistinguishable}.  We say that any two causal trees which arise from the same sequence $\mathcal{M}$ are \emph{indistinguishable}: they have the same causal tree of factors alone, and their only difference is the specific embedding into the full factor graph.   In our formalism, indistinguishable causal trees will always be treated as identical.  

\begin{figure}[t]
\centering
\includegraphics[width=0.75\textwidth]{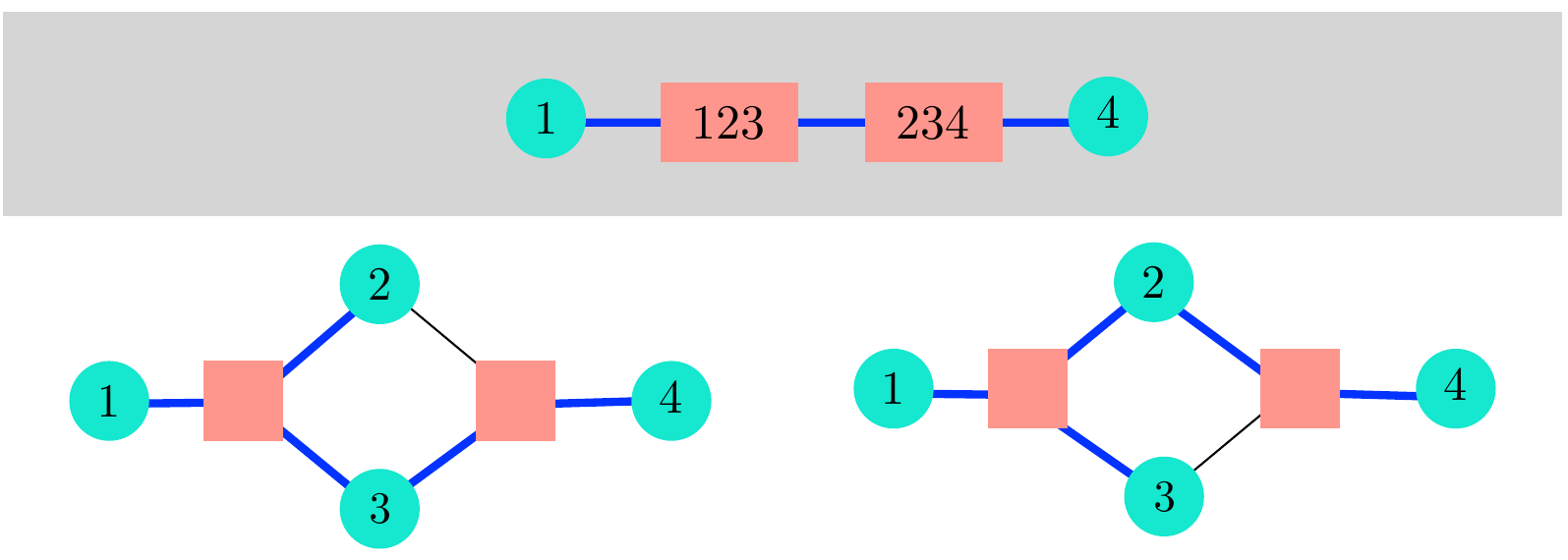}
\caption{Top/shaded: the causal tree as a unique tree of factors, together with nodes $i$ and $j$.  Bottom: two indistinguishable causal trees on the full factor graph.   Our formalism treats indistinguishable causal trees as equal.}
\label{fig:indistinguishable}
\end{figure}


\begin{propNB}
\AC{If $i\ne j$ and $\mathbb{P}_j \mathcal{L}_{X_n}\cdots \mathcal{L}_{X_1}|\mathcal{O}_i) \ne 0$, then there exists a unique self-avoiding path from $i$ to $j$ in the causal tree  $T((i,X_1,\ldots,X_n))$. We will call it the \emph{irreducible path}.} \label{propcausaltree}
\end{propNB}
\begin{proof}
\AC{
Using (\ref{eq:PjLX}), $\mathbb{P}_j \mathcal{L}_{X_n}\cdots \mathcal{L}_{X_1}|\mathcal{O}_i) = \mathcal{L}_{X_n}\cdots \mathcal{L}_{X_1} \mathbb{P}_j |\mathcal{O}_i) = 0$ if $j\notin X_1 \cup \cdots \cup X_n$.   We conclude that $j \in T(\mathcal{M})\cap V$.\footnote{This notation will be used frequently in this paper to denote that $j$ is a vertex, and lies in a certain graph -- in this case, the causal tree $T(\mathcal{M})$.  Also note that here we are explicitly thinking of the causal tree as a subgraph of the full factor graph.}.   Then there is a unique self-avoiding path from $i$ to $j$ in a tree graph containing both $i$ and $j$~\cite{stallings}.
}
\end{proof}

Let $\Gamma\subseteq T(\mathcal{M})$ be the irreducible path (ordered set) of factors obtained from Proposition~\ref{propcausaltree}.  We define the equivalence relation $\sim_{ji}$ on causal trees as follows:  if $T_{1,2}$ are causal trees with $\lbrace i,j\rbrace \subseteq V\cap T_{1,2}$, then $T_1\sim_{ji}T_2$ if and only if they share the same irreducible path of factors.  The uniqueness of irreducible path means this equivalence relation is well-defined.  We define $\mathcal{T}_{ji}$ as the set of all (distinguishable) causal trees containing $i$ and $j$, and $\mathcal{S}_{ji}=\mathcal{T}_{ji}/\sim_{ji}$ to be the set of irreducible paths from $i$ to $j$.  Alternatively, $\mathcal{S}_{ji}$ is the set of equivalence classes of causal tree.   We write $\Gamma \in \mathcal{S}_{ji}$ to mean the sequence of factors $\Gamma$ on the irreducible path.  We define the length $\ell(\Gamma)$ as the number of factors in $\Gamma$, and $X^\Gamma_k$ for $1 \le k \le \ell(\Gamma)$ to be the $k^{\mathrm{th}}$ factor in the path $\Gamma$.\footnote{When thinking of $\Gamma$ as a line subgraph of the factor graph $G$, $X^\Gamma_k \in \Gamma$ is the unique factor obeying $d_G(i,X^\Gamma_k) = 2k-1$.}  We refer to the equivalence relation $\sim_{ji}$ as topological because (up to indistinguishability) $\Gamma_1 \sim_{ji} \Gamma_2$ if and only if $\Gamma_1$ and $\Gamma_2$ are homotopic paths between $i$ and $j$ in the factor graph.

\subsection{Bounds from Irreducible Paths}
We begin with our first main result: a theorem relating $C_{ij}(t)$ to a combinatorial problem on the factor graph $G=(V,F,E)$.

\begin{theorNB}
If $H$ is a Hamiltonian on factor graph $G=(V,F,E)$, $\lbrace i,j\rbrace \subseteq V$ and $i\ne j$,
\AC{
\begin{align}
C_{ij}(t)=\frac{\fnorm{\BP_j\e^{\CL t}|\CO_i)} }{\fnorm{|\CO_i)}}&\le \sum_{[\Gamma] \in \mathcal{S}_{ji}} \frac{(2|t|)^{\ell(\Gamma)}}{\ell(\Gamma)!} \prod_{X\in \Gamma} \lVert H_X\rVert . \\
\hat{C}^{(p)}_{ij}(t)=\frac{\normp{\CL_{B_j}\e^{\CL t}|\CO_i)}{p}}{2\lV B_j \rV\cdot \normp{|\CO_i)}{p} } &\le \sum_{[\Gamma] \in \mathcal{S}_{ji}} \frac{(2|t|)^{\ell(\Gamma)}}{\ell(\Gamma)!} \prod_{X\in \Gamma} \lVert H_X\rVert\cdot  \label{eq:theor3}
\end{align}}
\label{theor3}
\end{theorNB}
\begin{proof}
\AC{We will present proof for the Frobenius norm.  However, the exact same steps will also work for general Schatten $p$-norms, up to a minor modification to the normalization factor which arises from the triangle ineqality. }

The strategy of proof here is as follows.    (\emph{1}) We show that $\fnorm{\BP_j\e^{\CL t}|\CO_i)}$ is given by summing over all creeping sequences, and organize this sum by topological class (Lemma~\ref{lemma4}).   (\emph{2}) Next, we prove a generalized Schwinger-Karplus identity to exponentiate the Liouvillians which contribute reducibly to each topological class (Lemma~\ref{lemma5}).  (\emph{3}) Combining these two lemmas, we prove that the reducible terms in each creeping sequence never grow $\fnorm{\BP_j\e^{\CL t}|\CO_i)}$ and bound only the growth arising from the irreducible path, obtaining (\ref{eq:theor3}).

\emph{Step 1:} We begin by organizing creeping sequences by topology.   Using Proposition~\ref{propcausaltree}, we may write \begin{align}
\mathbb{P}_j\mathrm{e}^{\mathcal{L}t}|A_i) &= \mathbb{P}_j \sum_{n=1}^\infty \frac{t^n}{n!} \sum_{T(\lbrace i, X_1,\ldots, X_n\rbrace) \in \mathcal{T}_{ji} } \mathcal{L}_{X_n}\cdots \mathcal{L}_{X_1}|A_i) \notag \\
&= \mathbb{P}_j\sum_{[\Gamma] \in \mathcal{S}_{ji}} \sum_{n=1}^\infty \frac{t^n}{n!} \sum_{T(i,X_1,\ldots, X_n)\in[\Gamma]} \mathcal{L}_{X_n}\cdots \mathcal{L}_{X_1}|A_i). \label{eq:firsttheor3}
\end{align}

\begin{lma}
For $[\Gamma]\in\mathcal{S}_{ji}$, 
\begin{align}
&\sum_{T(i,X_1,\ldots,X_n) \in [\Gamma]}  \frac{t^n}{n!}\CL_{X_n}\cdots \CL_{X_1}|\mathcal{O}_i) = \notag \\
&\;\;\;\;\;\;\;\sum_{m_0,\ldots, m_\ell =0}^\infty \frac{t^{(\ell+\sum\limits^\ell_{k=0} m_k)}}{(\ell+\sum\limits^\ell_{k=0} m_k)!} \mathcal{L}^{m_\ell}  \mathcal{L}_{X_{\ell}^\Gamma} (\mathcal{L}^\Gamma_{\ell-1})^{m_{\ell-1}}  \cdots \mathcal{L}_{X_{1}^\Gamma} (\mathcal{L}^\Gamma_0)^{m_0} |\mathcal{O}_i)\label{interactionsum}
\end{align}
where $\ell = \ell(\Gamma)$, \begin{equation}\label{neighbourhood}
\mathcal{L}^\Gamma_k := \mathcal{L} - \sum_{Y\in F: |Y\cap V^\Gamma_k| > 0 } \mathcal{L}_Y,
\end{equation}
and\begin{equation}
V^\Gamma_k := \left\lbrace \begin{array}{ll} \lbrace j\rbrace &\ k=\ell(\Gamma)-1 \\ \displaystyle  \bigcup_{m=2+k}^{\ell(\Gamma)} X^\Gamma_m &\ 0 \le k<\ell(\Gamma)-1 \end{array}\right.. \label{eq:VGammak}
\end{equation}
is the set of forbidden vertices between steps $k$ and $k+1$ of irreducible path $\Gamma$: if  $\mathcal{M}$ with $T(\mathcal{M})\in [\Gamma]$, then $\mathcal{L}_X$ does not appear in between $\mathcal{L}_{X^\Gamma_k}$ and $\mathcal{L}_{X^\Gamma_{k+1}}$ if $X\cap V^\Gamma_k = \emptyset$.
\label{lemma4}
\end{lma}
\begin{proof}
  Every term on the right hand side of (\ref{interactionsum}) forms a causal tree $T\in [\Gamma]$ by construction, and thus appears on the left hand side.  Conversely, every term on the left hand side must be expressible as a term on the right hand side:  using (\ref{eq:VGammak}) and our algorithm for constructing the causal tree, any $T$ which contains a factor $Y$ in between $X^\Gamma_k$ and $X^\Gamma_{k+1}$ with $Y\cap V^\Gamma_k \ne \emptyset$ is mapped to a different equivalence class: see Figure \ref{fig:lemma4}.   Since there is a bijection between all terms on each side of (\ref{interactionsum}), and each term on both sides of (\ref{interactionsum}) has the same real coefficient 1, the equality (\ref{interactionsum}) is established.
\end{proof}

\begin{figure}[t]
\centering
\includegraphics[width=0.7\textwidth]{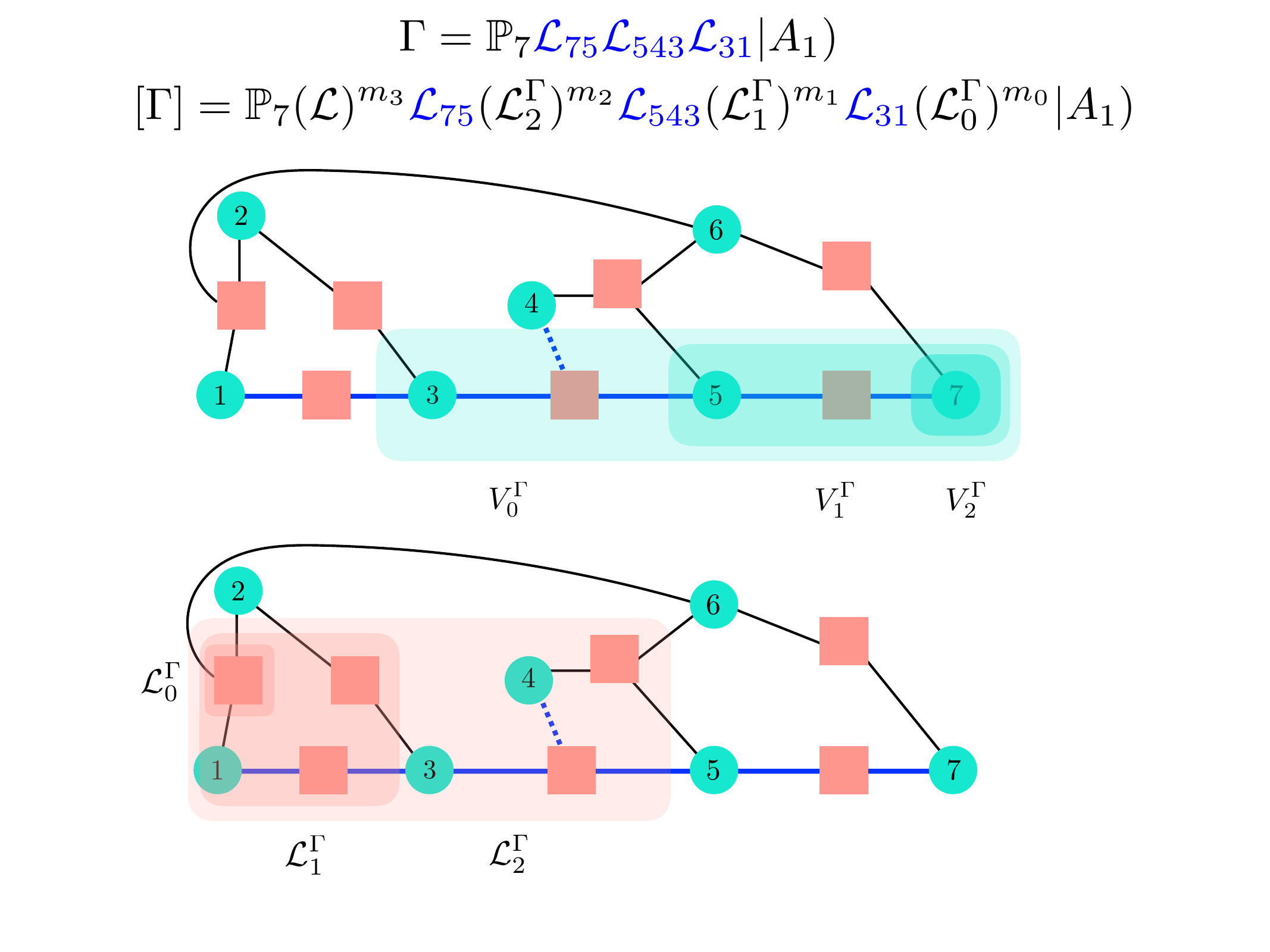}
\caption{The sets $V^\Gamma_k$ are chosen to preserve the topological class $[\Gamma]$, as are $\mathcal{L}^\Gamma_k$.  Shaded regions of different opacities denote  ``blocked" regions in between each step of the irreducible path.  }
\label{fig:lemma4}
\end{figure}

\emph{Step 2}:  Next, we define the canonical $n$-simplex 
\begin{equation}
\mathrm{\Delta}^n(t) = \lbrace (t_1,\ldots, t_n)\in [0,t]^n :  t_1\le t_2\le \cdots \le t_n\rbrace,
\end{equation}
which has volume \begin{equation}
\mathrm{Vol}(\mathrm{\Delta}^n(t)) = \frac{t^n}{n!}. \label{eq:volsimplex}
\end{equation}
\begin{lma}[\textsf{\textbf{Generalized Schwinger-Karplus Identity}}]
Let $\lbrace \mathcal{F}_0,\mathcal{A}_1, \mathcal{F}_1, \ldots, \mathcal{A}_\ell, \mathcal{F}_\ell\rbrace \subset \mathrm{End}(\mathcal{B})$, $t\in\mathbb{R}$,
\begin{subequations} \begin{align}
\mathcal{I}(t) &:= \sum_{m_0,\ldots, m_\ell = 0}^\infty  \frac{t^{\ell + \sum\limits_{k=0}^\ell m_k } }{(\ell + \sum\limits_{k=0}^\ell m_k )!}  \mathcal{F}_\ell^{m_\ell} \mathcal{A}_\ell \mathcal{F}_{\ell-1}^{m_{\ell-1}} \cdots \mathcal{F}_1^{m_1}\mathcal{A}_1\mathcal{F}_0^{m_0} , \\
\widetilde{\mathcal{I}} (t) & := \int\limits_{\mathrm{\Delta}^\ell(t)}\mathrm{d}t_1 \cdots\mathrm{d}t_\ell\;  \mathrm{e}^{\mathcal{F}_\ell (t-t_\ell)} \mathcal{A}_\ell \mathrm{e}^{\mathcal{F}_{\ell-1}(t_\ell-t_{\ell-1})} \cdots \mathrm{e}^{\mathcal{F}_1(t_2-t_1)}\mathcal{A}_1\mathrm{e}^{\mathcal{F}_0t_1}.
\end{align}\end{subequations}
Then $\mathcal{I}(t) = \widetilde{\mathcal{I}}(t)$.
\label{lemma5}
\end{lma}
\begin{proof}
We prove by induction.  In the base case $\ell=0$, \begin{equation}
\mathcal{I}(t) = \sum_{m=0}^\infty \frac{t^m}{m!} \mathcal{F}_0^m = \mathrm{e}^{\mathcal{F}_0t} = \widetilde{\CI}(t).
\end{equation}
Now assume that the identity above holds for when the value of $\ell$ is reduced by 1.  Then observe that \begin{equation}
\frac{\mathrm{d}}{\mathrm{d}t}\widetilde{\CI}(t)= \mathcal{F}_\ell \widetilde{\CI}(t) +  \mathcal{A}_\ell \int\limits_{\mathrm{\Delta}^{\ell-1}(t)} \mathrm{d}t_1\cdots \mathrm{d}t_{\ell-1} \; \mathrm{e}^{\mathcal{F}_{\ell-1}(t-t_{\ell-1}) }\mathcal{A}_{\ell-1} \cdots  \mathrm{e}^{\mathcal{F}_1(t_2-t_1)}\mathcal{A}_1\mathrm{e}^{\mathcal{F}_0t_1};
\end{equation}
\begin{align}
\frac{\mathrm{d}}{\mathrm{d}t}\mathcal{I}(t) &= \frac{\mathrm{d}}{\mathrm{d}t}\sum_{m_0,\ldots, m_{\ell -1}= 0}^\infty \left[ \sum_{m_\ell=1}^\infty  \frac{1}{(\ell + \sum\limits_{k=0}^\ell m_k )!} (\mathcal{F}_\ell t)^{m_\ell}  + \frac{1}{(\ell + \sum\limits_{k=0}^{\ell-1} m_k )!} \right] \mathcal{A}_\ell (\mathcal{F}_{\ell-1}t)^{m_{\ell-1}}\cdots (\mathcal{A}_1t)(\mathcal{F}_0t)^{m_0} \notag \\
&= \sum_{m_0,\ldots, m_{\ell-1} = 0}^\infty \left[ \sum_{m_\ell=0}^\infty  \frac{1}{(\ell + \sum\limits_{k=0}^\ell m_k )!} \mathcal{F}_\ell^{1+m_\ell}t^{m_\ell}  \right. \notag \\
&\left.\;\;\;\; \;\;\;\;\;\;\;\;\;\;\;\;\;\;\;\;\;\;\;\;\;+ \frac{1}{(\ell-1+ \sum\limits_{k=1}^\ell m_k )!} \right]\mathcal{A}_\ell (\mathcal{F}_{\ell-1}t)^{m_{\ell-1}} \cdots (\mathcal{F}_1t)^{m_1}(\mathcal{A}_1t)(\mathcal{F}_0t)^{m_0} \notag \\
&= \mathcal{F}_0 \mathcal{I}(t) + \mathcal{A}_\ell \sum_{m_0,\ldots,m_{\ell-1} = 0}^\infty \frac{1}{(\ell-1+ \sum\limits_{k=1}^\ell m_k )!} (\mathcal{F}_{\ell-1}t)^{m_{\ell-1}} (\mathcal{A}_{\ell-1}t)\cdots (\mathcal{A}_1t)(\mathcal{F}_0t)^{m_0}.
\end{align}
Using the identity at smaller values of $\ell$, we conclude that $\mathcal{I}(t)$ and $\widetilde{\CI}(t)$ obey the same first order linear ordinary differential equation.   When $\ell >0$, they also share the same initial condition $\mathcal{I}(0)=\widetilde{\mathcal{I}}(0)=0$.  Hence they are equal:  $\mathcal{I}(t)=\widetilde{\CI}(t)$.
\end{proof}

\emph{Step 3:} Now consider the following equalities.  Starting with (\ref{eq:firsttheor3}), we obtain 
\begin{align}
\mathbb{P}_j \mathrm{e}^{\mathcal{L}t} |\mathcal{O}_i) 
&= \mathbb{P}_j \sum_{\Gamma \in \mathcal{S}_{ji}} \sum_{m_0,\ldots, m_\ell = 0}^\infty  \frac{1}{ (\ell + \sum\limits_{k=0}^\ell m_\ell)!} (\mathcal{L} t)^{m_\ell}  \mathcal{L}_{X_{\ell}^\Gamma} t(\mathcal{L}^\Gamma_{X_{\ell-1}})^{m_{\ell-1}}  \cdots (\mathcal{L}_{X_1^\Gamma} t) (\mathcal{L}^\Gamma_0t)^{m_0} |\mathcal{O}_i) \notag \\
&=\mathbb{P}_j \sum_{\Gamma \in \mathcal{S}_{ji}} \int\limits_{\mathrm{\Delta}^{\ell(\Gamma)}(t)} \mathrm{d}t_1\cdots\mathrm{d}t_\ell \; \mathrm{e}^{\mathcal{L}(t-t_\ell)} \mathcal{L}_{X^\Gamma_\ell}\cdots \mathrm{e}^{\mathcal{L}^\Gamma_1(t_2-t_1)}\mathcal{L}_{X^\Gamma_1} \mathrm{e}^{\mathcal{L}^\Gamma_0t_1}|\mathcal{O}_i)
\label{eq:lemma4and5}
\end{align}
where we used Lemma \ref{lemma4} in the first line and Lemma~\ref{lemma5} in the second line. 

Since $\mathcal{L}_Y$ is antisymmetric for any $Y\in F$, each $\mathcal{L}^\Gamma_j$ is antisymmetric.  Then $\fnorm{ |\mathcal{O}} = \fnorm{ \mathrm{e}^{\mathcal{L}^\Gamma_j t} |\mathcal{O})}$ for any $t\in \mathbb{R}$, and 
\begin{equation}
\fnorm{ \mathcal{L}_X |\mathcal{O})} \le 2 \lVert H_X\rVert \fnorm{|\mathcal{O})}.  \label{eq:LXsubmult}
\end{equation} 
   Using (\ref{eq:lemma4and5}), 
\begin{align}
\fnorm{\BP_j\e^{\CL t}|\CO_i)} &=  \left[ \fnorm{ \mathbb{P}_j\sum_{\Gamma \in \mathcal{S}_{ji}}  \int\limits_{\mathrm{\Delta}^\ell(t)} \mathrm{d}t_1\cdots \mathrm{d}t_\ell \mathrm{e}^{\mathcal{L}(t-t_\ell)} \mathcal{L}_{X_\ell^\Gamma}  \mathrm{e}^{\mathcal{L}^\Gamma_{\ell-1} (t_\ell-t_{\ell-1})} \mathcal{L}_{X_{\ell-1}^\Gamma}\cdots \mathrm{e}^{\mathcal{L}^\Gamma_1(t_2-t_1)}\mathcal{L}_{X^\Gamma_1}\mathrm{e}^{\mathcal{L}^\Gamma_0t_1} |\mathcal{O}_i) } \right] \notag \\
&\le \left[ \sum_{\Gamma \in \mathcal{S}_{ji}} \int\limits_{\mathrm{\Delta}^\ell(t)} \mathrm{d}t_1\cdots \mathrm{d}t_\ell \fnorm{ \mathrm{e}^{\mathcal{L}(t-t_1)} \mathcal{L}_{X_1}  \mathrm{e}^{\mathcal{L}^\Gamma_1 (t_1-t_2)} \mathcal{L}_{X_2}\cdots \mathrm{e}^{\mathcal{L}^\Gamma_{\ell-1} (t_{\ell-1}-t_\ell)} \mathcal{L}_{X_\ell}\mathrm{e}^{\mathcal{L}^\Gamma_{\ell} t_\ell} |\mathcal{O}_i) } \right] \notag \\
&\le \sum_{\Gamma\in\mathcal{S}_{ji}} \mathrm{Vol}\left(\mathrm{\Delta}^{\ell(\Gamma)}(t)\right) 2^{\ell(\Gamma)}  \prod_{X\in \Gamma} \lVert H_X\rVert \cdot \fnorm{|\mathcal{O})}. \label{eq:thrm3proof}
\end{align}
In the third line above, we used the triangle inequality and unitary invariance of Frobenius norm.  Combining (\ref{eq:volsimplex}) and (\ref{eq:thrm3proof}) proves the theorem. For general p-norm, replace $\BP_j$ with $\CL_{B_j}\BP_j$.
\end{proof}

A simple application of Theorem~\ref{theor3} is depicted in Figure~\ref{fig:theor3ex}.

\begin{figure}[t]
\centering
\includegraphics[width=0.65\textwidth]{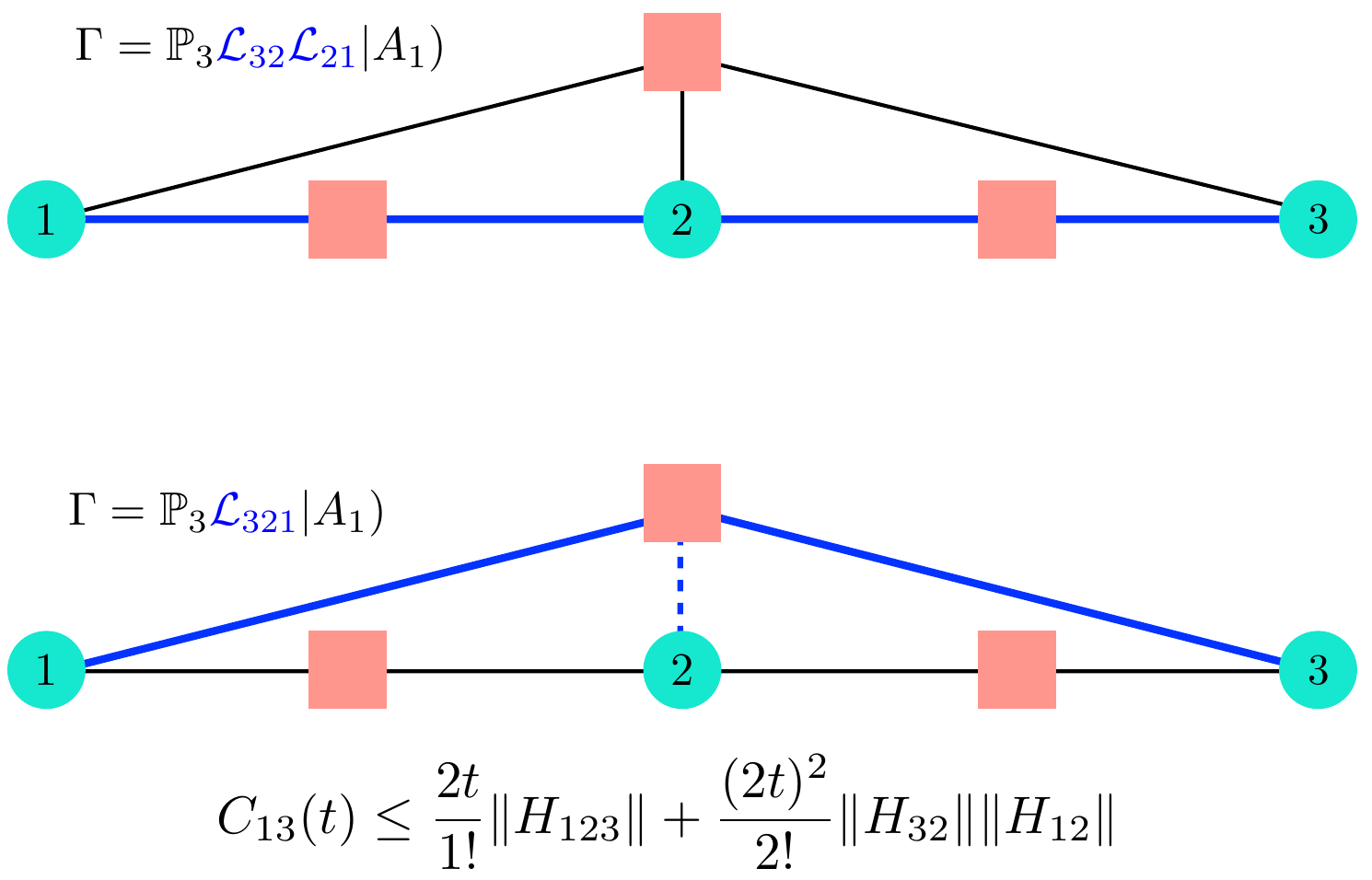}
\caption{A simple application of Theorem~\ref{theor3} to an elementary factor graph with 2 irreducible paths from 1 to 3.}
\label{fig:theor3ex}
\end{figure}

Theorem \ref{theor3} states that operator growth, as measured by $C_{ij}(t)$, is bounded by the ``weight" of all irreducible paths from $i$ to $j$.   This also immediately implies the following corollary:

\begin{corol}
Let $\lbrace i,j\rbrace\subseteq V$.  Define a symmetric real matrix $h\in \mathbb{R}^{N\times N}$ componentwise: \begin{equation}
h_{ij} :=  \left\lbrace \begin{array}{cc} \displaystyle \sum_{X \in F: \lbrace i,j\rbrace \subseteq X} \lVert H_X\rVert &\ i\ne j \\ 0 &\ i=j \end{array}\right.. \label{eq:hdef}
\end{equation}
Then \begin{equation}
{C}_{ij}(t), \hat{C}_{ij}(t) \le \exp[2|t| h]_{ij} ,  \label{eq:cor6}
\end{equation}
with $\exp[A]_{ij} = (\mathrm{e}^A)_{ij}$ the $ij$ component of the matrix exponential.   \label{corol6}
\end{corol}

\begin{proof}
It is a standard result in graph theory that \cite{chung} \begin{equation}
\exp[2|t| h]_{ij} = \sum_{n=0}^\infty \frac{(2|t|)^n}{n!} \left(h^n\right)_{ij} = \sum_{n=0}^\infty \frac{(2|t|)^n}{n!} \times \left(\text{total weight of all paths from $i$ to $j$ of length $n$} \right).
\end{equation}
 Since this sum includes reducible paths, and Theorem \ref{theor3} bounds $C_{ij}(t)$ by only terms in $(h^n)_{ij}$ over irreducible paths of length $n$, (\ref{eq:cor6}) is a weaker inequality than (\ref{eq:theor3}).
\end{proof}

\AC{While Corollary \ref{corol6} is probably the simplest and most conventional way of estimating the combinatorial sum in Theorem \ref{theor3}, we note that subsequent to the original posting of this work, the paper \cite{hazzard} proposed an alternative approach for trying to estimate similar geometric sums to those which arise in Theorem \ref{theor3} on $d$-dimensional lattice graphs.}

\subsection{Recovering the Lieb-Robinson Bound}
Define the undirected graph $\widetilde{G} = (V,\widetilde{E})$, where \begin{equation}
 \widetilde{E}  = \left\lbrace (i,j) \; : \; d_G(i,j) = 2 \right\rbrace.
\end{equation}
The distance between two vertices in $\widetilde{G}$ is \begin{equation}
\widetilde{d}_{\widetilde{G}}(i,j) = \frac{1}{2}d(i,j),  \label{eq:tildeGdist}
\end{equation}
once the factor vertices of $G$ have been removed.   Note that, by definition, $(h^n)_{ij} = 0$ if $\widetilde{d}(i,j) < n$.    We will usually suppress the subscript on $\widetilde{d}_{\widetilde{G}}$.

The Lieb-Robinson bound, with explicitly computed coefficients following \cite{hastings}, uses the triangle inequality and submultiplicativity of the operator norm to expand (\ref{eq:hatCijdef}).  This leads to a much weaker bound than (\ref{eq:cor6}): \cite{lucas1805} \begin{equation}
\hat{C}_{ij}(t) \le \exp[2|t|\widetilde{h}]_{ij} \label{eq:Cijwidetildeh}
\end{equation}
where \begin{equation}
\widetilde{h}_{ij} := \left\lbrace \begin{array}{ll} h_{ij} &\ i\ne j \\ \displaystyle \sum_k h_{ik} &\ i=j \end{array}\right.
\end{equation} 
is a positive semi-definite matrix.  (\ref{eq:Cijwidetildeh}) was proven in \cite{lucas1805}, where a bound on $\widehat{C}_{ij}$ was derived in terms of the exponential of a matrix involving the factor graph $G$'s adjacency matrix.  This result is much weaker than Theorem \ref{theor3}.

We now present the Lieb-Robinson bound:
\begin{corolNB}[\textsf{\textbf{Lieb-Robinson Bound}}]
Let $\lbrace i,j\rbrace \subseteq V$ with $i\ne j$, let $\alpha \in \mathbb{R}$ obey $\alpha > 1$, and let $\widetilde{h}_{\mathrm{max}}$ be the maximal eigenvalue of $\widetilde{h}_{ij}$.  Then $\hat{C}_{ij}(t) \le C_{ij}^{\alpha}(t)$, where
\begin{equation}
C_{ij}^{\alpha}(t) := \left(\mathrm{e}^{2\alpha \widetilde{h}_{\mathrm{max}}t}-1\right) \mathrm{e}^{-\widetilde{d}(i,j)\log\alpha} .  \label{eq:LRbound}
\end{equation}\label{corolLR}
\end{corolNB}

\begin{proof}
The proof is straightforward and is found in \cite{lucas1805}.    Using (\ref{eq:tildeGdist}) and component-wise positivity $h_{ij}\ge 0$:  \begin{align}
\exp[2|t| \widetilde{h}]_{ij} &=  \sum_{n=0}^\infty \frac{(2|t|)^n}{n!} \left(\widetilde{h}^n\right)_{ij} =  \sum_{n=\widetilde{d}(i,j)}^\infty \frac{(2|t|)^n}{n!} \left(\widetilde{h}^n\right)_{ij} \le \sum_{n=1}^\infty \alpha^{n-\widetilde{d}(i,j)}\frac{(2|t|)^n}{n!} \left(\widetilde{h}^n\right)_{ij} \notag \\
& < \sum_{n=1}^\infty \alpha^{n-\widetilde{d}(i,j)}\frac{(2|t|)^n}{n!} \widetilde{h}_{\mathrm{max}}^n .
\end{align}
Summing the series we obtain (\ref{eq:LRbound}).
\end{proof}

The Lieb-Robinson bound is usually interpreted as the statement that operators cannot grow faster than ballistically.    The region in which the operator is supported is said to expand with a Lieb-Robinson velocity: \cite{liebrobinson} \begin{equation}
\lim_{t,N\rightarrow \infty} \sup_{i,j:\widetilde{d}(i,j)=ut} \hat{C}_{ij}(t) = 0, \;\;\; \text{if } u> v_{\mathrm{LR}},  \label{eq:vLRdef}
\end{equation} where from (\ref{eq:LRbound}),
\begin{equation}
v_{\mathrm{LR}} := \inf_\alpha \frac{2\widetilde{h}_{\mathrm{max}}\alpha}{\log \alpha} = 2\mathrm{e}\widetilde{h}_{\mathrm{max}}.  \label{eq:vLReval}
\end{equation}
Note that the $N\rightarrow \infty$ limit must be taken so that two vertices can be found arbitrarily far apart.

In fact, the Lieb-Robinson velocity defined in (\ref{eq:vLRdef}) is not a sharp bound.  
\begin{propNB}
\begin{equation}
\lim_{t,N\rightarrow \infty} \sup_{i,j:\widetilde{d}(i,j)=ut} \hat{C}_{ij}(t) = 0, \;\;\; \text{if } u> \frac{v_{\mathrm{LR}}}{2}.  \label{eq:prop8eq}
\end{equation}\label{prop8}
\end{propNB}
\begin{proof}
 Since the maximal eigenvalue of a symmetric matrix is obtained by a variational principle: \begin{align}
\widetilde{h}_{\mathrm{max}} &= \sup_{\phi \in \mathbb{R}^{|V|}} \dfrac{\displaystyle \sum_{i,j \in V} \phi_i \widetilde{h}_{ij}\phi_j}{\displaystyle \sum_{i\in V} \phi_i^2} = \sup_{\phi \in \mathbb{R}^{|V|}} \dfrac{\displaystyle \sum_{i,j\in V}  h_{ij}(\phi_i+\phi_j)^2}{\displaystyle 2 \sum_{i\in V} \phi_i^2} =  \sup_{\phi \in \mathbb{R}^{|V|}} \left[ \dfrac{\displaystyle \sum_{i,j\in V}  h_{ij}(\phi_i-\phi_j)^2}{\displaystyle 2\sum_{i\in V} \phi_i^2}  + 2 \dfrac{\displaystyle \sum_{i,j \in V} \phi_i h_{ij}\phi_j}{\displaystyle\sum_{i\in V} \phi_i^2}\right] \notag \\
&\ge 2h_{\mathrm{max}}
\end{align}
where $h_{\mathrm{max}}$ is the maximal eigenvalue of $h_{ij}$.    Now repeat the proof of Corollary \ref{corolLR}, but replace $\widetilde{h}_{ij}$ with $h_{ij}$.  Generalizing (\ref{eq:vLReval}) we obtain (\ref{eq:prop8eq}).
\end{proof}

The ``Lieb-Robinson" bound of \cite{poulin} also is developed by thinking of operators as rotating vectors.  While it is weaker than (\ref{eq:cor6}), it is not quite as weak as (\ref{eq:Cijwidetildeh}).  And while it has been appreciated in the literature that the Lieb-Robinson bound is not tight, we are not aware of any prior literature with formal bounds as strong as Theorem~\ref{theor3} and Proposition~\ref{prop8}.

\subsection{Spin Chains}
In general, the bound (\ref{eq:theor3}) is  stronger than (\ref{eq:cor6}), which is in turn much stronger than (\ref{eq:LRbound}).  Let us consider the canonical example of a ``spin chain" Hamiltonian whose factor graph $G$ is the one-dimensional lattice with nearest neighbor interactions and open boundary conditions:
\begin{equation}
G_{\mathrm{1d,nn}} = (\mathbb{Z}_N,\mathbb{Z}_{N-1}, \lbrace (m,n) : m-n = 0 \text{ or } 1\rbrace ). 
\end{equation}
Our notation is as follows:  elements of the factor set, denoted as $n\in \mathbb{Z}_{N-1}$, are interpreted as $\lbrace n,n+1\rbrace$ for $0\le n < N-1$.  We suppose that for any factor $n\in \mathbb{Z}_{N-1}$, $\lVert H_n \rVert = h$.   

We now recite each bound in turn.   Let $\lbrace j_1,j_2\rbrace \subseteq \mathbb{Z}_N$ be vertices, and $t>0$ be time.   (\ref{eq:theor3}) is the strongest bound, and implies that \begin{equation}
\hat{C}_{ij}(t) \le \frac{(2ht)^{|j_1-j_2|}}{|j_1-j_2|!}.  \label{eq:LRfig1}
\end{equation}
(\ref{eq:cor6}) is most elegant to evaluate in the limit $N\rightarrow \infty$ with $\frac{N}{2}-i$ and $\frac{N}{2}-j$ held fixed.   This simply allows us to neglect boundary conditions.   We evaluate the matrix exponential in the plane wave eigenbasis using standard manipulations, and obtain  \begin{equation}
\hat{C}_{ij}(t) \le \int\limits_{-\mpi}^{\mpi} \frac{\mathrm{d}k}{2\mpi} \; \mathrm{e}^{\mathrm{i}k(j_1-j_2)} \mathrm{e}^{2ht\cdot 2\cos k} = \mathrm{I}_{|j_1-j_2|}(4ht),  \label{eq:LRfig2}
\end{equation}
where $\mathrm{I}_n$ is the $n^{\mathrm{th}}$ modified Bessel function.   Finally,   it is straightforward to see that the maximal eigenvalue of $h_{ij}$ is upper bounded by $\widetilde{h}_{\mathrm{max}} \le 4h$.  Thus for any $\alpha>1$, the Lieb-Robinson bound (\ref{eq:LRbound}) is \begin{equation}
\hat{C}_{ij}(t) \le \left(\mathrm{e}^{8h\alpha t} - 1\right)\mathrm{e}^{-|j_1-j_2|\log \alpha}.  \label{eq:LRfig3}
\end{equation}

\begin{figure}[t]
\centering
\includegraphics{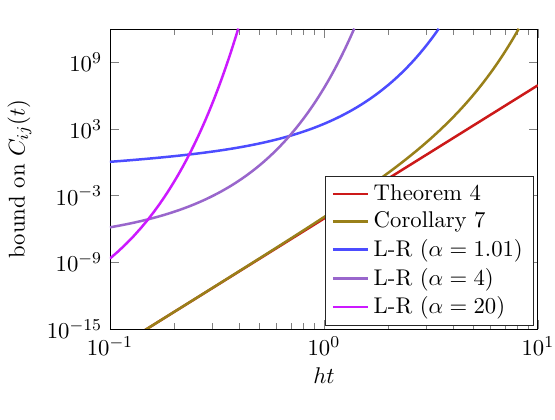}
\caption{A comparison of (\ref{eq:LRfig1}) (Theorem~\ref{theor3}), (\ref{eq:LRfig2}) (Corollary~\ref{corol6}), and (\ref{eq:LRfig3}) (Corollary \ref{corolLR}) for various $\alpha$.  We set $|j_1-j_2|=12$ for illustrative purposes but results are similar for any $|i-j|$.}
\label{fig:LR}
\end{figure}

Figure \ref{fig:LR} shows a numerical comparison of Theorem \ref{theor3}, Corollary \ref{corol6} and the Lieb-Robinson bound.  We observe that at early times, Theorem \ref{theor3} and Corollary \ref{corol6} are identical.  This is always the case at early times and is an elementary observation that the shortest path(s) from $i$ to $j$ are irreducible.   At \emph{all times}, the Lieb-Robinson bound is lousy and separated by a large constant factor from Theorem \ref{theor3}.   In fact, the Lieb-Robinson velocity $v_{\mathrm{LR}} = 8\mathrm{e}h$, whereas the true velocity of operator spreading is bounded by (\ref{eq:LRfig1}) as $v_{\mathrm{op}} \le 2\mathrm{e}h$.   The discrepancy of 4 comes from the following:  one factor of 2 arises from replacing $h_{ij}$ with $\widetilde{h}_{ij}$ (due to invoking the triangle inequality), and one factor of 2 arises from trying to bound $(h^n)_{ij}$ universally, instead of using its component-wise elements.  

\subsection{On the Tightness of Our Bounds}
\AC{
We can construct a solvable model which demonstrates that our bound on the velocity of commutator growth in spin chains approaches optimal (our bound is larger than the true value \textit{at late times} by around $35\%$). Consider the Hamiltonian
\begin{align}
    H = \sum_{i\in \mathbb{Z}} X_i Y_i,
\end{align}
defined on an infinite chain for simplicity (though it can easily be truncated if desired).   On each $i\in\mathbb{Z}$ we place a qubit ($d_i=2$). Let $X_i, Y_i,Z_i$ denote the three Pauli matrices on qubit $i$. Choose the initial operator \begin{equation}
    \mathcal{O} := \cdots Z_{-2}Z_{-1}\cdot Y_0.
\end{equation}
with $Z$s extending towards negative infinity.   A straightforward calculation reveals that \begin{equation}
    |\mathcal{O}(t)) = \sum_{j\in \mathbb{Z}} c_j(t)|j),
\end{equation}
where
\begin{align}
    |j) := \cdots Z\cdots Z_{j-1} Y_{j}
\end{align}
and \begin{equation}
    \frac{\mathrm{d}c_j}{\mathrm{d}t} = 2h\left(c_{j-1}(t) - c_{j+1}(t)\right).
\end{equation}
Hence the dynamics of this operator are equivalent to a one-dimensional quantum walk.  The explicit late time solution of this quantum walk was calculated in \cite{Konno_2005}; the wave front propagates with velocity $v = 4h$, which is very close to our bounds $v_{\mathrm{op}} \le 2\mathrm{e}h \approx 5.44h$. We note that the velocity of $4h$ can be found exactly using more sophisticated ``many body quantum walk" methods to bound OTOCs, which were developed in \cite{Lucas_2020} subsequently to the original posting of this work.}

\subsection{Proof of a Fast Scrambling Conjecture in Some Regular Systems}
Fix $\delta \in [0,1]$. We define the \emph{operator scrambling time} \begin{equation}
t_{\mathrm{s}}^\delta = \inf \lbrace t\in \mathbb{R}^+ : \text{for any }i,j\in V, \; C_{ij}(s) \ge \delta \text{ for some } s\in [0,t] \rbrace. \label{eq:definetsdelta}
\end{equation}
Our goal is to ask whether $t_{\mathrm{s}}^\delta = \mathrm{\Omega}(\log N)$, as suggested by the fast scrambling conjecture (\ref{eq:susskind}).  This section is largely a review of \cite{lucas1805, lashkari}; we summarize the results here both to improve slightly on bounds on $C_{ij}(t)$,  and more importantly to contrast their results with the developments in Section \ref{sec:random}.

First, we need to fix a convention for the scaling of $\lVert H_X\rVert$.  After all, we may freely rescale $H_X \rightarrow \alpha H_X$ and $t\rightarrow \alpha^{-1}t$ for $\alpha \in (0,\infty)$ without changing any physical predictions.   We partially fix the scale of $H_X$ as follows.  Let $\lbrace H^{(N)} \rbrace$ be a sequence of Hamiltonians on increasingly larger Hilbert spaces made up of $N$ degrees of freedom. \AC{ We will call the Hamiltonians belong to this sequence \emph{strongly extensive} if }\begin{equation}
0< \lim_{N\rightarrow \infty} \frac{(H^{(N)}|H^{(N)}) }{N} = \sum_{X\in F} \frac{(H_X|H_X)}{N} < \infty .  \label{eq:strongextensivity}
\end{equation}
Strong extensivity is different than the physicists' usual definition of an extensive Hamiltonian as one where for $\beta \in (0,\infty)$, 
\begin{equation}
-\infty< \lim_{N\rightarrow \infty} \frac{\log \left(\mathrm{tr}\left(\mathrm{e}^{-\beta H^{(N)}}\right)\right) }{N} < \infty .
\end{equation}
This latter definition is more physically motivated by the theory of statistical physics, but our formal definition of strong extensivity is significantly easier to work with.   We are not aware of any non-integrable model which is extensive but not strongly extensive.  

We define a factor graph to be \emph{regular} if for any $i,j\in V$,  $|\partial i| = |\partial j|$, and for any $X,Y\in F$, $|\partial X| = |\partial Y|$.    We call a regular factor graph $q$-local with \emph{degree} $k$ if $|\partial i| = k$ and $|\partial X|=q$.
\begin{prop}
Let $H$ be a strongly extensive Hamiltonian on a $q$-local, degree $k$ regular factor graph, with $k=\mathrm{O}(N^0)$ and $q=\mathrm{O}(N^0)$.  If $\max(\lVert H_X\rVert)/\min(\lVert H_X\rVert) < \infty$,  then  $t_{\mathrm{s}}^\delta = \mathrm{\Omega}(\log N)$ for $\delta \in (0,1)$.
\end{prop}
\begin{proof}
This result was essentially derived already in \cite{lucas1805, lashkari}.  Since our proposition is slightly more precise, and the result is instructive to derive, we repeat a proof here. 

First, pick any node $i\in V$.   Note that $|\lbrace j\in V : d(i,j)=2m\rbrace| \le k(k-1)^{m-1}(q-1)^m$.  Hence, there exists a vertex $j\in V$ such that $d(i,j) = \mathrm{\Omega}(\log N)$.

Next, observe that \begin{equation}
q|F|=kN \label{eq:regnumfactors}
\end{equation}
determines the total number of factors in any regular factor graph.  Hence if $k=\mathrm{O}(N^0)$ and $q=\mathrm{O}(N^0)$, $|F|=\mathrm{O}(N)$.    From (\ref{eq:strongextensivity}) and the relative boundedness of $\lVert H_X\rVert$, we see that $\lVert H_X\rVert  = \mathrm{O}(N^0)$.   

Now we apply Corollary \ref{corolLR} with $\alpha =\mathrm{e}$ to the vertices $i$ and $j$ from the previous paragraphs: \begin{equation}
\log {\delta} \le \log {\hat{C}_{ij}(t)} < 2\mathrm{e}\widetilde{h}_{\mathrm{max}} |t| - \frac{d(i,j)}{2}.
\end{equation}
Finally we must prove that $\widetilde{h}_{\mathrm{max}} = \mathrm{O}(N^0)$, which follows from the variational principle \begin{align}
\widetilde{h}_{\mathrm{max}} = \sup_{\phi \in \mathbb{R}^{|V|}} \dfrac{\displaystyle \sum_{i,j \in V} h_{ij}(\phi_i + \phi_j)^2}{\displaystyle \sum_{i \in V} \phi_i^2} \le  \sup_{\phi \in \mathbb{R}^{|V|}} \dfrac{2\displaystyle \sum_{i,j \in V} h_{ij}(\phi_i^2+ \phi_j^2)}{\displaystyle \sum_{i \in V} \phi_i^2} \le 2 k \max_{X\in F} \lVert H_X\rVert = \mathrm{O}(N^0).
\end{align}
This completes the proof.
\end{proof} 

Another system of interest is what we will call a \emph{$q$-local transverse field model}\footnote{The name is inspired by the transverse field Ising model, which is used in quantum annealing and optimization \cite{nishimori, farhi}.} (TFM) on any $q$-local factor graph $G$.  We define a TFM as any quantum system whose Hamiltonian may be expressed as \begin{equation}
H := H^{\mathrm{q}} + H^{\mathrm{c}} = \sum_{i\in V} H^{\mathrm{q}}_i + \sum_{X\in F : |X|>1} H^{\mathrm{c}}_X,
\end{equation}  
such that (\emph{1}) $\lVert H^{\mathrm{q}}_i\rVert = \mathrm{O}(N^0)$ for any $i \in V$,  (\emph{2}) for any $X,Y\in F$ with $|X|,|Y|>1$, $[H^{\mathrm{c}}_X, H^{\mathrm{c}}_Y] = 0$, and (\emph{3}) $H_X$ is scaled such that \begin{equation}
\sum_{X\in F:|X|>1} \lVert H^{\mathrm{c}}_X\rVert = \mathrm{O}(N).  \label{eq:Hcconstraint}
\end{equation}
The ``q" and ``c" superscripts refer to the ``quantum" and ``classical" parts of the Hamiltonian respectively.  The ``quantum" part of $H$ ensures it is always strongly extensive.  Because the classical parts of the Hamiltonian all commute with one another, they will only be relevant for the thermodynamic properties of the system when (\ref{eq:Hcconstraint}) holds.  

\begin{prop}
Let $H$ be a TFM, such that $H^{\mathrm{c}}$ is $q$-local on a regular factor graph of degree $k$, and \begin{equation}
\lVert H^{\mathrm{c}}_X \rVert \le \frac{N}{q|F|} a.  \label{eq:prop10a}
\end{equation}
 Then \begin{equation}
 \sum_{i,j=1}^N \hat{C}_{ij}(t) \le \frac{1}{N} \mathrm{e}^{2a|t|} \label{eq:prop10bound}
 \end{equation}and $t_{\mathrm{s}}^\delta = \mathrm{\Omega}(\log N)$ for any $\delta\in(0,1)$. \label{propFS3}
\end{prop}
\begin{proof}
This proposition also comes from \cite{lucas1805}, albeit with a weaker bound than (\ref{eq:prop10bound}).   We repeat the short proof for clarity.  Let $u_i = 1$ for each $i\in V$ denote a vector in $\mathbb{R}^{|V|}$.  Then \begin{equation}
\sum_{j\in V}h_{ij} u_j = \widetilde{h}_{ii} \le a = au_i.  \label{eq:prop10ui}
\end{equation} 
The inequality in the middle follows from (\ref{eq:regnumfactors}) and (\ref{eq:prop10a}).  Now using Corollary \ref{corol6}: \begin{equation}
\sum_{i,j=1}^N \hat{C}_{ij}(t) \le \sum_{i,j=1}^N u_i \left(\mathrm{e}^{2|t|h}\right)_{ij} u_j \le \frac{1}{N}\mathrm{e}^{2at}.
\end{equation}
The second inequality follows from repeated application of (\ref{eq:prop10ui}).   The full proposition follows from noting that we bounded $\hat{C}_{ij}(t) $ by a monotonically decreasing function for all $i$ and $j$; hence \begin{equation}
\inf_{s\in [0,t]} \min_{i,j} \hat{C}_{ij}(s) \le \frac{1}{N}\mathrm{e}^{2at}.
\end{equation}
This completes the proof.
\end{proof}

One might ask whether or not Corollary \ref{corol6} is a sharp bound at early times, as it was in the spin chain.   We expect that it often is on regular graphs.  As another example, consider a 2-local TFM on the complete graph, whose factor graph is \begin{equation}
G = (\mathbb{Z}_N, \lbrace \lbrace j_1,j_2\rbrace : j_{1,2}\in \mathbb{Z}_n, j_1<j_2\rbrace, \lbrace (j_1,\lbrace j_1,j_2\rbrace):j_{1,2}\in\mathbb{Z}_N\rbrace).
\end{equation}
Suppose that (\ref{eq:prop10a}) holds.  Then Theorem \ref{theor3} gives that \begin{align}
\hat{C}_{ij}(t) &\le \frac{2at}{N-1} + \sum_{m=2}^{N-2} \frac{(2at)^m}{(N-1)^{m-1}m!} \prod_{k=2}^m (N-k) = \frac{1}{N-1}\sum_{m=1}^{N-2} \left(\begin{array}{c} N-1 \\ m \end{array}\right) \left(\frac{2at}{N-1}\right)^m \notag \\
&\le \frac{1}{N-1} \left( \left(1+\frac{2at}{N-1}\right)^{N-1}-1\right) = \frac{1}{N-1} \left( \exp\left(2at - 2\frac{(at)^2}{N-1} + \mathrm{O}\left(\frac{t^3}{N^2}\right)\right)-1\right).  \label{eq:theor3prop10}
\end{align}
Keeping in mind that a fraction $\frac{N-1}{N}$ of the pairs $\hat{C}_{ij}$ have $i\ne j$, we conclude that (\ref{eq:prop10bound}) and (\ref{eq:theor3prop10}) are asympotically identical when $t=\mathrm{O}(\log N)$, as $N\rightarrow \infty$.   This is easy to understand from the graph theory perspective: almost surely, a randomly chosen path of length $\ell = \mathrm{o}(\sqrt{N})$ is irreducible.

There are factor graphs which are not regular, however, for which Theorem \ref{theor3} and Corollary \ref{corol6} give very different results.  An example corresponds to a 2-local Hamiltonian on the star graph \cite{lucas1805, lucas1903}, whose factor graph is \begin{equation}
G = (\mathbb{Z}_N, \mathbb{Z}_{N-1}, \lbrace (i,j):i\in\mathbb{Z}_N, j \in \mathbb{Z}_{N-1}, \; i\in \lbrace j,N-1\rbrace ).
\end{equation}
Assuming $\lVert H_X\rVert = a$ for factor $X\in G$, Theorem \ref{theor3} gives \begin{equation}
\hat{C}_{01}(t) \le 2a^2t^2
\end{equation}
while Corollary \ref{corol6} gives 
\begin{equation}
\hat{C}_{01}(t) \le \frac{\cosh(2h|t|\sqrt{N-1}) - 1}{N-1}.
\end{equation}
Hence on factor graphs with very heterogeneous degree distributions, Theorem \ref{theor3} is much stronger.  However, we also emphasize that $t_{\mathrm{s}}^* \ge \mathrm{O}(1)$ is inconsistent with the fast scrambling conjecture, as postulated for operator growth.  An explicit 2-local model on the star graph where $t_{\mathrm{s}}^* = \mathrm{O}(1)$ is found is contained in \cite{lucas1903}.

Since we have reduced the problem of bounding operator growth to a combinatorial problem in graph theory, it should be possible to use the known properties of random graph ensembles to make detailed predictions about the behavior of $\hat{C}_{ij}(t)$ on random graphs.  This interesting problem should be pursued elsewhere.

\section{Random Systems}\label{sec:random}
\subsection{Random Hamiltonian Ensembles}
In this section, we shift our focus to Hamiltonians drawn from a certain kind of random ensemble, which we now introduce.   Let $G=(V,F,E)$ be a factor graph with $|V|=N$.   In this section, we allow the factor set $F$ to contain the same element of $\mathbb{Z}_2^V$ \emph{multiple times}, for reasons we will soon explain.  For example, if $V=\lbrace 1,2,3\rbrace$, we may choose $F=\lbrace \lbrace 1,2\rbrace_1, \lbrace 1,2\rbrace_2, \lbrace 1,3\rbrace\rbrace$;  the subscript on $X\in F$ counts this multiplicity, and $\lbrace 1,2\rbrace_1 \ne \lbrace 1,2\rbrace_2$.   

We define the probability space $(\mathbb{R}^{|F|}, \sigma(\mathbb{R}^{|F|}), \mu_F)$ where $\sigma(\cdots)$ denotes a canonical Borel algebra.   Points $J\in\mathbb{R}^{|F|}$  can be written as $J:=(J_1,\ldots, J_{|F|})$; we call $J_X \in \mathbb{R}$ a \emph{random coupling}.  There is one random coupling for each $X\in F$.     In this paper we will exclusively focus on measures $\mu_F$ where $J_X$ and $J_Y$ are independent if $X\ne Y$:   \begin{equation}
\mathrm{d}\mu_F = \prod_{X\in F} \mathrm{d}\mu_X,
\end{equation}
where $(\mathbb{R},\sigma(\mathbb{R}),  \mu_X)$ is a probability space for each $X\in F$, and where for any $X\in F$,  $J_X$ is \emph{zero-mean}:  \begin{equation}
\int \mathrm{d}\mu_X \; J_X = 0.  \label{eq:zeromean}
\end{equation}

Consider the linear Hamiltonian map $H:\mathbb{R}^{|F|} \rightarrow \mathcal{B}$ defined by \begin{equation}
H(J) := \sum_{X\in F} J_X H_X,
\end{equation}
where $H_X\in \mathcal{B}_X$ is an $X$-local Hermitian operator with $\lVert H_X\rVert > 0 $ and \begin{equation}
(H_X|H_{X^\prime}) = \mathbb{I}(X=X^\prime).  \label{eq:41orth}
\end{equation}
The orthogonality of $H_X$ means that $H$ is an injection.   In our 3-site example above, we may choose $H_{\lbrace 1,2\rbrace_1} = X_1 X_2$, $H_{\lbrace 1,2\rbrace_1} = Y_1Y_2$ and $H_{\lbrace 1,3\rbrace} = X_3 Z_3$, where $X_i, Y_i,Z_i$ denote Pauli matrices for $i\in V$.  Thus, since there can be multiple orthogonal $X$-local operators, we allowed multiplicity in the factor set.        

We can now formally define our \emph{random Hamiltonian ensemble} as a probability space $(\mathcal{B},\sigma(\mathcal{B}), \mu_{\mathcal{B}})$, where for any subset $Q\subseteq \mathcal{B}$ by \begin{equation}
\mu_{\mathcal{B}}(Q) = \mu_F\left(H^{-1}(Q)\right).
\end{equation}
We define averages over the ensembles $\mu_{\mathcal{B}}$ and $\mu_F$, related by the above formula, with the notation \begin{equation}
\mathbb{E}[f] = \int \mathrm{d}\mu_F \; f.
\end{equation}
We say that a random Hamiltonian ensemble is \emph{simple} if (\emph{1}), for any $X\in F$, $\mathbb{E}[J_X]=0$, and (\emph{2}) for any $\lbrace X_1,X_2\rbrace \subseteq F$ with $X_1\ne X_2$, $J_{X_1}$ and $J_{X_2}$ are independent random variables.

The goal of this section is to bound \begin{equation}
\mathbb{E}\left[C_{ij}(t)^2\right] =  \frac{\mathbb{E}[(A_j(t)|\mathbb{P}_i |A_j(t))]}{(A_j|A_j)} =  \frac{\mathbb{E}\left[ (A_j| \mathrm{e}^{-\mathcal{L}t} \mathbb{P}_i \mathrm{e}^{\mathcal{L}t}|A_j)\right]}{(A_j|A_j)} .  \label{eq:boundavg}
\end{equation}
There are two important (and related) complications of this object relative to $\lVert \mathbb{P}_j \mathrm{e}^{\mathcal{L}t}|A_i)\rVert$, which we studied to prove Theorem \ref{theor3}.   (\emph{1}) There are now two factors of $\mathrm{e}^{\mathcal{L}t}$.   (\emph{2}) Each random coupling $J_X$ must show up at least twice (but could show up twice in the same $\mathrm{e}^{\mathcal{L}t}$): using (\ref{eq:zeromean}), $\mathbb{E}[J_X f ]= 0 $ for arbitrary functions $f$ which do not depend on $J_X$.  

Our motivation for developing separate bounds for (\ref{eq:boundavg}) is that many systems of interest in physics are drawn from random Hamiltonian ensembles where  Theorem \ref{theor3} is far too weak to prove the fast scrambling conjecture.  For example, consider the random $\mathrm{SU}(2)$ Heisenberg model on the complete graph, with Hamiltonian \cite{arrachea} \begin{equation}
H = \sum_{i<j} \sum_{\alpha=1}^3 J_{ij} X^\alpha_i X^\alpha_j \label{eq:su2heisenberg}
\end{equation}
where $X^\alpha_i \in \mathcal{B}_i$ ($\alpha\in\lbrace1,2,3\rbrace$) denote Pauli matrices, and 
\begin{equation}
\mathbb{E}[J_{ij}J_{i^\prime j^\prime}] = \frac{1}{N} \mathbb{I}(\lbrace i,j\rbrace =\lbrace i^\prime ,j^\prime\rbrace).
\end{equation}
This Hamiltonian is almost surely strongly extensive as $N\rightarrow \infty$, and should obey (\ref{eq:susskind}).    We will perturbatively prove that it does in Theorem \ref{theorFS}.  However, applying Corollary \ref{corol6} to (\ref{eq:su2heisenberg}), we obtain \begin{equation}
C_{ij}(t) \le \frac{1}{N} \mathrm{e}^{6N\max(|J_{ij}|)|t|}.
\end{equation} 
The factor of 6 is not tight, but more problematic is that $N\max(|J_{ij}|) \ge \sqrt{N}$.  Hence we cannot use Theorem \ref{theor3} to obtain (\ref{eq:susskind}).  It is necessary to  average over the random couplings to recover (\ref{eq:susskind}).

For the rest of this section, we will exclusively focus on simple random Hamiltonian ensembles.

\subsection{Causal Tree Pairs}\label{sec:causaltreepair}
In (\ref{eq:boundavg}), the numerator is the inner product of $(A_i|\mathrm{e}^{-\mathcal{L}t} \mathbb{P}_j$ and $\mathbb{P}_j \mathrm{e}^{\mathcal{L}t} |A_i)$.   Proposition \ref{propcausaltree} implies that the only sequences of $\mathcal{L}$ which contribute to either vector above are creeping.  
So with \begin{equation*}
(A_i|\mathcal{L}_{X_n}\cdots \mathcal{L}_{X_k}\mathbb{P}_j \mathcal{L}_{X_{k-1}}\cdots \mathcal{L}_{X_1} |A_i)
\end{equation*}
 in mind, we define the two sequences \begin{subequations}\label{eq:MLMR}\begin{align}
 \mathcal{M}_{\mathrm{R}} &:= ( i, X_1,\ldots, X_k,j,X_{k+1},\ldots, X_n), \\
 \mathcal{M}_{\mathrm{L}} &:= ( i, X_n,\ldots,X_{k+1},j,X_k,\ldots,X_1)
 \end{align}\end{subequations}
 to be the couplings read from right to left ($\mathcal{M}_{\mathrm{R}}$) and from left to right ($\mathcal{M}_{\mathrm{L}}$).  We have also now added the \emph{target} node $j\in V$ wherever the projector $\mathbb{P}_j$ appeared in the sequence, for reasons which will become clear later.   $\mathcal{M}_{\mathrm{L,R}}$ are reverses of each other, but it is useful to think of them as ``independent".   We define the set $\mathbb{M}^2_{ji}$ to be the set of all sequences of factors where (\emph{1}) each factor shows up at least twice, (\emph{2}) at least one factor contains $j$, and (\emph{3}) the sequence is creeping when read from left to right or right to left, as in (\ref{eq:MLMR}).  We will usually write $(\mathcal{M}_{\mathrm{L}},\mathcal{M}_{\mathrm{R}})\in\mathbb{M}^2_{ji}$ to emphasize the ``two sided" nature of the sequences.   The reason for defining $\mathbb{M}^2_{ji}$ is the following generalization of Proposition~\ref{propcausaltree}:
 
\begin{propNB}
If $i\ne j$, $\mathbb{E}\left[(A_i|\mathcal{L}_{X_n}\cdots \mathcal{L}_{X_k}\mathbb{P}_j \mathcal{L}_{X_{k-1}}\cdots \mathcal{L}_{X_1} |A_i) \right]\ne 0$ and (\ref{eq:MLMR}) defines $\mathcal{M}_{\mathrm{L,R}}$, then $(\mathcal{M}_{\mathrm{L}},\mathcal{M}_{\mathrm{R}}) \in \mathbb{M}^2_{ji}$.
\label{propMLMR}
\end{propNB}
 \begin{proof}
(\emph{1}) Since the average is non-vanishing, we conclude from (\ref{eq:zeromean}) that each factor in $(\mathcal{M}_{\mathrm{L}},\mathcal{M}_{\mathrm{R}})$ must show up at least twice.  (\emph{2}) Using (\ref{eq:PjLX}), there exists $X\in \mathcal{M}_{\mathrm{L,R}}$ with $j\in X$.  (\emph{3}) Proposition~\ref{propcreeping} implies that $\mathcal{M}_{\mathrm{L,R}}$ are \AC{both} creeping.  These are the three properties necessary for $(\mathcal{M}_{\mathrm{L}},\mathcal{M}_{\mathrm{R}}) \in \mathbb{M}^2_{ji}$.
 \end{proof}
 
The rest of this section mirrors Section \ref{sec:causaltrees}, introducing the framework necessary to prove Theorem~\ref{theor4}.  First, we define the \emph{causal tree pair} \begin{equation}
(T_{\mathrm{L}},T_{\mathrm{R}}) := (T(\mathcal{M}_{\mathrm{L}}), T(\mathcal{M}_{\mathrm{R}})) := \mathbb{T}(\mathcal{M}_{\mathrm{L}},\mathcal{M}_{\mathrm{R}}).
\end{equation}
 made out of the causal trees of $\mathcal{M}_{\mathrm{L,R}}$.  The last definition in the above equation is of the function $\mathbb{T}:\mathbb{M}^2_{ji} \rightarrow \mathcal{T}_{ji}\times\mathcal{T}_{ji}$.   The set of all causal tree pairs is thus given by the image of $\mathbb{T}$: namely, $\mathcal{T}^2_{ji}:= \mathbb{T}(\mathbb{M}^2_{ji})$. 
When $(T_{\mathrm{L}},T_{\mathrm{R}})\in\mathcal{T}^2_{ji}$, the causal trees $T_{\mathrm{L,R}}$ have the same factors: $T_{\mathrm{L}}\cap F = T_{\mathrm{R}}\cap F$.  When thought of as subgraphs of the full factor graph, $T_{\mathrm{L,R}}$ also contain the same nodes: $T_{\mathrm{L}}\cap V = T_{\mathrm{R}}\cap V$.   We call two causal trees with such properties \emph{compatible} and denote the relation as $T_{\mathrm{L}}\cong T_{\mathrm{R}}$.   While the node and factor sets of each causal tree in the pair $(T_{\mathrm{L}},T_{\mathrm{R}})$ are identical, in general $T_{\mathrm{L}} \ne T_{\mathrm{R}}$ because the two trees contain distinct edges.  

We define a \emph{causal subtree pair} $(S_{\mathrm{L}},S_{\mathrm{R}}) \subseteq (T_{\mathrm{L}},T_{\mathrm{R}})$ by $S_{\mathrm{L,R}} \subseteq T_{\mathrm{L,R}}$, if both $(S_{\mathrm{L}},S_{\mathrm{R}}), (T_{\mathrm{L}},T_{\mathrm{R}}) \in \mathcal{T}^2_{ji}.$  We emphasize that in this definition, indistinguishable causal trees are treated as the same object.
  We define an equivalence relation $\sim_{ji}$ on causal tree pairs:  $(T_{\mathrm{L}},T_{\mathrm{R}}) \sim_{ji} (S_{\mathrm{L}},S_{\mathrm{R}}) $ if they share a common causal subtree pair.  
  \begin{propNB}
  There is a unique \emph{irreducible causal tree pair} $(Q_{\mathrm{L}},Q_{\mathrm{R}})$ in each equivalence class of causal tree pairs $[(Q_{\mathrm{L}},Q_{\mathrm{R}})]_{ji}$, with the property that for any $(S_{\mathrm{L}},S_{\mathrm{R}}) \in [(T_{\mathrm{L}},T_{\mathrm{R}})]_{ji}$,  $(Q_{\mathrm{L}},Q_{\mathrm{R}}) \subseteq (S_{\mathrm{L}} , S_{\mathrm{R}})$.
  \end{propNB}
  \begin{proof}
  This follows constructively: given $(T_{\mathrm{L}},T_{\mathrm{R}}), (S_{\mathrm{L}},S_{\mathrm{R}}) \in [(T_{\mathrm{L}},T_{\mathrm{R}})]_{ji}$, by definition there exists a causal tree pair $(Q_{\mathrm{L}},Q_{\mathrm{R}}) \in \mathcal{T}^2_{ji}$ with $Q_{\mathrm{L,R}}\subseteq S_{\mathrm{L,R}} \cap T_{\mathrm{L,R}}$.  Since the equivalence class has a finite number of elements, this intersection can be repeated until the irreducible element, whose left/right causal trees are subtrees of the left/right causal trees of any other element, is found. 
  \end{proof}
  
  \begin{figure}[t]
  \centering
  \includegraphics[width=0.8\textwidth]{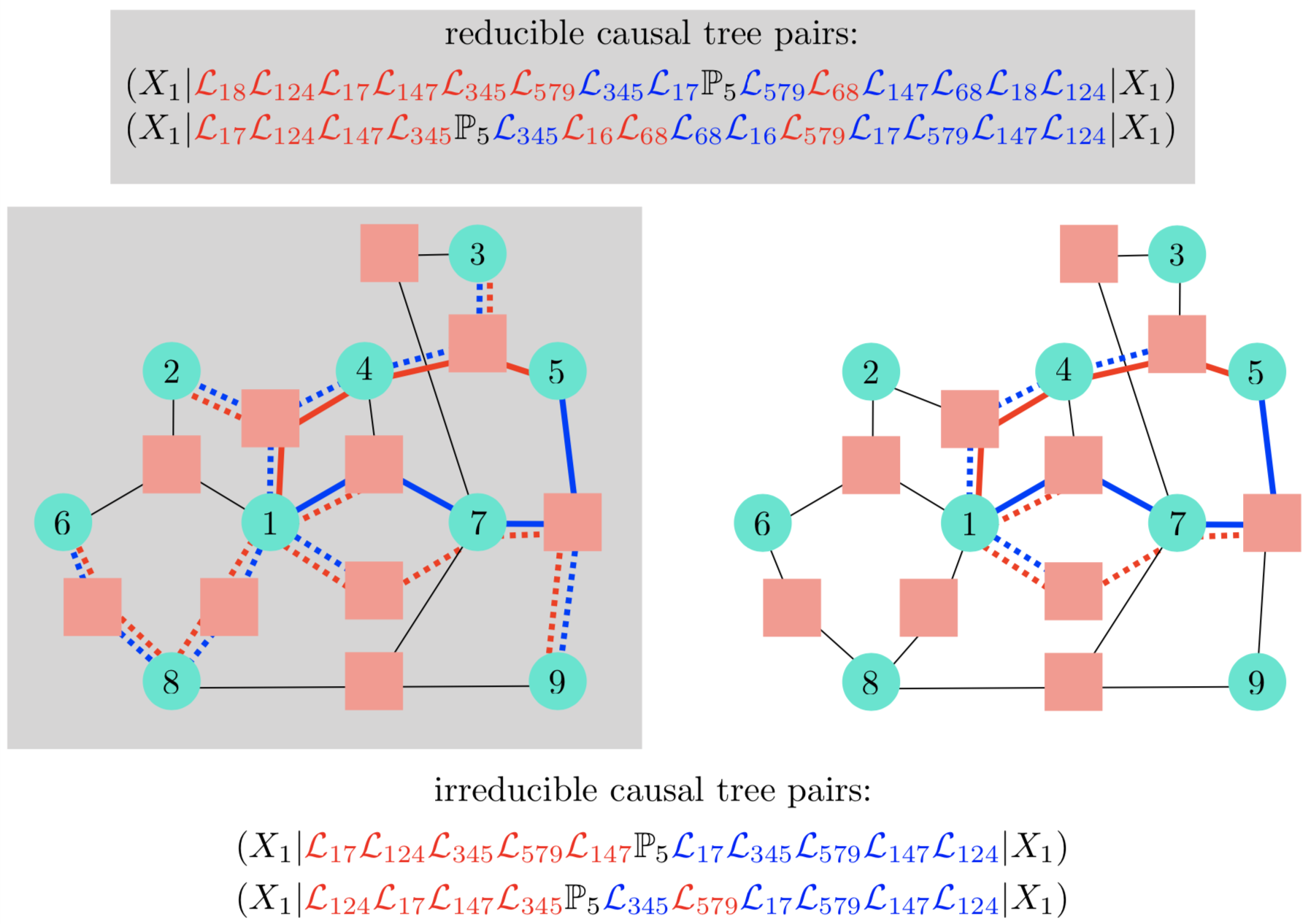}
  \caption{Reducible (left) and irreducible (right) causal tree pairs on a factor graph.   The left tree is shown in red, and the right tree is shown in blue.  Examples of reducible/irreducible sequences in the same equivalence class are given above/below, respectively.
  }
  \label{fig:causalgraph}
  \end{figure}

Figure \ref{fig:causalgraph} gives some intuition behind our definition of irreducible causal tree pairs.  Given the two sequences $\mathcal{M}_{\mathrm{L,R}}$ of Liouvillians, we are tempted to simply find the irreducible path between $i$ and $j$ of both sequences.   However, these irreducible paths may not share factors -- in this case, the average $\mathbb{E}[\cdots]$ implies that any unrepeated factors must show up a second time somewhere in each sequence.   Since each sequence must be creeping, we then look for the smallest possible connected subtrees of $T_{\mathrm{L,R}} :=  T(\mathcal{M}_{\mathrm{L,R}})$ which share factors (thus implying that every factor appears twice, and that $\mathbb{E}[\cdots]$ can be non-vanishing).  Looking only for causal subtrees is how we restrict our study to creeping sequences within the graph theoretic formalism.  As shown in Figure \ref{fig:causalgraph}, the requirement that $T_{\mathrm{L,R}}$ are each connected subtrees can force the causal graph to be larger than simply the union of two irreducible paths.    (We will revisit this issue in Section \ref{sec:sqrtlogN}).

We define the set of all irreducible causal pairs as $\mathcal{S}^2_{ji} = \mathcal{T}^2_{ji}/\sim_{ji}$.   Unlike in Section \ref{sec:causaltrees}, there will be many sequences $\mathcal{M}$ which lead to the same irreducible causal pair, and there is no canonical sequence.   As such, given $(Q_{\mathrm{L}},Q_{\mathrm{R}})\in \mathcal{S}^2_{ji}$, we define \begin{equation}
\Psi(Q_{\mathrm{L}},Q_{\mathrm{R}}) := \lbrace \mathcal{M} \in \mathbb{M}^2_{ji} : \mathbb{T}(\mathcal{M}) = (Q_{\mathrm{L}},Q_{\mathrm{R}})\text{ and } |\mathcal{M}| = 2|F\cap Q_{\mathrm{L}}|\rbrace 
\end{equation}
to be the set of all possible orderings of factors on the left and right causal trees, simultaneously.  Let $\psi\in\Psi(T_{\mathrm{L}},T_{\mathrm{R}})$, and let $\ell(\psi)=|\psi|$ be the number of factors in the sequence (which will always be an even number).   Let $\ell_{\mathrm{L,R}}(\psi)$ to be the number of factors between $i$ and $j$ in the sequences $\psi = (\mathcal{M}_{\mathrm{L}},\mathcal{M}_{\mathrm{R}})$ respectively.  Slightly abusing notation, we will also write  $\psi = ( i, X_1,X_2,\ldots, X_{\ell(\psi)})$ as the sequence of factor nodes (together with root $i$ and target $j$) in the right sequence $\mathcal{M}_{\mathrm{R}}$.   For any factor $X\in \psi$, let $1 \le \min_\psi(X)<\max_\psi(X) \le \ell(\psi)$ denote the first and second location of $X$ in the sequence $\mathcal{M}_{\mathrm{R}}$.  We define $\min_\psi(j)$ and $\max_\psi(j)$ to be the first and last factors which include the node $j$ (as read from right to left).   As in Section \ref{sec:causaltrees}, we now define a sequence of $V^\psi_k$ which encode nodes which cannot be reached in a growing operator without changing either the left or right causal tree: for $0\le k \le \ell(\psi)$,
\begin{align}
&V_k^\psi := \left\lbrace v\in V-i : v\notin \bigcup_{p=1}^k X^\psi_p\text{, and }v\in X^\psi_m \text{ for some }m=\min(X^\psi_m)>k+1\right\rbrace \cup \left\lbrace j, \text{ if }k<\min_\psi (j) \right\rbrace \notag \\
&\cup \left\lbrace v\in V-i : v\notin \bigcup_{p=k+1}^{\ell(\psi)} X^\psi_p\text{, and }v\in X^\psi_m \text{ for some }m=\max(X^\psi_m)< k\right\rbrace\cup \left\lbrace j, \text{ if }k\ge \max_\psi (j) \right\rbrace.
\end{align}
We also define the set \begin{equation}
Y^\psi_k := \lbrace X\in F : X\cap V^\psi_k \ne \emptyset\rbrace \cup \lbrace X\in \psi : \min_\psi(X) > k \rbrace \cup \lbrace X\in \psi : \max_\psi(X) \le k \rbrace 
\end{equation}
to be a set of forbidden factors in between $X^\psi_k$ and $X^\psi_{k+1}$:  the first set changes either the left or right causal tree; the second/third set prevents the right/left-most appearance of any factor from being modified, respectively.    An example of how to construct both $V^\psi_k$ and $Y^\psi_k$ is given in Figure \ref{fig:vpsi42}.

\begin{figure}[t]
\centering
\includegraphics[width=0.85\textwidth]{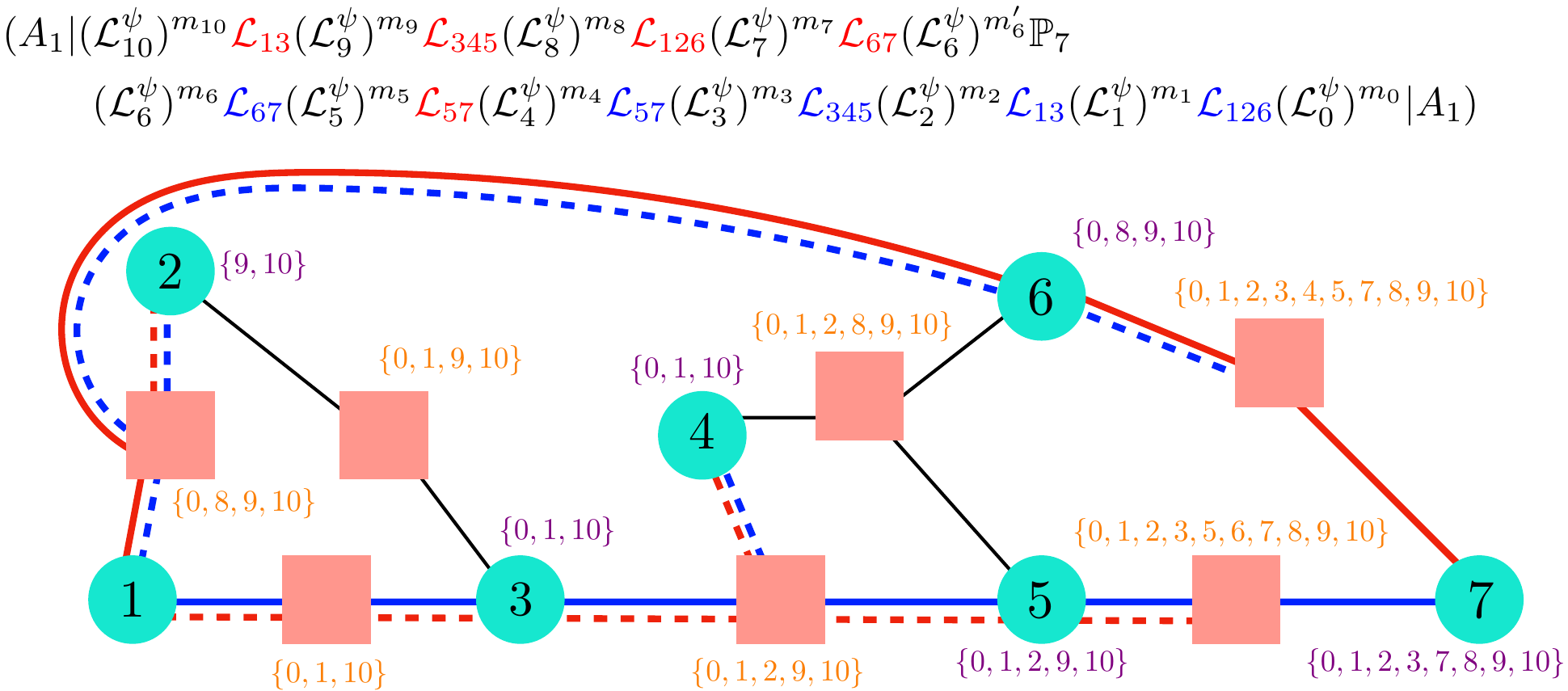}
\caption{An irreducible causal tree pair $(Q_{\mathrm{L}},Q_{\mathrm{R}})\in\mathcal{S}^2_{ji}$, which follows from the sequence depicted. $Q_{\mathrm{L}}$ is shown in red and $Q_{\mathrm{R}}$ is shown in blue.  Sets in purple next to each node $v\in V$ denote values of $k$ for which $v\in V^\psi_k$; sets in orange next to each factor $X\in F$ denote values of $k$ for which $X\in Y^\psi_k$.}
\label{fig:vpsi42}
\end{figure}

\subsection{Bounding the Norm Variance}
We now state our second main theorem:
\begin{theor}
Let $\mathbb{E}[\cdots]$ denote expectation value in the measure of a simple random Hamiltonian ensemble on $G=(V,F,E)$, with \begin{equation}
\mathbb{E}\left[J_X^2\right] :=  \mathcal{J}_X^2.  \label{eq:theor4J}
\end{equation}
Then \begin{align}
\mathbb{E}\left[C_{ij}(t)^2\right] \le  \sum_{(Q_{\mathrm{L}},Q_{\mathrm{R}}) \in \mathcal{S}^2_{ji}}  \sum_{\psi \in \Psi(Q_{\mathrm{L}},Q_{\mathrm{R}})}  \frac{1}{\ell_{\mathrm{L}}(\psi)!\ell_{\mathrm{R}}(\psi)!} \prod_{X\in Q_{\mathrm{L}}\cup Q_{\mathrm{R}}} (2\CJ_X t)^2. \label{eq:theor4}
\end{align}
 \label{randombound}\label{theor4}
\end{theor}

\begin{proof}
The method of proof follows closely Theorem \ref{theor3}.  (\emph{1}) We generalize Lemma \ref{lemma4} to the sum over sequences in $\mathbb{M}^2_{ji}$, using the expectation value and zero-mean property of $J_X$ to simplify things.  (\emph{2}) Next, we show that the disorder average can be used to exactly re-write  the sum over random couplings as a sum over sequences $\psi$ describing all possible irreducible causal tree pairs, along with orthogonal rotations.  We bound the resulting sum exactly as in Theorem \ref{theor3}.  

\emph{Step 1:}  With the classification developed in Section \ref{sec:causaltreepair}, we may write \begin{equation}
(A_i|\mathrm{e}^{-\mathcal{L}t}\mathbb{P}_j \mathrm{e}^{\mathcal{L}t}|A_i) = \sum_{(Q_{\mathrm{L}},Q_{\mathrm{R}})} \sum_{\substack{\mathcal{M}\in\mathbb{M}^2_{ji} :\\ \mathbb{T}(\mathcal{M})\in[(Q_{\mathrm{L}},Q_{\mathrm{R}})] }} \frac{(-t)^{n-r}t^r}{(n-r)!r!} (A_i|\mathcal{L}_{X_n}\cdots \mathcal{L}_{X_{r+1}}\mathbb{P}_j\mathcal{L}_{X_r} \cdots \mathcal{L}_{X_1}|A_i).
\end{equation}
Now we state the following generalization of Lemma \ref{lemma4}:

\begin{lma}
For $(Q_{\mathrm{L}},Q_{\mathrm{R}}) \in \mathcal{S}^2_{ji}$,
\begin{align}
&\mathbb{E}\left[\sum_{\mathcal{M} \in \mathbb{M}^2_{ji}: \mathbb{T}(\mathcal{M})\in [(Q_{\mathrm{L}},Q_{\mathrm{R}})]}  \frac{(-t)^{n-r}t^r}{(n-r)!r!} (A_i|\mathcal{L}_{X_n}\cdots \mathcal{L}_{X_{r+1}}\mathbb{P}_j \mathcal{L}_{X_r}\cdots \mathcal{L}_{X_1} |A_i) \right]\label{classsum} \notag \\
&=\mathbb{E}\left[ \sum_{\psi \in \Psi(Q_{\mathrm{L}},Q_{\mathrm{R}})}  \sum_{m_0,\ldots, m_\ell = 0}^\infty \frac{(-t)^{\ell-\ell_{\mathrm{R}}+ m_{\ell_{\mathrm{R}}}^\prime + \sum^\ell_{i=\ell_{\mathrm{R}}+1} m_i }}{ (\ell-\ell_{\mathrm{R}} + m_{\ell_{\mathrm{R}}}^\prime+ \sum^\ell_{i=\ell_{\mathrm{R}}+1} m_i )!} \frac{t^{\ell_{\mathrm{R}} + \sum^{\ell_{\mathrm{R}}}_{i=0} m_i }}{(\ell_{\mathrm{R}} + \sum^{\ell_{\mathrm{R}}}_{i=0} m_i )!}\times\right. \notag \\
& \left. (A_i|(\mathcal{L}^\psi_\ell)^{m_0} \mathcal{L}_{X_\ell^\psi} \cdots\mathcal{L}_{X_{1+\ell_{\mathrm{R}}}^\psi}  (\mathcal{L}^\psi_{\ell_{\mathrm{R}}})^{m^\prime_{\ell_{\mathrm{R}}}} \mathbb{P}_j(\mathcal{L}^\psi_{\ell_{\mathrm{R}}})^{m_{\ell_{\mathrm{R}}}}  \mathcal{L}_{X_{\ell_{\mathrm{R}}}^\psi} (\mathcal{L}^\psi_{\ell_{\mathrm{R}}-1})^{m_{\ell_{\mathrm{R}}-1}}\cdots(\mathcal{L}^\psi_1)^{m_1}  \mathcal{L}_{X_1^\psi}  (\mathcal{L}^\psi_0)^{m_0}  |A_i) \right]
\end{align}
where $\ell := \ell(\psi)$,  $\ell_{\mathrm{R}}:=\ell_{\mathrm{R}}(\psi)$ and \begin{equation}\label{neighbourhood}
\mathcal{L}^\psi_k := \mathcal{L} - \sum_{Y\in Y^\psi_k } \mathcal{L}_Y . 
\end{equation}
 \label{twoclass}
\end{lma}
\begin{proof}
The proof is analogous to Lemma \ref{lemma4}.  Every non-vanishing term on right hand side of (\ref{classsum}) is in the class $[(Q_{\mathrm{L}},Q_{\mathrm{R}})]$ with ordering $\psi$ by construction.\footnote{While there are terms in the sum on the right hand side where a coupling $X\notin Q_{\mathrm{L,R}}$ can show up a single time, these terms are killed by the average.} Conversely, for every $(\mathcal{M}_{\mathrm{L}},\mathcal{M}_{\mathrm{R}}) \in \mathbb{M}^2_{ji}$ with $[\mathbb{T}(\mathcal{M}_{\mathrm{L}},\mathcal{M}_{\mathrm{R}})] = [(Q_{\mathrm{L}},Q_{\mathrm{R}})]$, the sequence $(\mathcal{M}_{\mathrm{L}},\mathcal{M}_{\mathrm{R}})$ must be expressible as a term on the right hand side of (\ref{classsum}), as otherwise the sequence would belong to a different equivalence class in $\mathcal{S}^2_{ji}$.  The coefficients of the non-vanishing terms on each side of the sum are also identical, as they are only determined by the number of Liouvillians in between the projector $\mathbb{P}_j$ and the left/right vectors $(A_i|$ and $|A_i)$.\footnote{Note that there is an extra $(\mathcal{L}^\psi_{\ell_{\mathrm{R}}})^{m^\prime_{\ell_{\mathrm{R}}}}$ term on the left hand side of $\mathbb{P}_j$.  Its forbidden factors are $Y^\psi_{\ell_{\mathrm{R}}}$ are the same as $(\mathcal{L}^\psi_{\ell_{\mathrm{R}}})^{m_{\ell_{\mathrm{R}}}}$, which appears to the right of $\mathbb{P}_j$, because by definition $Y^\psi_{\ell_{\mathrm{R}}}$ only depends on the relative position of factors to each other, and not on the location of the projector $\mathbb{P}_j$.}   Hence there is a bijection between the non-vanishing terms on both sides of the sum: they are equivalent.  
\end{proof}

\emph{Step 2:}  Now we further simplify (\ref{classsum}):
\begin{align}
&\mathbb{E}\left[\sum_{\mathcal{M} \in \mathbb{M}^2_{ji}: \mathbb{T}(\mathcal{M})\in [(Q_{\mathrm{L}},Q_{\mathrm{R}})]}  \frac{(-t)^{n-r}t^r}{(n-r)!r!} (A_i|\mathcal{L}_{X_n}\cdots \mathcal{L}_{X_{r+1}}\mathbb{P}_j \mathcal{L}_{X_r}\cdots \mathcal{L}_{X_1} |A_i) \right] \notag \\
&= \mathbb{E}\left[\sum_{\psi \in \Psi(Q_{\mathrm{L}},Q_{\mathrm{R}})} \int\limits_{\mathrm{\Delta}^{\ell_{\mathrm{L}}}(t)} \mathrm{d}t^{\mathrm{L}}_1\cdots \mathrm{d}t^{\mathrm{L}}_{\ell_{\mathrm{L}}} (A_i| \mathrm{e}^{-\mathcal{L}^\psi_\ell (t-t^{\mathrm{L}}_{\ell_{\mathrm{L}}})}(-\mathcal{L}^\psi_{X_\ell}) \mathrm{e}^{-\mathcal{L}^\psi_\ell (t^{\mathrm{L}}_{\ell_{\mathrm{L}}}-t^{\mathrm{L}}_{\ell_{\mathrm{L}}-1})} \cdots (-\mathcal{L}^\psi_{X_{\ell_{\mathrm{R}}+1}}) \mathrm{e}^{-\mathcal{L}^\psi_{\ell_{\mathrm{R}}}t_1^{\mathrm{L}}}\mathbb{P}_j   \right. \notag \\
&\left. \;\;\; \times \int\limits_{\mathrm{\Delta}^{\ell_{\mathrm{R}}}(t)} \mathrm{d}t^{\mathrm{R}}_1\cdots \mathrm{d}t^{\mathrm{R}}_{\ell_{\mathrm{R}}} \mathbb{P}_j \mathrm{e}^{\mathcal{L}^\psi_{\ell_{\mathrm{R}}}(t-t^{\mathrm{R}}_{\ell_{\mathrm{R}}})} \mathcal{L}^\psi_{\ell_{\mathrm{R}}} \cdots \mathrm{e}^{\mathcal{L}^\psi_1(t_2^{\mathrm{R}}-t_1^{\mathrm{R}})} \mathcal{L}^\psi_1 \mathrm{e}^{\mathcal{L}^\psi_0t_1^{\mathrm{R}}} |A_i) \right] \notag \\
&\le \mathbb{E}\left[\sum_{\psi \in \Psi(Q_{\mathrm{L}},Q_{\mathrm{R}})} \frac{t^{\ell_{\mathrm{L}}+\ell_{\mathrm{R}}}}{\ell_{\mathrm{L}}!\ell_{\mathrm{R}}!}\prod_{X\in \psi}(2|J_X|)\right] (A_i|A_i) \label{eq:step2theor4}
\end{align}
where in the first step, we applied Lemma \ref{lemma5}, and in the second step we used the orthogonality of $\mathrm{e}^{\mathcal{L}^\psi_k t}$, (\ref{eq:LXsubmult}) and (\ref{eq:41orth}).  

Lastly, we note that \begin{align}
\mathbb{E}\left[C_{ij}(t)^2\right] &= \mathbb{E}\left[\frac{1}{(A_i|A_i)}\sum_{(Q_{\mathrm{L}},Q_{\mathrm{R}})\in\mathcal{S}^2_{ji}} \sum_{\substack{\mathcal{M} \in \mathbb{M}^2_{ji}:\\ \mathbb{T}(\mathcal{M})\in [(Q_{\mathrm{L}},Q_{\mathrm{R}})]}}  \frac{(-t)^{r}}{r!} \frac{t^{n-r}}{(n-r)!}(A_i|\mathcal{L}_{X_1}\cdots \mathcal{L}_{X_r}\mathbb{P}_j \mathcal{L}_{X_{r+1}}\cdots \mathcal{L}_{X_n} |A_i)\right] \notag \\
&\le \mathbb{E}\left[\sum_{(Q_{\mathrm{L}},Q_{\mathrm{R}})\in\mathcal{S}^2_{ji}} \sum_{\psi \in \Psi(Q_{\mathrm{L}},Q_{\mathrm{R}})} \frac{t^{\ell_{\mathrm{L}}+\ell_{\mathrm{R}}}}{\ell_{\mathrm{L}}!\ell_{\mathrm{R}}!}\prod_{X\in \psi}(2|J_X|) \right] \label{eq:step2theor42}
\end{align}
where in the first line we used the fact that the only non-vanishing sequences are elements of $\mathbb{M}^2_{ji}$, and in the second line we used (\ref{eq:step2theor4}).
Since each factor in $\psi$ shows up exactly twice (due to irreducibility), we may use that $|J_X|^2 = J_X^2$ to evaluate the expectation value using (\ref{eq:theor4J}).   Hence the last line of (\ref{eq:step2theor42}) is equivalent to (\ref{eq:theor4}).
\end{proof}

\subsection{Causal Graphs}\label{sec:causalgraph}
Having proved Theorem~\ref{theor4}, our next major goal is to prove Theorem~\ref{theorFS}: the perturbative fast scrambling conjecture.  In order to achieve this, we organize the irreducible causal tree pairs in terms of graphs.  Given $(Q_{\mathrm{L}},Q_{\mathrm{R}})\in\mathcal{S}^2_{ji}$, the \emph{causal graph} is defined as $M=Q_{\mathrm{L}}\cup Q_{\mathrm{R}}$.  

When thinking of the causal graph as a factor subgraph of the Hamiltonian's factor graph, it is important to remember that the important information in each causal tree is the relative sequence of factors that arises.  We always take the causal graph to correspond to an indistinguishable causal tree pair $(Q_{\mathrm{L}},Q_{\mathrm{R}})$ whose union forms the graph of minimal genus.   An example of a causal graph is shown in Figure \ref{fig:causalgraphexample}.   As with causal trees, we say that two causal graphs are indistinguishable if they come from indistinguishable causal tree pairs.

\begin{figure}[t]
\centering
\includegraphics[width=0.75\textwidth]{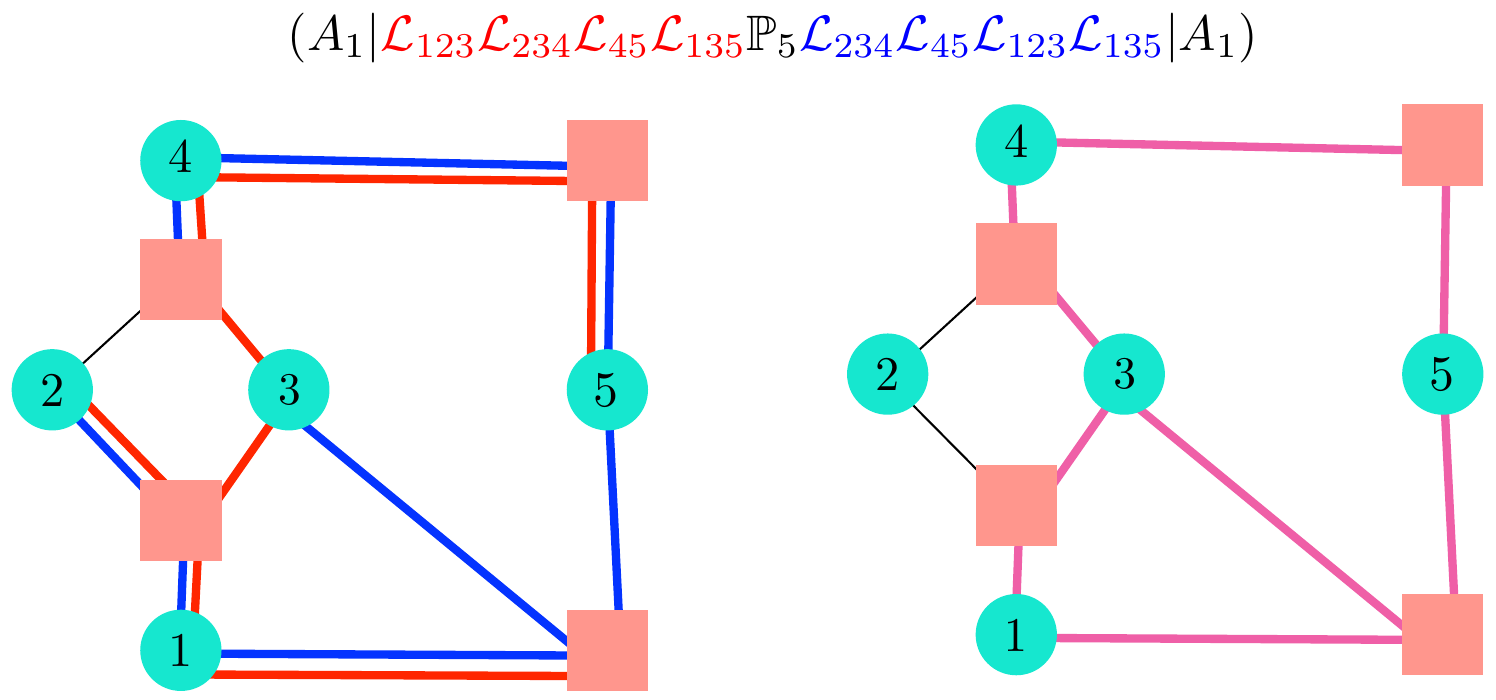}
\caption{Left: an irreducible causal tree pair drawn on the full factor graph.  Right: the causal graph.  Note that the irreducible pair is chosen so that the causal graph is genus 2, and not 3.}
\label{fig:causalgraphexample}
\end{figure}

  Recall that the genus $g$ of a (factor) graph is defined as \begin{equation}
g((V,F,E)) := |E| + 1 - |V| - |F|,  \label{eq:genusgen}
\end{equation}
and that $g\ge 0$.   A graph with $g=0$ is simply connected, i.e. a tree \cite{stallings}.   Generically, a non-contractible causal graph will be associated with multiple irreducible causal tree pairs.    In Section \ref{sec:FS},  we will see that the genus of the causal graph controls the order at $1/N$ at which a given irreducible causal tree pair contributes to our bound on $\mathbb{E}[C_{ij}(t)^2]$.  

%


In this section, we derive three useful properties of causal graphs for the proof of Theorem~\ref{theorFS} below.

\begin{prop}
If $M$ is a causal graph, $0 \le g(M) \le N-1$.
\label{propgenus}
\end{prop}
\begin{proof}
There exists an irreducible causal tree pair $(Q_{\mathrm{L}},Q_{\mathrm{R}})$ with $M=Q_{\mathrm{L}}\cup Q_{\mathrm{R}}$.    Observe that in causal tree $Q_{\mathrm{L}}$, a factor $X$ is the first to hit a node in a creeping sequence if and only if $\deg_{Q_{\mathrm{L}}}(X) > 1$; clearly the statement holds in $Q_{\mathrm{R}}$ also.  It follows that the number of such vertices $X$ is at most $N-1$, since an operator can only hit each node first a single time.   Also, since $(Q_{\mathrm{L}},Q_{\mathrm{R}})$ is irreducible, each factor $X\in M$ must obey $\max(\deg_{Q_{\mathrm{L}}}(X),\deg_{Q_{\mathrm{R}}}(X))  > 1$.   These previous two results can be combined to conclude that if $Q_{\mathrm{L,R}} = (V,F,E_{\mathrm{L,R}})$, $|E_{\mathrm{L}}\cup E_{\mathrm{R}}| \le 2(N-1)$.    Since $g(Q_{\mathrm{L}})=g(Q_{\mathrm{R}})=0$, \begin{equation}
g(M) = |E_{\mathrm{L}} \cup E_{\mathrm{R}}| +1-|V|-|F| = |E_{\mathrm{L}}\cup E_{\mathrm{R}}| - |E_{\mathrm{L}}| = |(E_{\mathrm{L}}\cup E_{\mathrm{R}} )- E_{\mathrm{L}}| \le N-1.
\end{equation}
The last inequality above follows from the fact that $|(E_{\mathrm{L}}\cup E_{\mathrm{R}}) - E_{\mathrm{L}}|$ is less than or equal to  the number of factors in $E_{\mathrm{R}}$ with degree larger than 1.
\end{proof}
\begin{prop}
 Let $(Q_{\mathrm{L}},Q_{\mathrm{R}})\in\mathcal{S}^2_{ji}$.  Suppose $M=Q_{\mathrm{L}}\cup Q_{\mathrm{R}}$ and $g(M)=g$.   Then there are at most $2g$ nodes or factors of degree greater than 2:  \begin{equation}
 |\lbrace a\in V\cup F : \deg(a) > 2 \rbrace | \le 2g,
 \end{equation}
 with $\deg$ = $\deg_M$, $\deg_{Q_{\mathrm{L}}}$ or $\deg_{Q_{\mathrm{R}}}$.
\label{propdegge2}
\end{prop}
\begin{proof}
Since $Q_{\mathrm{L,R}}\subseteq M$, a proof for $\deg = \deg_M$ suffices.  Let $y =  |\lbrace a\in V\cup F : \deg(a) > 2 \rbrace |$.   Since $(Q_{\mathrm{L}},Q_{\mathrm{R}})$ is irreducible, there are no factors in $M$ of degree 1 and at most 2 nodes ($i$ and $j$) of degree 1 in $M$.  We complete the proof by combining  \begin{equation}
2|E\cap M| = \sum_{v\in V\cap M} \deg(v) + \sum_{X\in F\cap M} \deg(X) \ge y+2(|V\cap M| + |F\cap M|)-2
\end{equation}  
with (\ref{eq:genusgen}): $|V\cap M| + |F\cap M| - 1 = |E\cap M| - g$.
\end{proof}

\begin{prop}
If $M$ is a causal graph of genus $g$, then there are at least $g+1$ factors in $M$.  \label{propnumfactors}
\end{prop}
\begin{proof}
Let $(Q_{\mathrm{L}},Q_{\mathrm{R}})\in\mathcal{S}^2_{ji}$ be an irreducible causal tree pair with $M=Q_{\mathrm{L}}\cup Q_{\mathrm{R}}$.  If $M$ has 1 factor, then there is a unique path from $i$ to $j$ and so $g(M)=0$.  

We now prove the proposition by induction, and contradiction.  Suppose that $M$ is a causal graph, $|F\cap M| = g+1$ and $g(M)=g+1$.   Suppose that there is a node $v \in V\cap M$ and a subset of edges $\lbrace (v,X_1),(v,X_2)\rbrace\subset E\cap M$ which can be removed from $M$:  $M^\prime = M-(v,\emptyset,\lbrace (v,X_1),(v,X_2)\rbrace)$ remains connected and causal.  Since the factor set is unchanged, $M^\prime$ is a smaller graph formed from the union of two smaller indistinguishable causal trees to $Q_{\mathrm{L,R}}$.  Hence $M$ was not in fact a causal graph, contradicting our assumptions.

Now suppose that we can remove a single edge from $e\in E\cap M$:  $M^\prime = M-(\emptyset,\emptyset, \lbrace e\rbrace)$ remains connected and causal.  If $e=(v,X)$ then $X$ cannot contain the first appearance of $v$ in any sequence of $T^{-1}(Q_{\mathrm{L,R}})$.  Hence $M$ was not causal.

Since factor graphs are bipartite between nodes and factors, it is not possible to reduce the genus by 1 while retaining connectivity in any other way.  This proves the proposition.
\end{proof}

Intuitively, Proposition~\ref{propnumfactors} states that the causal graph only becomes higher genus when creeping order can be broken.  This requires the insertion of an edge in between a chain of two other factors.   A graph of genus $g-1$ with a sequence of two factors has at least $g+1$ factors, consistent with the proposition.  

\subsection{Perturbative Proof of the Fast Scrambling Conjecture on Erd\"os-R\'enyi Hypergraphs}
\label{sec:FS}
We define $\mathrm{K}_N^q$ to be a model with all-to-all $q$-body interactions; or, more mathematically, the \emph{complete $q$-local factor graph}: let $V=\lbrace 1,\ldots, N\rbrace$, and \begin{equation}
\mathrm{K}_N^q := \left(V, \lbrace X\in \mathbb{Z}_2^V : |X|=q\rbrace, \lbrace (i,X): i\in X\rbrace  \right).
\end{equation}
We define $\mathrm{K}_N^q(m)$ analogously to the above, except that we also attach a flavor index $\alpha \in \lbrace 1,\ldots,m\rbrace$ to each factor:  
\begin{equation}
\mathrm{K}_N^q(m) := \left(V, \lbrace Y_\alpha: Y\in \mathbb{Z}_2^V , |Y|=q, \alpha \in \lbrace 1, \ldots, m\rbrace \rbrace, \lbrace (i,Y): i\in Y\rbrace  \right).
\end{equation}
In the quantum mechanical context, if all Hilbert spaces $\mathcal{H}_i$ have dimension $d$, it is natural to take $m=(d^2-1)^q$:  each ``flavor" of factor corresponds to a possible orthogonal $q$-local Hermitian operator.  

\begin{theor}[\textsf{\textbf{Perturbative Fast Scrambling in Complete $q$-Local Models}}]
Let $\mathbb{E}[\cdots]$ denote expectation value in a simple random Hamiltonian ensemble on $\mathrm{K}_N^q$ with $\lVert H_X\rVert = 1$ and \begin{equation}
\mathcal{J}_X^2 \le  \frac{\mathcal{J}^2(q-1)!}{qN^{q-1}}.  \label{eq:theorFSJ}
\end{equation}
Then for $i\ne j \in V$,
 \begin{equation}
\mathbb{E}\left[C_{ij}(t)^2\right] \le \frac{\mathrm{e}^{\lambda_* t}}{N}\sum_{g=0}^{N-1} g!\left(\frac{6144\mathrm{e}^4(q-1)^3}{qN}(\mathcal{J}t)^2\mathrm{e}^{\lambda_* t} \right)^{g}  \label{eq:theorFS}
\end{equation}
with 
\begin{equation}
 \lambda_* := 48\mathcal{J} \sqrt{\frac{q-1}{q}}.  \label{eq:lambdastar}
 \end{equation} 
\label{theorFS}
\end{theor}

\begin{proof}
We prove this theorem as follows.  (\emph{1}) First, we apply Theorem~\ref{theor4}.  We further loosen the bound and sum over the left and right causal tree relatively independently.  (\emph{2}) We enumerate over all possible creeping sequences behind the right causal trees for any given causal graph and left causal tree.  (\emph{3}) We enumerate possible ways of embedding a casual graph of factors into the full factor graph $G$. (\emph{4}) Then we sum over all possible causal graphs arising from a given left causal tree.  (\emph{5}) We enumerate all possible creeping sequences behind the left causal tree.  We show that the genus of the causal graph controls the power of $\frac{1}{N}$ at which a term contributes to the bound.  (\emph{6}) Combining these combinatoric results for sums over trees and further loosening the bound we arrive at the elegant formula (\ref{eq:theorFS}).

\emph{Step 1:}  First we apply Theorem \ref{theor4}.   Using the inequality \begin{equation}
\frac{1}{a!b!}\le \frac{2^{a+b}}{(a+b)!} \label{eq:factorialinequality}
\end{equation}
we find that \begin{align}
\mathbb{E}\left[C_{ij}(t)^2\right] &\le \mathbb{E}\left[\sum_{(Q_{\mathrm{L}},Q_{\mathrm{R}})\in\mathcal{S}^2_{ji}} \sum_{\psi \in \Psi(Q_{\mathrm{L}},Q_{\mathrm{R}})} \frac{(2t)^{\ell_{\mathrm{L}}+\ell_{\mathrm{R}}}}{\ell_{\mathrm{L}}!\ell_{\mathrm{R}}!}\prod_{X\in \psi}|J_X| \right] \notag \\
&\le \mathbb{E}\left[\sum_{(Q_{\mathrm{L}},Q_{\mathrm{R}})\in\mathcal{S}^2_{ji}} \sum_{\psi \in \Psi(Q_{\mathrm{L}},Q_{\mathrm{R}})} \frac{(4t)^{2\ell}}{(2\ell)!}\prod_{X\in \psi}|J_X| \right] = \sum_{(Q_{\mathrm{L}},Q_{\mathrm{R}})\in\mathcal{S}^2_{ji}}|\Psi(Q_{\mathrm{L}},Q_{\mathrm{R}})| \frac{(4t)^{2\ell}}{(2\ell)!}\prod_{X\in Q_{\mathrm{L}}}\mathcal{J}_X^2 \label{eq:theorFS1st}
\end{align}
where we defined $\ell_{\mathrm{R}} + \ell_{\mathrm{L}} := 2\ell$:   i.e. $\ell$ is the number of factors in the causal tree $Q_{\mathrm{L}}$ or $Q_{\mathrm{R}}$.    Intuitively, in this step we have chosen to ignore the location of $j$ in the sequences $(\mathcal{M}_{\mathrm{L}},\mathcal{M}_{\mathrm{R}})$.

We found it challenging to enumerate $|\Psi(Q_{\mathrm{L}},Q_{\mathrm{R}})|$ directly.   An easier object to work with is \begin{equation}
\mathcal{N}(Q) := |\lbrace \mathcal{M}= (i,X_1,\ldots,  X_n) : T(\mathcal{M}) = Q\rbrace |,
\end{equation}
the number of distinct creeping sequences whose causal tree is $Q$.  Clearly, \begin{equation}
|\Psi(Q_{\mathrm{L}},Q_{\mathrm{R}})| \le \frac{(2\ell)!}{\ell!^2} \mathcal{N}(Q_{\mathrm{L}})\mathcal{N}(Q_{\mathrm{R}}) < 2^{2\ell}  \mathcal{N}(Q_{\mathrm{L}})\mathcal{N}(Q_{\mathrm{R}}),
\end{equation}
as in the first inequality, we have counted all possible sequences, including those where $(Q_{\mathrm{L}},Q_{\mathrm{R}})$ is not irreducible.  We have also not enforced the constraint that $Q_{\mathrm{L,R}}$ are causal trees of the same sequence read in opposite orders.  Hence we obtain 
\begin{align}
&\mathbb{E}\left[C_{ij}(t)^2\right] < \sum_{(Q_{\mathrm{L}},Q_{\mathrm{R}})\in\mathcal{S}^2_{ji}} \frac{(8t)^{2\ell}}{(2\ell)!}\prod_{X\in Q_{\mathrm{L}}}\mathcal{J}_X^2 \mathcal{N}(Q_{\mathrm{L}}) \mathcal{N}(Q_{\mathrm{R}}) \notag \\
&< \sum_{\substack{Q_{\mathrm{L}}: (Q_{\mathrm{L}},Q_{\mathrm{R}})\in \mathcal{S}^2_{ji} \\ \text{for some }Q_{\mathrm{R}}} } \frac{(8t)^{2\ell}}{(2\ell)!} \mathcal{N}(Q_{\mathrm{L}})\prod_{X\in Q_{\mathrm{L}}}\mathcal{J}_X^2 \sum_{M:F\cap M = F\cap Q_{\mathrm{L}}} \sum_{\substack{Q_{\mathrm{R}} \in \mathcal{T}_{ji} \\ F\cap M = F\cap Q_{\mathrm{R}} }}\mathcal{N}(Q_{\mathrm{R}}) \notag \\
&< \sum_{g=0}^{N-1} \sum_{\ell=g+1}^\infty\frac{1}{(2\ell)!} \left(8\mathcal{J}t \sqrt{\frac{(q-1)!}{qN^{q-1}}}\right)^{2\ell} \sum_{\substack{Q_{\mathrm{L}}: (Q_{\mathrm{L}},Q_{\mathrm{R}})\in \mathcal{S}^2_{ji} \\ \text{for some }Q_{\mathrm{R}}, \\  |F\cap Q_{\mathrm{L}}| = \ell } }  \mathcal{N}(Q_{\mathrm{L}}) \sum_{\substack{M:g(M)=g \\ F\cap M = F\cap Q_{\mathrm{L}}}} \sum_{\substack{Q_{\mathrm{R}} \in \mathcal{T}_{ji} \\ F\cap M = F\cap Q_{\mathrm{R}} }} \mathcal{N}(Q_{\mathrm{R}})  \label{eq:theorFSstep1}
\end{align}
where in the second step above, we have further organized the sum by the way in which the graph $M$ arises from the factors of $Q_{\mathrm{L}}$, and in the last step above we employed (\ref{eq:theorFSJ}), and organized the sum by the genus $g$ of the causal graph and the number $\ell$ of factors in the irreducible causal tree pair.  We also employed Proposition~\ref{propgenus} and Proposition~\ref{propnumfactors} in the last line above to restrict the range of the sums over $g$ and $\ell$.

\emph{Step 2:}  What remains is to enumerate $\mathcal{N}(Q_{\mathrm{R}})$ and $\mathcal{N}(Q_{\mathrm{L}})$.  First, keeping $g$ and $\ell$ fixed, we will bound the sum over all possible choices of $Q_{\mathrm{R}}$ using the following loose (but simple) result: 
 \begin{lma}
 Let $M$ be a causal graph with $g(M)=g$.  Then 
 \begin{align}
\sum_{\substack{Q_{\mathrm{R}} \in \mathcal{T}_{ji} \\ F\cap M = F\cap Q_{\mathrm{R}} }} \mathcal{N}(Q_{\mathrm{R}}) < (6(g+1))^\ell \label{eq:lma17}
 \end{align}
 \label{sumR}
 \end{lma}
\begin{proof}
Using Proposition \ref{propdegge2}, there are at most $2g$ nodes/factors of degree $>2$ in $M$.   For an ordinary undirected graph, the largest number of edges in a genus $g$ graph where all vertices have degree $>2$ is $3g-1$, using (\ref{eq:genusgen}).   This implies that the actual causal graph can be constructed by joining together at most $3g+1$ linear subgraphs consisting of alternating nodes and factors: (\emph{1})  segments $\gamma_i$ and $\gamma_j$ may contain $i$ and $j$ as endpoints respectively, but then must be traversed in a fixed creeping order;  (\emph{2}) otherwise, there are line segments $\gamma_1,\ldots, \gamma_{3g-1}$, which a priori can be traversed in creeping order from both endpoints simultaneously.    Let us write $\ell_i$ and $\ell_j$ to be the number of factors in $\gamma_i$ and $\gamma_j$ respectively, and $\ell^{1,2}_k$ for $1\le k\le 3g-1$ to be the number of factors on 2 line segments in $\gamma_k$, along which creeping order is enforced.   Clearly, $\ell^1_k + \ell^2_k = \ell(\gamma_k)$.  See Figure \ref{fig:QRlemma} for a demonstration.

\begin{figure}[t]
\centering
\includegraphics[width=0.85\textwidth]{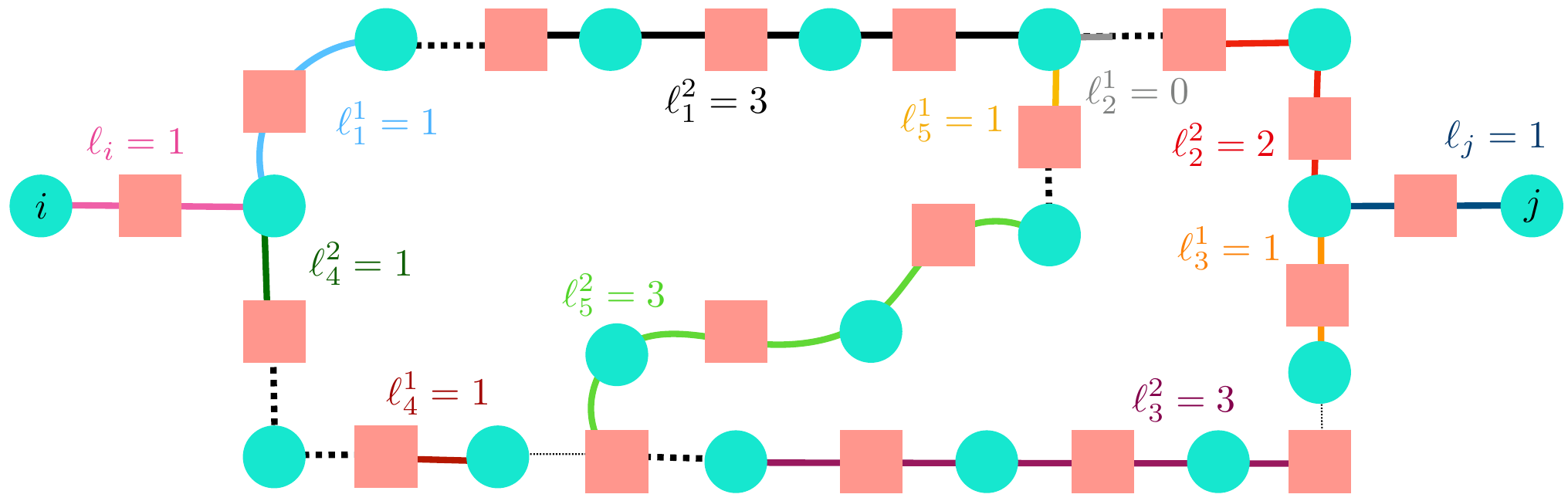}
\caption{An example of how we overcount all possible ways of traversing $Q_{\mathrm{R}}$ in Lemma~\ref{sumR}.  The thick dotted lines denote edges in $Q_{\mathrm{R}}$ which put additional constraints on the relative ordering of factors which we have not accounted for.  The thin dotted lines denote edges in $Q_{\mathrm{L}}$ and $M$ (but not $Q_{\mathrm{R}}$) which are not traversed.  We have labeled each $\ell_k^{1,2}$ and $\ell_{i,j}$ with a unique color for clarity. }
\label{fig:QRlemma}
\end{figure}

Suppose $g>0$, and consider the inequalities \begin{align}
\sum_{\substack{Q_{\mathrm{R}} \in \mathcal{T}_{ji} \\ F\cap M = F\cap Q_{\mathrm{R}} }} \mathcal{N}(Q_{\mathrm{R}}) &< \sum_{\ell_i, \ell_j, \ell^1_k, \ell_2^k} \frac{\ell!}{\ell_i!\ell_j!} \prod_{k=1}^{3g-1} \frac{\mathbb{I}\left(\ell^1_k + \ell^2_k = \ell(\gamma_k)\right)}{\ell^1_k! \ell^2_k!} \notag \\
&< \sum_{m_1,\ldots, m_{6g}=0}^\infty \mathbb{I}\left(\ell = \sum_{k=1}^{6g} m_k\right) \ell!\prod_{k=1}^{6g} \frac{1}{m_k!} = (6g)^\ell.  \label{eq:lma17int}
\end{align}
In the first line, the inequality comes from the fact that the right hand side counts numerous sequences which are not globally creeping.   In the second line, we have relaxed the constraints on the lengths of each (pair of) line segments, and simply summed over all possible causal trees on all topologically equivalent graphs $M$.  The multinomial theorem was used to simplify the final answer.   If $g=0$, then there is a unique choice of $Q_{\mathrm{R}}$ and a unique ordering of factors which is creeping.   Combining this fact with (\ref{eq:lma17int}) we obtain (\ref{eq:lma17}).
\end{proof}

\emph{Step 3:} The following two steps are summarized in Figure~\ref{fig:steps3and4}.  We will sum over the possible ways to choose trees $Q_{\mathrm{L}}$ and causal graphs $M$ as follows: in (\ref{eq:theorFSstep1}), we replace
\begin{equation}
 \sum_{\substack{Q_{\mathrm{L}}: (Q_{\mathrm{L}},Q_{\mathrm{R}})\in \mathcal{S}^2_{ji} \\ \text{for some }Q_{\mathrm{R}}, \\  |F\cap Q_{\mathrm{L}}| = \ell } }  \mathcal{N}(Q_{\mathrm{L}}) \sum_{\substack{M:g(M)=g \\ F\cap M = F\cap Q_{\mathrm{L}}}} =  \sum_{\substack{Q_{\mathrm{L}} \in \mathcal{T}_{ji}/\imath \\  |F\cap Q_{\mathrm{L}}| = \ell } }  \mathcal{N}(Q_{\mathrm{L}}) \sum_{\substack{M/ \imath \\ g(M)=g \\ F\cap M = F\cap Q_{\mathrm{L}}}} \sum_{\imath}
\end{equation}
In this equation $\imath$ is the ``inclusion" of the irreducible causal tree pair of factors $(Q_{\mathrm{L}},Q_{\mathrm{R}})$ into the full factor graph $G=(V,F,E)$.   

\begin{figure}[t]
\centering
\includegraphics[width=\textwidth]{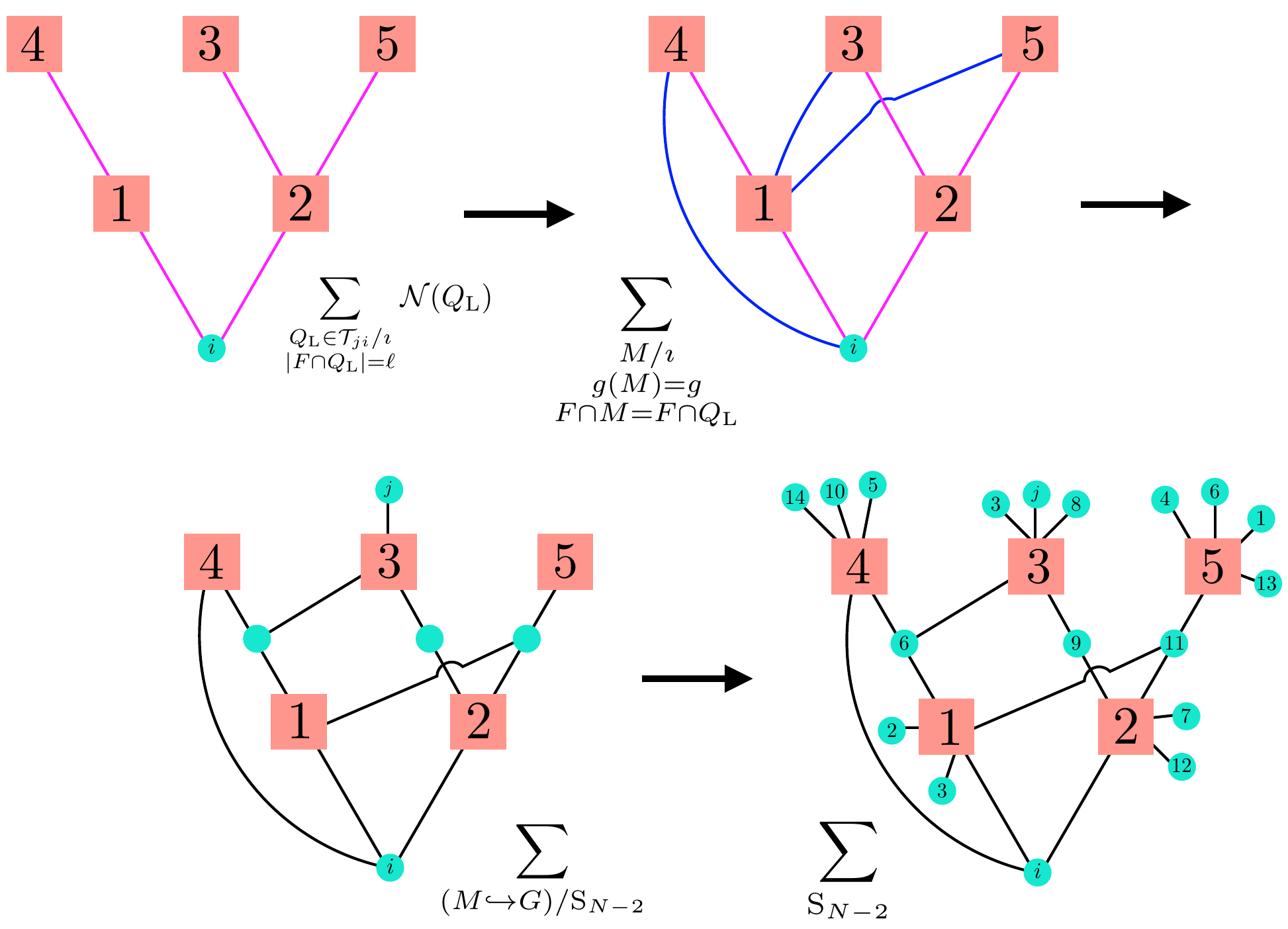}
\caption{A demonstration of the splitting of sums in \emph{Steps 3-4} of the proof of Theorem~\ref{theorFS}.  Note that the numbering on factors is used to denote the order in which they appear in $\mathcal{M}_{\mathrm{L}}$.   Also note that in the last step, we have allowed the same node to show up in ``forgotten" vertices multiple times, as this need not contribute higher genus to the causal graph.}
\label{fig:steps3and4}
\end{figure}

The reason for including a sum over $\imath$ is that two isomorphic causal graphs of factors can have non-isomorphic embeddings into the full factor graph depending on what nodes are shared in common: see Figure~\ref{fig:steps3and4}.   We then further write
\begin{align}
\sum_{\imath} = \sum_{(M\hookrightarrow G)/\mathrm{S}_{N-2}} \sum_{\mathrm{S}_{N-2}}
\end{align}
where $\mathrm{S}_{N-2}$ denotes the action of the permutation group on the nodes $V$ keeping $i$ and $j$ fixed, and $M\hookrightarrow G/\mathrm{S}_{N-2}$ denotes the number of non-isomorphic subgraphs of a factor graph which come from identical causal graphs of factors alone.

We first bound the sum over permutations $\mathrm{S}_{N-2}$ as follows:
\begin{align}
\sum_{\mathrm{S}_{N-2}, \text{ given }M\subseteq G}1 &\le  N^{|V\cap M|-2} \times \prod_{X\in F\cap M} \frac{N^{q-|\partial_M X|}}{(q-|\partial_MX|)!} \notag \\
&\le  N^{|V\cap M|-2} \times \prod_{X\in F\cap M} \frac{N^{q-|\partial_M X|}}{(q-2)!} (q-1)^{|\partial_M X| - 2}\notag \\
&\le (q-1)^{|E\cap M| - 2\ell}  \times\frac{ N^{|V\cap M| - 2 + q\ell - |E\cap M|}}{(q-2)!^\ell} \notag \\
&\le (q-1)^{g-\ell +|V\cap M|-1} \times \frac{N^{(q-1)\ell-1-g}}{(q-2)!^\ell} = \left(\frac{N^{q-1}}{(q-2)!}\right)^\ell \frac{1}{(q-1)N} \left(\frac{(q-1)^2}{N}\right)^g. \label{eq:sumoverM}
\end{align}
In the first line, the first factor corresponds to the choice of all vertices in $V\cap M - \lbrace i,j\rbrace$, and the second factor corresponds to the number of ways to choose any nodes for the factors of degree 1 (which are not included in the irreducible causal graph).  Recall that $\partial_MX$ is the degree of factor $X$ in the subgraph $M\subseteq G$.   In the second line, we used the fact that the only factor which can have degree 1 contains $j$ (and if it has degree 1, then $j$ must appear as a node of that factor), to replace the inequality by a simpler one to manipulate.   In the third line, we used the fact that factor graphs are bipartite to replace $\sum_X |\partial_M X| = |E\cap M|$, along with $|F\cap M| = \ell$.  In the final line, we used (\ref{eq:genusgen}) together with the simple inequality \begin{equation}
|V\cap M | \le g+\ell
\end{equation}
to simplify the answer.

The sum over non-isomorphic embeddings is related to the question of which factors share which nodes, together with the location of $j$.  An extremely simple way to bound this is as follows: firstly, there are at most $(q-1)\ell$ factors in $V\cap M - i$, which we could assign to $j$.   Secondly, since there are $\ell+g-1$ edges between factors in $M$ (unembedded in $G$), we may bound \begin{equation}
\ell-1 \le |V\cap M | -2 \le \ell +g -1.
\end{equation}
For any fixed value of $|V\cap M|$, we may allocate all of these nodes (except possibly $i$ and $j$) to the edges between factors in the causal graph of factors.  Using the multinomial theorem: \begin{align}
 \sum_{(M\hookrightarrow G)/\mathrm{S}_{N-2}}1 &\le (q-1)\ell\times  \sum_{m=\ell-1}^{\ell-1+g} \sum_{k_1,\ldots, k_m = 1}^{\ell+g-1} \frac{(\ell+g-1)!}{m!} \prod_{n=1}^m \frac{1}{k_n!} \notag \\
 &< (q-1)\ell \sum_{m=\ell-1}^{\ell-1+g} \frac{m^{\ell+g-1}}{m!} < (q-1)\ell \mathrm{e}^{\ell+g}. \label{eq:sumisomorphism}
\end{align}
The first factor above corresponded to the placement of $j$, and the second to the possible ways to share vertices between the factors in the causal graph.

\emph{Step 4:} Next, we will perform the sum over all causal graphs $M$ which can arise given a fixed causal tree $Q_{\mathrm{L}}$, up to the inclusion $\imath$ into the full factor graph $G$.  This is easy to enumerate: as we are now considering the orderings of factors alone, we simply choose $g$ pairs of factors (or node $i$) and add an edge between them.  Accounting for a permutation symmetry factor of $2^gg!$:
\begin{align}
\sum_{\substack{M/ \imath \\ g(M)=g \\ F\cap M = F\cap Q_{\mathrm{L}}}} 1 \le\frac{(\ell+1)^{2g}}{2^gg!} \label{eq:Mgiveni1}
\end{align}
Keep in mind that all of the edges we have added are implicitly understood to be part of the right causal tree $Q_{\mathrm{R}}$.   Using the fact that there is a unique choice of $M$ at each genus $g$ when the number of factors is $\ell=1$ or 2, we can further simplify (\ref{eq:Mgiveni1}) to \begin{equation}
\sum_{\substack{M/ \imath \\ g(M)=g \\ F\cap M = F\cap Q_{\mathrm{L}}}} 1 \le \frac{\ell^{2g}}{g!} \label{eq:Mgiveni2}
\end{equation}

\emph{Step 5:}  Next, we will bound $\mathcal{N}(Q_{\mathrm{L}})$ for all possible trees at fixed $g$ and $\ell$, up to (factor) graph isomorphism.  We define the \emph{number of branches} of a causal tree to be the number of nodes or factors of degree 1, not including $i$.  Using Proposition~\ref{propdegge2} to constrain $b\le 2g+1$, we then write \begin{equation}
\sum_{\substack{Q_{\mathrm{L}} \in \mathcal{T}_{ji}/\mathrm{S}_{N-2} \\  |F\cap Q_{\mathrm{L}}| = \ell } } \mathcal{N}(Q_{\mathrm{L}})  \le  \sum_{b=1}^{2g+1} \mathcal{N}(b,\ell),
\end{equation}
where $\mathcal{N}(b,\ell)$ is the number of creeping sequences of $\ell$ factors whose causal tree has $b$ branches, up to graph isomorphism.

\begin{lmaNB}
\begin{equation}
\mathcal{N}(b,\ell)<   b^\ell . \label{eq:lma15a} 
\end{equation}
\label{lma15}
\end{lmaNB}
\begin{proof}
We follow standard generating function techniques \cite{sedgewick}.    Define
\begin{subequations}\begin{align}
A(t,x) &:= \sum_{n=1}^\infty \sum_{\ell=1}^\infty \frac{x^nt^\ell}{\ell!} \mathcal{N}(b,\ell)
\end{align}\end{subequations}
We will derive a recursive equation for $A(t,x)$.  Observe that $\mathcal{N}(b,\ell)$ is also the number of creeping sequences with length $\ell+1$ factors, assuming that the first factor in the creeping sequence is the only one to contain $i$.  A generating function for the number of these creeping sequences is given by \begin{equation}
A_0(t,x) = xt+ \int\limits_0^t\mathrm{d}t^\prime \; A(t^\prime,x) = xt+ \sum_{b=1}^\infty \sum_{\ell=1}^\infty \frac{x^bt^{\ell+1}}{(\ell+1)!} \mathcal{N}(b,\ell).  \label{eq:A0A}
\end{equation}
The first factor of $xt$ is because a two vertex graph needs to be counted in $A_0$, but there is no one vertex graph counted in $A$.

Next, we claim that $A_0$ obeys the following differential equation:  \begin{equation}
\frac{\partial}{\partial t} A_0 = x + \sum_{n=1}^\infty \frac{A_0^n}{n!} = x + \mathrm{e}^{A_0} - 1.  \label{eq:AOdt}
\end{equation}
To understand why, note that the left hand side of (\ref{eq:AOdt}) corresponds to removing the unique choice of first factor.   On the right hand side, we are faced with the two possible outcomes:  (\emph{1}) there are no remaining factors, in which case we are left with $\frac{\partial}{\partial t}(xt)=x$, or (\emph{2}) there are $n$ factors which can appear next.  Each of these $n$ factors itself gives rise to a causal subtree generated by a creeping sequence whose first factor is fixed.   Now suppose that these subtrees had $\widetilde{\mathcal{N}}_1,\ldots, \widetilde{\mathcal{N}}_n$ ways of growing.  Then the number of ways of growing their union is given by \begin{equation}
\widetilde{\mathcal{N}}_{\mathrm{tot}} = \frac{1}{n!}\times \frac{(\ell_1+\cdots + \ell_n)!}{\ell_1!\cdots \ell_n!} \widetilde{N}_1\cdots \widetilde{N}_n.  \label{eq:tildeNtot}
\end{equation}
Due to permutation symmetry among the $n$ subtrees, we must divide by $\frac{1}{n!}$.   Putting this all together, and adding back the factors of $x$ and $t$ to form a generating function, we obtain (\ref{eq:AOdt}).  

Solving (\ref{eq:AOdt}) with initial condition $A_0(0,x)=0$, we find \begin{equation}
t = \frac{1}{1-x}\log \frac{1 - (1-x)\mathrm{e}^{-A_0(t,x)} }{x}\label{eq:A0imp}
\end{equation}
Using (\ref{eq:A0A}) and (\ref{eq:A0imp}) we obtain 
\begin{equation}
A(t,x) = \sum_{b=1}^\infty \sum_{\ell=1}^\infty \frac{x^bt^\ell}{\ell!} \mathcal{N}(b,\ell) = \frac{x\mathrm{e}^t-x\mathrm{e}^{tx}}{\mathrm{e}^{tx} - x\mathrm{e}^t}.  \label{eq:rootgen}
\end{equation} 

Since $\mathcal{N}(b,\ell) \ge 0$, consider the generating function \begin{equation}
\widetilde{A}_*(t,x) := \sum_{\ell=1}^\infty \frac{x^bt^\ell}{\ell!} \widetilde{\mathcal{N}}_*(b,\ell) := \widetilde{A}\left(\frac{t}{1-x},x\right) = \frac{x\mathrm{e}^t - x}{1-x\mathrm{e}^t}.
\end{equation}
which encodes a sequence obeying $\mathcal{N}(b,\ell) \le\widetilde{\mathcal{N}}_*(b,\ell)$.  Since \begin{equation}
\widetilde{\mathcal{N}}_*(b,\ell) = b^\ell - (b-1)^\ell < b^\ell,
\end{equation}
and since a tree with $b$ branches clearly needs at least $\ell$ edges, we obtain (\ref{eq:lma15a}).
\end{proof}


\emph{Step 6:} We now combine all of the results above to obtain
\begin{align}
\mathbb{E}\left[C_{ij}(t)^2\right] &< \sum_{g=0}^{N-1} \sum_{\ell=g+1}^\infty\frac{1}{(2\ell)!} \left(8\mathcal{J}t \sqrt{\frac{(q-1)!}{qN^{q-1}}}\right)^{2\ell} \sum_{\mathrm{S}_{N-2}} \sum_{\substack{Q_{\mathrm{L}} \in \mathcal{T}_{ji}/\mathrm{S}_{N-2} \\  |F\cap Q_{\mathrm{L}}| = \ell } }  \mathcal{N}(Q_{\mathrm{L}}) \sum_{\substack{M/ \mathrm{S}_{N-2} \\ g(M)=g \\ F\cap M = F\cap Q_{\mathrm{L}}}} (6(g+1))^\ell \notag \\
&< \sum_{g=0}^{N-1} \sum_{\ell=g+1}^\infty\frac{1}{(2\ell)!} \left(8\mathcal{J}t \sqrt{\frac{(q-1)!}{qN^{q-1}}}\right)^{2\ell}  \left(\frac{6\mathrm{e}(g+1)N^{q-1}}{(q-2)!}\right)^\ell \left(\frac{\mathrm{e}(q-1)^2\ell^2}{N}\right)^{g} \frac{\ell}{g! N} \sum_{b=1}^{2g+1}b^\ell
\notag \\
&< \sum_{g=0}^{N-1} \sum_{\ell=g+1}^\infty\frac{1}{(2\ell)!} \left(8\mathcal{J}t \sqrt{\frac{(q-1)!}{qN^{q-1}}}\right)^{2\ell}\left(\frac{12\mathrm{e}(g+1)^2 N^{q-1}}{(q-2)!}\right)^\ell \left(\frac{\mathrm{e}(q-1)^2\ell^2}{N}\right)^{g} \frac{1}{g! N} \notag \\
&< \frac{1}{N} \sum_{g=0}^{N-1}\sum_{\ell=g+1}^\infty \frac{1}{(2\ell)!g!N} \left(\frac{\mathrm{e}(q-1)^2\ell^2}{N}\right)^{g} \left(48(g+1)\mathcal{J}t \sqrt{\frac{q-1}{q}}\right)^{2\ell} \notag \\
&< \frac{1}{N}\sum_{g=0}^{N-1}\sum_{m=1}^\infty g!\left(\frac{6144\mathrm{e}^4(q-1)^3}{qN}(\mathcal{J}t)^2 \right)^{g} \frac{1}{(2m)!} \left(48(g+1)\mathcal{J}t \sqrt{\frac{q-1}{q}}\right)^{2m}  \notag \\
&<\frac{\mathrm{e}^{\lambda_* t}}{N}\sum_{g=0}^{N-1} g!\left(\frac{6144\mathrm{e}^4(q-1)^3}{qN}(\mathcal{J}t)^2\mathrm{e}^{\lambda_* t} \right)^{g} 
 \end{align}
 where in the first line we used Lemma \ref{sumR}, in the second line we used Lemma \ref{lma15}, (\ref{eq:sumoverM}), (\ref{eq:sumisomorphism}) and (\ref{eq:Mgiveni2}), in the third line we used \begin{equation}
 \sum_{b=1}^{2g+1} b^\ell < \int\limits_1^{2g+2} \mathrm{d}x x^\ell < \frac{(2g+2)^\ell}{\ell},
 \end{equation}
 in the fifth line we used that when $\ell > g$, \begin{equation}
 \frac{\ell^{2g} (g+1)^{2g}}{g!(2\ell)!} < \frac{(2\mathrm{e}^2(g+1))^g}{(2(\ell-g))!} < \frac{2^g\mathrm{e}^{3g}g!}{(2(\ell-g))!} 
 \end{equation}and in the last line we used (\ref{eq:lambdastar}).  
\end{proof}

  Let $\mathcal{G}(V,m)$ be the set of all possible factor graphs with node set $V$ with up to $m$ distinct factors  $X\in F$ with identical neighbors $\partial X$.  We define the probability space $\mathrm{G}_V(q,k,m):= (\mathcal{G}(V,m), \sigma(\mathcal{G}(V,m)), \mu_{q,k})$ as follows:    \begin{equation}
\mu_{q,k}\left[ \lbrace (V,F,E)\rbrace  \right] = \prod_{X\in F} \left(\mathbb{I}(|X| = q) \frac{(q-1)! k}{N^{q-1}m} \right)
\end{equation}
This is called an \emph{Erd\"os-R\'enyi random hypergraph ensemble} of factor graphs \cite{shamir}.   Note that the probability measure $\mu_{q,k}$ corresponds to independently choosing each possible degree $q$ factor with probability $\frac{(q-1)!(N-q)! k}{(N-1)!}$, such that $\mathbb{E}_G[|\partial i|] = k + \mathrm{O}(\frac{1}{N})$ for any $i\in V$.

\begin{corol}[\textsf{\textbf{Perturbative Fast Scrambling on Erd\"os-R\'enyi Hypergraphs}}]
Let $\mathbb{E}_{q,k}$ denote expectation value in $\mathrm{G}_V(m,q,k)$, with $q=\mathrm{O}(N^0)$.  For any $G\in \mathrm{G}_V$, let $\mathbb{E}_G[\cdots]$ denote expectation value in a simple random Hamiltonian ensemble on $G$ with $\lVert H_X\rVert = 1$ and \begin{equation}
\mathcal{J}_X^2 \le  \frac{\mathcal{J}^2}{qk}.  \label{eq:cor19J}
\end{equation}
Then $\mathbb{E}_{q,k}[\mathbb{E}_G[C_{ij}(t)^2]]$ is bounded by (\ref{eq:theorFS}) and (\ref{eq:lambdastar}).
\label{corolFS}
\end{corol}

\begin{proof}
Let $(V,F,E)\in\mathcal{G}(V,m)$.  For $X\in F$, define the independent random variables\footnote{If $Z\sim \mathrm{Bernoulli}(p)$, $\mathbb{P}(Z=0) = 1-p$ and $\mathbb{P}(Z=1) = p$.} \begin{equation}
Z_X \sim \mathrm{Bernoulli}\left(\mathbb{I}(|X|=q) \frac{(q-1)!k}{N^{q-1}m}\right). \label{eq:ZX}
\end{equation}
It is convenient to think about the problem not on the factor graph $G$, but instead on the complete factor graph $\mathrm{K}_N^q(m)$ with $m$ flavors, with coupling constants $Z_XJ_X$.  Using Theorem \ref{theor4}, and the fact that $Z_X=Z_X^2$, 
\begin{align}
\mathbb{E}_{q,k}\left[\mathbb{E}_G\left[C_{ij}(t)^2\right]\right] &\le \mathbb{E}_{q,k}\left[\mathbb{E}_G\left[\sum_{(Q_{\mathrm{L}},Q_{\mathrm{R}})\in\mathcal{S}^2_{ji}(\mathrm{K}_N^q(m))} \sum_{\psi \in \Psi(Q_{\mathrm{L}},Q_{\mathrm{R}})} \frac{(2t)^{\ell_{\mathrm{L}}+\ell_{\mathrm{R}}}}{\ell_{\mathrm{L}}!\ell_{\mathrm{R}}!}\prod_{X\in \psi}(Z_X |J_X|) \right]\right] \notag \\
&\le \mathbb{E}_{\mathrm{K}_N^q(m)}\left[\sum_{(Q_{\mathrm{L}},Q_{\mathrm{R}})\in\mathcal{S}^2_{ji}(\mathrm{K}_N^q(m))} \sum_{\psi \in \Psi(Q_{\mathrm{L}},Q_{\mathrm{R}})} \frac{(2t)^{\ell_{\mathrm{L}}+\ell_{\mathrm{R}}}}{\ell_{\mathrm{L}}!\ell_{\mathrm{R}}!}\prod_{X\in Q_{\mathrm{L}}}\left(\mathbb{E}_{q,k}(Z_X) J_X^2\right) \right] \notag \\
&\le \sum_{(Q_{\mathrm{L}},Q_{\mathrm{R}})\in\mathcal{S}^2_{ji}(\mathrm{K}_N^q(m))} \sum_{\psi \in \Psi(Q_{\mathrm{L}},Q_{\mathrm{R}})} \frac{(2t)^{\ell_{\mathrm{L}}+\ell_{\mathrm{R}}}}{\ell_{\mathrm{L}}!\ell_{\mathrm{R}}!}\prod_{X\in Q_{\mathrm{L}}}  \frac{\mathcal{J}^2}{qk}\frac{(q-1)! k}{mN^{q-1}} \notag \\
&\le \sum_{(Q_{\mathrm{L}},Q_{\mathrm{R}})\in\mathcal{S}^2_{ji}(\mathrm{K}_N^q)} \sum_{\psi \in \Psi(Q_{\mathrm{L}},Q_{\mathrm{R}})} \frac{(2t)^{\ell_{\mathrm{L}}+\ell_{\mathrm{R}}}}{\ell_{\mathrm{L}}!\ell_{\mathrm{R}}!}\prod_{X\in Q_{\mathrm{L}}}  \frac{\mathcal{J}^2}{q}\frac{(q-1)! }{N^{q-1}}
\end{align}
where in the second line we used linearity of expectation value and independence of $Z_X$, in the third line we used (\ref{eq:cor19J}) and (\ref{eq:ZX}), and in the last line we used that there are $m^{\ell(Q_{\mathrm{L}})}$ causal trees in $\mathrm{K}_N^q(m)$ connecting the same nodes (in the same order) as there are in $\mathrm{K}_N^q$.  The last line above is equivalent to (\ref{eq:theorFS1st}), except that we have not yet simplified the sum over $\Psi(Q_{\mathrm{L}},Q_{\mathrm{R}})$.  The rest of the proof of Theorem~\ref{theorFS} hence applies.
\end{proof}

We now discuss the implication of Theorem~\ref{theorFS} and Corollary~\ref{corolFS}.  While the bound (\ref{eq:theorFS}) appears to be an asymptotic series as $N\rightarrow \infty$, it is absolutely convergent for times $\lambda_{\mathrm{L}}t \le a$, with $a=\mathrm{O}(N^0)$.  To see this, use the identity $g! < N^g$ in (\ref{eq:theorFS}), which holds for $g<N$: \begin{align}
\mathbb{E}\left[C_{ij}(t)^2\right] &\le q \frac{\mathrm{e}^{\lambda_*t}}{N} \sum_{g=0}^{N-1} \left(6144\mathrm{e}^4(q-1)^2(\mathcal{J}t)^2\mathrm{e}^{\lambda_* t} \right)^{g} \notag \\
&< q\frac{\mathrm{e}^{\lambda_*t}}{N} \left[ 1 -6144\mathrm{e}^4(q-1)^2(\mathcal{J}t)^2\mathrm{e}^{\lambda_* t} \right]^{-1}.  \label{eq:gfactorialconvergence}
\end{align}
Note we have also simplified the formula slightly using $q^{-1}(q-1)^3 < (q-1)^2$.  
  Even in the limit $N\rightarrow \infty$, this Taylor series converged for short times, since as $t\rightarrow 0$, $t\mathrm{e}^{\lambda_* t} \rightarrow 0$, as $t\rightarrow 0$ the bound lies within the radius of absolute convergence of the series.   Unfortunately, this radius of convergence is too small to prove (\ref{eq:susskind}).

Corollary \ref{corolFS} implies that a typical Hamiltonian in a simple random Hamiltonian ensemble a typical\footnote{The notion of typical being used here is that of Erd\"os and R\'enyi, as is canonical in random (hyper)graph theory \cite{shamir} when a more specific ensemble is not provided.} $q$-local factor graph of arbitrary degree is a fast scrambler to almost any order in perturbation theory: for any genus $g=  \mathrm{O}(N^{1-a})$ with $a > 0$, we find an operator scrambling time $t_{\mathrm{s}}^\delta = \mathrm{O}(a\log N)$.    After all, if we truncate the genus expansion at order $g_* = AN^{1-a} \in \mathbb{Z}$, with $A=\mathrm{O}(N^0)$: \begin{align}
\mathbb{E}\left[C_{ij}(t)^2\right] &\le q\frac{\mathrm{e}^{\lambda_*t}}{N} \sum_{g=0}^{g_*} g!\left(\frac{6144\mathrm{e}^4(q-1)^2}{N}(\mathcal{J}t)^2\mathrm{e}^{\lambda_* t} \right)^{g} + \mathrm{O}\left(\frac{1}{N^{2+g_*}}\right) \notag \\
&< q\frac{\mathrm{e}^{\lambda_*t}}{N} \sum_{g=0}^{g_*} \left(6144\mathrm{e}^4(q-1)^2 \frac{g_*}{N} (\mathcal{J}t)^2\mathrm{e}^{\lambda_* t} \right)^{g} + \mathrm{O}\left(\frac{1}{N^{2+g_*}}\right) \notag \\
&< q \frac{\mathrm{e}^{\lambda_*t}}{N} \left[1 - \frac{6144\mathrm{e}^4(q-1)^2}{N}(\mathcal{J}t)^2\mathrm{e}^{\lambda_* t} \frac{A}{N^a}\right]^{-1} + \mathrm{O}\left(\frac{1}{N^{2+g_*}}\right).
\end{align}
In fact, as is typical with (nearly) asymptotic series, up until $\lambda_*t\sim 1$, the bound (\ref{eq:theorFS}) is very well approximated by the leading order term, up until time scales very close to the radius of convergence.  At this point, the bound abruptly diverges.  We expect this divergence is a failure of our bound and not a counterexample to fast scrambling in regular systems.

\subsection{On the Combinatorics at High Genus}\label{sec:highgenus}
In fact, it may be the case that the $g!$ in (\ref{eq:theorFS}) is physically meaningful.    In ordinary quantum field theory, asymptotic series are quite common and arise from the fact that there are $\mathrm{O}(g!)$ topologically distinct graphs (Feynman diagrams) at genus $g$ \cite{riddell}.   Moreover, on physical grounds one expects these series to be asymptotic \cite{dyson}.  Mathematically speaking, the proliferation of topologically distinct graphs is also the origin of the factor of $g!$ in the proof of Theorem~\ref{theorFS}.   

One interesting observation is that in quantum field theory, there is no reason to truncate the genus expansion at genus $g=N-1$.  After all, the same particle or node could show up arbitrarily often in a high genus graph.  On physical grounds, it is actually natural to expect that (\ref{eq:theorFS}) could be the qualitatively correct form of $C_{ij}(t)^2$, with one crucial modification:  $g! \rightarrow (-1)^gg!$.  Such a modification would make the series (\ref{eq:theorFS}) Borel resummable \cite{bender} and prove the fast scrambling conjecture.   We will see significant support in Theorem~\ref{theorsqrtlogN} for the conjecture that the true $\frac{1}{N}$ expansion of $C_{ij}(t)^2$ is resummable.

Although we have noted that on physical grounds the ``asymptotic" nature of (\ref{eq:theorFS}) is quite generic, one might ask whether the multiple inequalities invoked in the proof of Theorem \ref{theorFS} are responsible for the factor of $g!$ in (\ref{eq:theorFS}).  We now give a few heuristic arguments that this is unlikely to be the case.

(\emph{1}) First, we give an affirmative argument that there exist causal graphs arising from irreducible causal tree pairs which are genus $g=\mathrm{O}(N)$ which do not have a significant set of constraints on the ordering of factors.   Consider the irreducible causal tree pair sketched in Figure \ref{fig:highgenustree}, consisting of a binary tree expanding out from the target node $j$.  Even in a 2-local model, this factor subgraph may exist.   Now consider a random creeping ordering of the couplings on the left/right trees.  The probability that this random ordering is compatible with the causal graph is $\mathrm{O}(2^{-g})$.   As this holds for both the left and right graph, we conclude that the constraint of creeping order is not strong enough to remove $g!$.

\begin{figure}[t]
\centering
\includegraphics[width=3in]{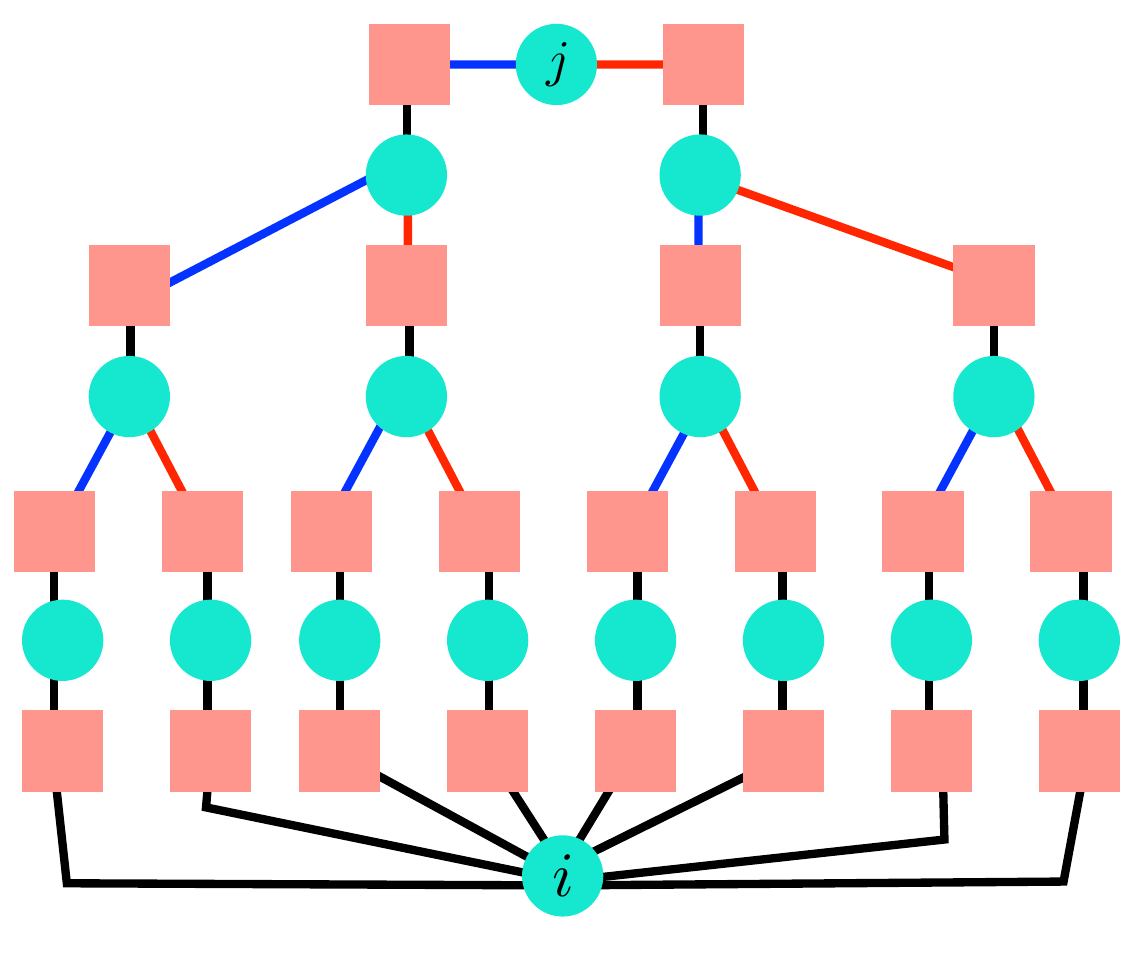}
\caption{A high genus causal graph formed by the merger of two binary trees.  Here $Q_{\mathrm{R}}$ takes the left path (blue) at a merger while $Q_{\mathrm{L}}$ takes the right path (red).  Edges in $E\cap(Q_{\mathrm{L}}\cup Q_{\mathrm{R}})$ are denoted in black.   It is possible to reach $g=\mathrm{O}(N)$ with such a construction.  The number of constraints on $\mathcal{M}_{\mathrm{L,R}}$ from this causal graph is only $\mathrm{O}(\mathrm{e}^g)$, not $\mathrm{O}(g!)$.}
\label{fig:highgenustree}
\end{figure}

(\emph{2}) Secondly, we argue that the oversimplification of allowing all left and right creeping sequences to ``weave" through each other does not add $g!$.  One way to see this is to demand that all couplings show up once to the left of $\mathbb{P}_j$, and once to the right of $\mathbb{P}_j$.  The sum over all such sequences would be smaller by a factor of $\mathrm{O}\left(2^{-\ell}\right)$.

%

\subsection{The Sachdev-Ye-Kitaev Model}
As an application of Theorem \ref{theorFS}, we now turn to the $q$-local Sachdev-Ye-Kitaev (SYK) random ensemble \cite{sachdevye, maldacena2016remarks, suh}.   Let $q=\mathrm{O}(N^0)$ be an even integer, and let $\psi_1,\ldots,\psi_N$ denote $N$ Majorana fermion operators with anticommutation relations
\begin{equation}
\lbrace \psi_i, \psi_j \rbrace = 2\mathbb{I}(i=j).  \label{eq:psiipsij}
\end{equation}
(\ref{eq:psiipsij}) implies that $\lVert \psi_i\rVert = 1$.  While thus far we have only worked with bosonic (commuting) degrees of freedom, we will explain shortly why our formalism straightforwardly generalizes.   The Hamiltonian is \begin{equation}
H = \mathrm{i}^{q/2}\sum_{i_1<i_2<\cdots<i_q} J_{i_1\cdots i_q} \psi_{i_1}\cdots \psi_{i_q}
\end{equation}
where $J_{i_1\cdots i_q}$ are independent, identically distributed zero-mean random variables of variance \begin{equation}
\mathbb{E}\left[{J_{i_1\cdots i_q}}^2\right] = \frac{(q-1)!}{2qN^{q-1}}\mathcal{J}^2.
\end{equation}
We also take $J_{i_1\cdots i_q}$ to be antisymmetric rank-$q$ tensors, for convenience.

It is straightforward to generalize our formalism for this model.  Consider the operator vector space spanned by \begin{equation}
|i_1\cdots i_m) := \mathrm{i}^{m/2}\psi_{i_1}\cdots \psi_{i_m}.
\end{equation}
Using the inner product (\ref{eq:innerproduct}) with $\dim(\mathcal{H}) = 2^{N/2}$,  for any subsets $X,Y\in \mathbb{Z}_2^V$, $(X|Y) = \mathbb{I}(X=Y)$.  For any permutation $\sigma \in \mathrm{S}_m$, $|i_1\cdots i_m) = \mathrm{sign}(\sigma) |\sigma(i_1)\cdots \sigma(i_m))$.   Using these properties, we conclude that (for example) \begin{equation}
\mathcal{L}_{j_1\cdots j_q}|j_qi_1\cdots i_m) = 2J_{j_1\cdots j_q}|j_1\cdots j_{q-1}i_1\cdots i_m).
\end{equation}
 Hence we obtain that the $|i_1\cdots i_m)$ form an orthonormal basis.
We conclude that $|X)$ for $X\in \mathbb{Z}_2^V$ span a real vector space.  As before, $\mathcal{L}_X$ is antisymmetric.   We define the projector \begin{equation}
(\mathcal{O}|\mathbb{P}_j|\mathcal{O}^\prime) := \frac{1}{2^{2+N/2}}\mathrm{tr}\left(\lbrace \mathcal{O},\psi_j\rbrace^\dagger\lbrace \mathcal{O}^\prime,\psi_j\rbrace\right),
\end{equation}
which continues to obey (\ref{eq:3simpleidentities}).  These properties are sufficient to use the formalism of this paper.

Let us first focus on the behavior of $\mathbb{E}[C_{ij}(t)^2] =  2^{-2-N/'2}\mathrm{tr}(\lbrace \psi_i(t),\psi_j\rbrace^2)$ at leading order in $1/N$: i.e. genus 0 causal graphs.  Since genus 0 causal graphs are irreducible paths from $i$ to $j$, they are easily resummed: assuming $i\ne j$, Theorem~\ref{theor4} implies that \begin{equation}
\mathbb{E}\left[C_{ij}(t)^2\right] \le \sum_{n=1}^\infty \left(\frac{t^n}{n!} \right)^{2} \times \left(\frac{2(q-1)!}{q N^{q-1}} \mathcal{J}^2\right)^n  \times \left(\frac{N^{q-2}}{(q-2)!}\right)^n N^{n-1} + \mathrm{O}\left(\frac{1}{N^2}\right).
\end{equation}
The first factor comes from the Taylor expansion of $\mathrm{e}^{\pm \mathcal{L}t}$; the second expansion comes from the Liouvillian and the average over random couplings; the third factor comes from the sum over all inequivalent irreducible paths containing $n$ factors on the factor graph.   Using (\ref{eq:factorialinequality}) we conclude that \begin{equation}
\mathbb{E}\left[C_{ij}(t)^2\right] \le \frac{1}{N}\cosh\left(2\sqrt{\frac{2(q-1)}{q}}\mathcal{J}t\right) + \mathrm{O}\left(\frac{1}{N^2}\right).  \label{eq:SYKbound}
\end{equation}

The exact answer is known analytically to leading order in $1/q$: \cite{stanford1802} 
\begin{equation}
\mathbb{E}\left[C_{ij}(t)^2\right]  = \frac{1}{N}\left[\cosh\left(2\mathcal{J}t\right) + \mathrm{O}\left(\frac{1}{q}\right)\right]+ \mathrm{O}\left(\frac{1}{N^2}\right).
\end{equation}
We conclude that our bound (\ref{eq:SYKbound}) is not tight: we overestimate the exponential growth rate by a factor of $\sqrt{2}$.  The origin of this effect is that the exact evaluation of $\mathbb{E}\left[C_{ij}(t)^2\right] $ contains terms such as $\mathcal{L}_{i_1\cdots i_q}^2 |\mathcal{O}) = -4J_{i_1\cdots i_q}^2 |\mathcal{O})$ which we have not been able to subtract out.  In more physical terms, the growing operator $|\psi_i(t))$ quantum mechanically interferes with itself, destructively, as it grows.   This destructive interference is not captured by either Theorem \ref{theor3} or Theorem \ref{theor4}.

As far as we know, (\ref{eq:SYKbound}) is the first analytic constraint on OTOCs at infinite temperature in the SYK model at finite $q$, even at leading order in $1/N$.  Figure \ref{fig:syk} compares our predictions for the $q$-dependence of the growth rate to numerical calculations of the result \cite{stanford1802}, and our exact bound is indeed consistent with their result.    

\begin{figure}[t]
\centering
\includegraphics{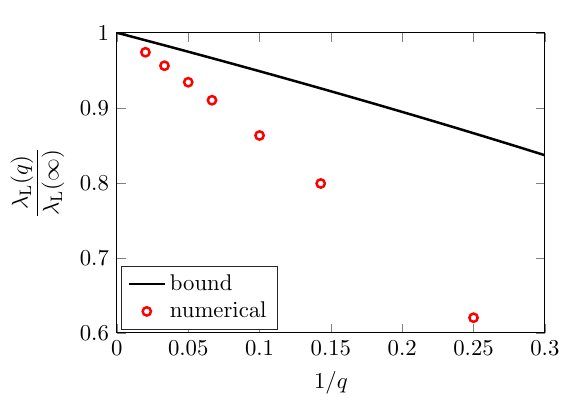}
\caption{Comparing the ratio of Lyapunov exponent $\lambda_*$ at finite $q$ to large $q$ in the SYK model.  (\ref{eq:SYKbound}) is shown in black, and the numerical results of \cite{stanford1802} are shown in red.}
\label{fig:syk}
\end{figure}

Of course, we have also provided bounds to all orders in $1/N$.  Applying Theorem~\ref{theorFS} to the SYK model, we conclude that the infinite temperature Lyapunov exponent, at all orders in $\frac{1}{N}$, is bounded by \begin{equation}
\lambda_* \le 48\sqrt{\frac{(q-1)}{q}}\mathcal{J}.
\end{equation}
When $q$ is large, this differs from the leading order Lyapunov exponent by a factor of $24$.  In other words, upon including all orders $g\le m$ in perturbation theory in the large $N$ limit, for fixed $m$, the operator growth time is at least 4\% of the leading order prediction in $\frac{1}{N}$.  

\subsection{Improved Non-Perturbative Bound on the Scrambling Time}\label{sec:sqrtlogN}
Theorem \ref{theorFS} demonstrates (\ref{eq:susskind}) to any order $N^{-g}$ with $g=\mathrm{O}(N^{1-\epsilon})$ with $\epsilon>0$.  However, the non-perturbative bound on $t_{\mathrm{s}}^\delta$ from Theorem \ref{theorFS} is $t_{\mathrm{s}}^\delta = \mathrm{\Omega}(N^0)$, as discussed in (\ref{eq:gfactorialconvergence}).   The following theorem improves on this result.

\begin{theor}
Let $G \in \mathrm{G}_V(m,q,k)$ be drawn from the Erd\"os-R\'enyi ensemble.  Let $H$ be a random Hamiltonian drawn from a simple random Hamiltonian ensemble on $G$, obeying (\ref{eq:theor4J}).  Then 
\begin{equation}
t_{\mathrm{s}}^\delta = \mathrm{\Omega}\left(\sqrt{\log N}\right) \label{eq:theorsqrtlogN}
\end{equation}
holds almost surely as $N\rightarrow \infty$.
\label{theorsqrtlogN}
\end{theor}
\begin{proof}
We prove this theorem as follows. (\emph{1}) First we consider an alternative classification of sequences in $\mathbb{M}^2_{ji}$ not by irreducible causal tree pair, but simply by irreducible paths in the left/right sequences $(\mathcal{M}_{\mathrm{L}},\mathcal{M}_{\mathrm{R}})\in \mathbb{M}^2_{ji}$.  (\emph{2}) We study the problem on the $\mathrm{K}_N^q$.    Using the fact that each coupling must show up twice in the disorder-averaged OTOC, and that the factors along the left/right irreducible path must show up in that order at least once, we obtain a non-perturbative bound on $\mathbb{E}[C_{ij}(t)^2]$.  (\emph{3}) We argue that, analogously to Corollary~\ref{corolFS}, our bound also holds upon averaging over an Erd\"os-R\'enyi factor graph ensemble.  We then use Markov's inequality to find (\ref{eq:theorsqrtlogN}).

\emph{Step 1:}  This step introduces a new formalism much along the lines of Sections \ref{sec:causaltrees} and \ref{sec:causalgraph}.   We will state many facts without proof, as their justification is straightforward and follows our earlier results.

  Let $(\mathcal{M}_{\mathrm{L}},\mathcal{M}_{\mathrm{R}}) \in \mathbb{M}^2_{ji}$.  Let $T_{\mathrm{L,R}} := T(\mathcal{M}_{\mathrm{L,R}})$ denote the left/right causal tree, and let $\Gamma_{\mathrm{L,R}}$ denote the irreducible paths from $i$ to $j$ in $T_{\mathrm{L,R}}$, as defined via the equivalence class on $\mathcal{T}_{ji}$ of Section \ref{sec:causaltrees}:  $[\Gamma_{\mathrm{L,R}}] = [T_{\mathrm{L,R}}]$.   Then define the following equivalence relation $\sim^\prime_{ji}$ on $\mathbb{M}^2_{ji}$:  $(\mathcal{M}_{\mathrm{L}},\mathcal{M}_{\mathrm{R}})\sim^\prime_{ji} (\mathcal{M}_{\mathrm{L}}^\prime,\mathcal{M}_{\mathrm{R}}^\prime)$ if and only if $[T(\mathcal{M}_{\mathrm{L}})]=[T(\mathcal{M}_{\mathrm{L}}^\prime)]$ and $[T(\mathcal{M}_{\mathrm{R}})]=[T(\mathcal{M}_{\mathrm{R}}^\prime)]$.  The unique irreducible element of this equivalence class is $(\Gamma_{\mathrm{L}},\Gamma_{\mathrm{R}})$.   As in Section \ref{sec:causaltreepair} there are many sequences $\mathcal{M}$ which lead to the same irreducible path pair. We define $\mathcal{S}^{\prime 2}_{ji} = \mathcal{T}^2_{ji}/\sim^\prime_{ji}$.   Figure~\ref{fig:vpsi48} gives an example of an irreducible path pair.
  
  \begin{figure}[t]
\centering
\includegraphics[width=0.85\textwidth]{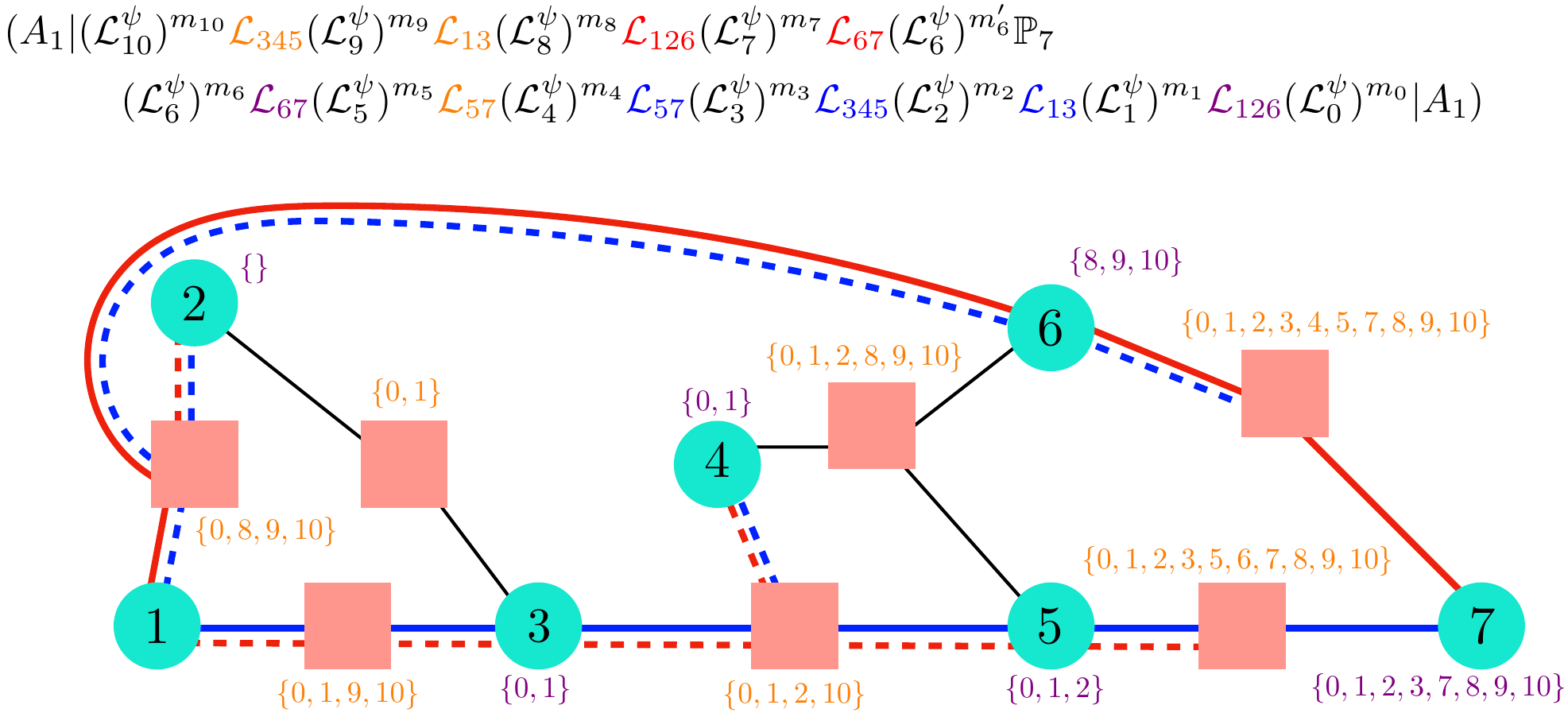}
\caption{An irreducible path pair $(\Gamma_{\mathrm{L}},\Gamma_{\mathrm{R}})\in\mathcal{S}^{\prime 2}_{ji}$, which follows from the sequence depicted (identical to that in Figure~\ref{fig:vpsi42}). $\Gamma_{\mathrm{L}}$ is shown in red and $\Gamma_{\mathrm{R}}$ is shown in blue; the rightmost appearance of  $X\in \Gamma_{\mathrm{L}}$ is shown in purple, and the leftmost appearance of  $X\in \Gamma_{\mathrm{R}}$ is shown in orange.  Sets in purple next to each node $v\in V$ denote values of $k$ for which $v\in V^\psi_k$; sets in orange next to each factor $X\in F$ denote values of $k$ for which $X\in Y^\psi_k$.  Note that $\psi$ is not creeping from left to right.}
\label{fig:vpsi48}
\end{figure}

We define the set \begin{equation}
\Psi^\prime(\Gamma_{\mathrm{L}},\Gamma_{\mathrm{R}}) := \lbrace \mathcal{M}\in\mathbb{M}^2_{ji}: \mathcal{M}\sim^\prime_{ji} (\Gamma_{\mathrm{L}},\Gamma_{\mathrm{R}}) \text{ and } F\cap \mathcal{M} = F\cap (\Gamma_{\mathrm{L}}\cup\Gamma_{\mathrm{R}}) \rbrace .
\end{equation}
For $\psi \in \Psi^\prime(\Gamma_{\mathrm{L}},\Gamma_{\mathrm{R}})$, we define the set \begin{align}
V^{\prime \psi}_k &:= \left\lbrace\begin{array}{ll} V^{\Gamma_{\mathrm{R}}}_p &\  \min_\psi(X^{\Gamma_{\mathrm{R}}}_p)\le k < \min_\psi(X^{\Gamma_{\mathrm{R}}}_{p+1}) \\ 
\lbrace j\rbrace &\ \min_\psi(X^{\Gamma_{\mathrm{R}}}_{\ell_{\mathrm{R}}-1}) \le k < \min_\psi(X^{\Gamma_{\mathrm{R}}}_{\ell_{\mathrm{R}}}), \text{ or } \max_\psi(X^{\Gamma_{\mathrm{L}}}_{\ell_{\mathrm{L}}}) \le k < \min_\psi(X^{\Gamma_{\mathrm{L}}}_{\ell_{\mathrm{L}}-1}) \\
V^{\Gamma_{\mathrm{L}}}_p &\  \min_\psi(X^{\Gamma_{\mathrm{L}}}_{p+1})\le k < \min_\psi(X^{\Gamma_{\mathrm{L}}}_{p}) \\ \emptyset &\ \min_\psi(j) \le k < \max_\psi(j) \end{array} \right. 
\end{align}
and the set \begin{equation}
Y^{\prime \psi}_k := \lbrace X\in F : X\cap V^{\prime \psi}_k \ne \emptyset\rbrace \cup \lbrace X\in \psi : \min_\psi(X) > k \rbrace \cup \lbrace X\in \psi : \max_\psi(X) \le k \rbrace .
\end{equation}
The generalization of Theorem~\ref{theor4} is that \begin{align}
\mathbb{E}&\left[C_{ij}(t)^2\right] = \mathbb{E}\left[\sum_{(\Gamma_{\mathrm{L}},\Gamma_{\mathrm{R}})\in \mathcal{S}^{\prime 2}_{ji}} \sum_{\psi \in \Psi^\prime(\Gamma_{\mathrm{L}},\Gamma_{\mathrm{R}})} \int\limits_{\mathrm{\Delta}^{\ell_{\mathrm{L}}}(t)} \mathrm{d}t^{\mathrm{L}}_1\cdots \mathrm{d}t^{\mathrm{L}}_{\ell_{\mathrm{L}}} (A_i| \mathrm{e}^{-\mathcal{L}^\psi_\ell (t-t^{\mathrm{L}}_{\ell_{\mathrm{L}}})}(-\mathcal{L}^\psi_{X_\ell}) \mathrm{e}^{-\mathcal{L}^\psi_\ell (t^{\mathrm{L}}_{\ell_{\mathrm{L}}}-t^{\mathrm{L}}_{\ell_{\mathrm{L}}-1})} \cdots  \right. \notag \\
&\left. \;\;\; \times (-\mathcal{L}^\psi_{X_{\ell_{\mathrm{R}}+1}}) \mathrm{e}^{-\mathcal{L}^\psi_{\ell_{\mathrm{R}}}t_1^{\mathrm{L}}}\mathbb{P}_j  \int\limits_{\mathrm{\Delta}^{\ell_{\mathrm{R}}}(t)} \mathrm{d}t^{\mathrm{R}}_1\cdots \mathrm{d}t^{\mathrm{R}}_{\ell_{\mathrm{R}}} \mathbb{P}_j \mathrm{e}^{\mathcal{L}^\psi_{X_{\ell_{\mathrm{R}}}}(t-t^{\mathrm{R}}_{\ell_{\mathrm{R}}})} \mathcal{L}^\psi_{\ell_{\mathrm{R}}} \cdots \mathrm{e}^{\mathcal{L}^\psi_1(t_2^{\mathrm{R}}-t_1^{\mathrm{R}})} \mathcal{L}^\psi_{X_1} \mathrm{e}^{\mathcal{L}^\psi_0t_1^{\mathrm{R}}} |A_i) \right] 
\end{align}
where \begin{equation}
\mathcal{L}^{\prime \psi}_k := \mathcal{L} - \sum_{Y\in Y^{\prime\psi}_k}\mathcal{L}_Y.
\end{equation}
Figure \ref{fig:vpsi48} shows the construction of $V^{\prime\psi}_k$ and $Y^{\prime\psi}_k$ in an example.  Using (\ref{eq:factorialinequality}) we arrive at the bound \begin{equation}
\mathbb{E}\left[C_{ij}(t)^2\right] \le \sum_{(\Gamma_{\mathrm{L}},\Gamma_{\mathrm{R}})\in\mathcal{S}^{\prime2}_{ji}} |\Psi^\prime(\Gamma_{\mathrm{L}},\Gamma_{\mathrm{R}})| \frac{(2t)^\ell}{\ell!} \prod_{X\in \psi}(2\mathcal{J}_X) \label{eq:EC47}
\end{equation}
Here, the total number of factors in the sequence $(\mathcal{M}_{\mathrm{L}},\mathcal{M}_{\mathrm{R}}$ is denoted by $n$.  This is slightly different notation from in the proof of Theorem~\ref{theorFS}.   We will also use the different notation $\ell_{\mathrm{L,R}} := \ell(\Gamma_{\mathrm{L,R}})$ below.

\emph{Step 2:}  Now we proceed by first analyzing (\ref{eq:EC47}) on $\mathrm{K}_N^q$.    First, let us fix $(\Gamma_{\mathrm{L}},\Gamma_{\mathrm{R}})\in\mathcal{S}^{\prime2}_{ji}$, and suppose that $\Gamma_{\mathrm{L}}\cap \Gamma_{\mathrm{R}}$ contains $n\le \min(\ell_{\mathrm{L}},\ell_{\mathrm{R}})$ factors in common.   If this inequality is saturated, it means that the graph is genus 0 ($\Gamma_{\mathrm{L}}=\Gamma_{\mathrm{R}}$), and $|\Psi^\prime(\Gamma,\Gamma)|=1$.   Otherwise, there are $\ell_{\mathrm{L}}-n > 0$ free factors on the left and $\ell_{\mathrm{R}}-n > 0 $ free factors on the right.   We can overestimate the size of $|\Psi^\prime(\Gamma_{\mathrm{L}},\Gamma_{\mathrm{R}})| $ by assuming that the non-shared factors can appear anywhere, while (by construction) the factors in the irreducible paths must show up in a fixed order.  Using the multinomial theorem: \begin{equation}
|\Psi^\prime(\Gamma_{\mathrm{L}},\Gamma_{\mathrm{R}})| \le \frac{(2\ell_{\mathrm{L}}+2\ell_{\mathrm{R}}-2n)!}{\ell_{\mathrm{L}}!\ell_{\mathrm{R}}!} = \frac{\ell!}{\ell_{\mathrm{L}}!\ell_{\mathrm{R}}!} .  \label{eq:irrpathpsi}
\end{equation}

Now let us plug (\ref{eq:irrpathpsi}) into (\ref{eq:EC47}), and split the sum into two terms depending on whether or not the irreducible paths are the same or not: \begin{equation}
\mathbb{E}\left[C_{ij}(t)^2\right] \le \sum_{(\Gamma,\Gamma)\in\mathcal{S}^{\prime2}_{ji}} \frac{(2t)^\ell}{\ell!} \prod_{X\in \psi}(2\mathcal{J}_X)
 + \sum_{(\Gamma_{\mathrm{L}},\Gamma_{\mathrm{R}})\in\mathcal{S}^{\prime2}_{ji}, \Gamma_{\mathrm{L}}\ne \Gamma_{\mathrm{R}}}  \frac{(2t)^\ell}{\ell_{\mathrm{L}}!\ell_{\mathrm{R}}!} \prod_{X\in \psi}(2\mathcal{J}_X).\label{eq:472terms}
\end{equation}
The first term in the above sum is easy to evaluate.  We evaluate the sum by first summing over all inequivalent paths $\Gamma$ of fixed length $\ell$, and then summing over $\ell$:\begin{align}
\sum_{(\Gamma,\Gamma)\in\mathcal{S}^{\prime2}_{ji}} \frac{(2t)^\ell}{\ell!} \prod_{X\in \psi}(2\mathcal{J}_X) &\le \sum_{\ell=1}^\infty \sum_{\substack{\Gamma \in \mathcal{T}_{ji}\\ \ell(\Gamma) = m}} \frac{(2t)^{2m}}{(2m)!} \prod_{X\in \psi}(2\mathcal{J}_X) = \frac{1}{N}\sum_{m=1}^\infty \frac{(2t)^{2m}}{(2m)!} \left(\frac{q-1}{q}4 \mathcal{J}^2\right)^{m} \notag \\
&= \frac{1}{N}\left[\cosh\left(4\sqrt{\frac{q-1}{q}}\mathcal{J}t\right) - 1\right].
\end{align}
In the last equality in the first line, we have used identical counting of the number of independent paths as in Theorem~\ref{theorFS}.  Indeed, at genus 0, this bound and the bound of Theorem~\ref{theor4} are (by construction) identical.\footnote{The genus 0 term in (\ref{eq:theorFS}) is not as strong as this bound, a choice which was deliberately made to avoid cumbersome formulas in Theorem~\ref{theorFS} and its proof.}

Now we turn to the second sum.  Without loss of generality, we assume that $\ell_{\mathrm{L}}\ge \ell_{\mathrm{R}}$.  It is useful to write \begin{align}
 \sum_{(\Gamma_{\mathrm{L}},\Gamma_{\mathrm{R}})\in\mathcal{S}^{\prime2}_{ji}, \Gamma_{\mathrm{L}}\ne \Gamma_{\mathrm{R}}}  &\frac{(2t)^\ell}{\ell_{\mathrm{L}}!\ell_{\mathrm{R}}!} \prod_{X\in \psi}(2\mathcal{J}_X) \le 2\sum_{\ell_{\mathrm{L}}=1}^\infty \sum_{\ell_{\mathrm{R}}=1}^{\ell_{\mathrm{L}}} \sum_{n=0}^{\ell_{\mathrm{R}}-1} 
\sum_{\substack{(\Gamma_{\mathrm{L}},\Gamma_{\mathrm{R}})\in \mathcal{S}^{\prime 2}_{ji}\\ \ell(\Gamma_{\mathrm{L,R}})=\ell_{\mathrm{L,R}} \\ \widetilde{F}(\Gamma_{\mathrm{L}}\cap\Gamma_{\mathrm{R}})=n}} \frac{(2t)^{2(\ell_{\mathrm{L}}+\ell_{\mathrm{R}}-n)}}{\ell_{\mathrm{L}}!\ell_{\mathrm{R}}!} \left(4\mathcal{J}_X^2\right)^{\ell_{\mathrm{L}}+\ell_{\mathrm{R}}-n} \notag \\
&\le 2\sum_{\ell_{\mathrm{L}}=1}^\infty \sum_{\ell_{\mathrm{R}}=1}^{\ell_{\mathrm{L}}} \sum_{n=0}^{\ell_{\mathrm{R}}-1} \frac{(2t)^{2(\ell_{\mathrm{L}}+\ell_{\mathrm{R}}-n)}}{\ell_{\mathrm{L}}!\ell_{\mathrm{R}}!} \left(\frac{q-1}{q}4\mathcal{J}^2\right)^{\ell_{\mathrm{L}}+\ell_{\mathrm{R}}-n} \times \frac{(q-1)\ell_{\mathrm{R}}}{N^2}.
\end{align}
The overall factor of 2 comes from our simplifying assumption.  The new factor of $(q-1)$ the second line above comes from the fact that if $\Gamma_{\mathrm{L}}\ne\Gamma_{\mathrm{R}}$, the two paths must join up at some node or factor at least once.   There are $n+1$ possible places where the paths can join up.  If the paths join up at a node, then there is a relative factor of $\frac{1}{N}$; if they join up at a factor, there is also a relative factor of $\frac{q-2}{N}$.   This leads to the bound above.   We now simply carry out the sums in order: \begin{align}
\sum_{(\Gamma_{\mathrm{L}},\Gamma_{\mathrm{R}})\in\mathcal{S}^{\prime2}_{ji}, \Gamma_{\mathrm{L}}\ne \Gamma_{\mathrm{R}}}  &\frac{(2t)^\ell}{\ell_{\mathrm{L}}!\ell_{\mathrm{R}}!} \prod_{X\in \psi}(2\mathcal{J}_X) \le \frac{4}{N^2} \left(\exp\left[\frac{q-1}{q}(4\mathcal{J}t)^2\right]-1\right) \sum_{\ell_{\mathrm{R}}=m}^\infty \sum_{m=1}^\infty \left(\frac{q-1}{q}(4\mathcal{J}t)^2\right)^m \frac{1}{\ell_{\mathrm{R}}!} \notag \\
&\le \frac{2(q-1)(\mathrm{e}-1)}{N^2}\left(\exp\left[\frac{q-1}{q}(4\mathcal{J}t)^2\right]-1\right)\left(\frac{q-1}{q}(4\mathcal{J}t)^2\right)\exp\left[\frac{q-1}{q}(4\mathcal{J}t)^2\right].\label{eq:47last}
\end{align}

\emph{Step 3:} It is straightforward to generalize the proof of Corollary~\ref{corolFS}, and we find \begin{align}
\mathbb{E}_{q,k}\left[\mathbb{E}_G\left[C_{ij}(t)^2\right]\right] &\le \frac{1}{N}\left[\cosh\left(4\sqrt{\frac{q-1}{q}}\mathcal{J}t\right) - 1\right] \notag \\
&+ \frac{2(q-1)(\mathrm{e}-1)}{N^2}\left(\exp\left[\frac{q-1}{q}(4\mathcal{J}t)^2\right]-1\right)\left(\frac{q-1}{q}(4\mathcal{J}t)^2\right)\exp\left[\frac{q-1}{q}(4\mathcal{J}t)^2\right]. \label{eq:47qk}
\end{align} Finally, we may bound $t_{\mathrm{s}}^\delta$.   Let $\mathbb{E}[\cdots] = \mathbb{E}_{q,k}[\mathbb{E}_G[\cdots]]$.
 Using Markov's inequality,\begin{equation}
\mathbb{P}\left[C_{ij}(t)^2 > \delta^2 \right] \le \frac{\mathbb{E}[C_{ij}(t)^2]}{\delta^2}
\end{equation}
 From (\ref{eq:47qk}), there exist constants $c_1,c_2,c_3,c_4 = \mathrm{O}(N^0)$ such that \begin{equation}
\mathbb{E}\left[C_{ij}(t)^2\right] \le \frac{c_1}{N}\mathrm{e}^{c_2t} + \frac{c_3}{N^2}\mathrm{e}^{c_4t^2}.
\end{equation}
Letting $t=\sqrt{a\log N}$, we find that as $N\rightarrow \infty$, \begin{equation}
 \mathbb{E}\left[C_{ij}(t)^2\right] \le 2c_3 N^{ac_4-2} 
\end{equation}
Hence for $ac_4<2$, as $N\rightarrow \infty$, $\mathbb{P}\left[C_{ij}(t) > \delta\right]  \rightarrow 0$.  Thus $t_{\mathrm{s}}^\delta \ge \sqrt{c_4^{-1}\log N}$ almost surely.  Hence we obtain (\ref{eq:theorsqrtlogN}).
\end{proof}

\section{Epilogue}
\AC{
This work was under review during a time period of extremely rapid developments in the mathematical theory of quantum information dynamics, some of which follows directly from the results presented above.   In what follows, we will outline some recent accomplishments, together with remaining open questions.

(\emph{1}) The infinite temperature fast scrambling conjecture, which was not proven in this paper, has been demonstrated in certain models. One of us first proved this by developing a ``quantum walk" formalism for bounding OTOCs \cite{Lucas_2020}; the other more directly generalized Theorem~\ref{theor3} to random Hamiltonians by using techniques from matrix martingale theory \cite{chen2021concentration}.  Each of these techniques is broadly generalizable, and some of these generalizations will be discussed below.

(\emph{2}) There are many systems which are not random but for which $H_X$ are strongly extensive.  An example related to the SYK model is a more general class of melonic models \cite{witten, gurau, klebanov, gubser}, all of which exhibit similar correlation functions in the large $N$ limit.   We have proven the fast scrambling conjecture in a toy model of these systems as well \cite{osborne}.

(\emph{3}) It would be interesting to try and generalize the inner product (\ref{eq:innerproduct}) to a thermal inner product \cite{lucas1809}.  Of particular interest is the presence of a ``bound on chaos" \cite{stanfordbound} which (under certain circumstances) bounds the Lyapunov exponent from above by $2\mpi T$, where $T$ is the temperature.  It would be remarkable if this result (which follows from analytic properties of correlators in the complex plane) can be understood from a graph theoretic perspective.  Recent work \cite{alex1811} on the SYK model may be of relevance.  A Lieb-Robinson perspective on thermal commutator norms can be found in \cite{xizhi}.   We note in passing that one of the two assumptions of the chaos bound of \cite{stanfordbound} is quite similar to $C_{ij}(t)$ being parametrically small at finite time $t$, a fact which we have proven in a broad class of models in this paper.   Some initial progress along these lines has been made in \cite{xiaochen}, which studied systems at finite chemical potential (but infinite temperature), and demonstrated a provable and generic slow down of quantum dynamics in low-density systems with a conserved charge.

(\emph{4}) Many physically relevant systems involve bosonic degrees of freedom and thus necessarily have an infinite dimensional Hilbert space.  It would be interesting to study scrambling in such systems \cite{yoshida}.  Lieb-Robinson techniques for infinite-dimensional Hilbert spaces have been developed in \cite{sims}.  It is postulated that a matrix model with bosonic degrees of freedom is dual to dynamical quantum gravity in higher dimensions \cite{bfss}, and so the developments of our formalism to infinite dimensional Hilbert spaces may be important in order to provide some rigorous constraints on scrambling in matrix model formulations of quantum gravity and string theory.

(\emph{5}) 
As we discussed at the start of the paper, our work may imply that certain theories of quantum gravity that admit a holographic description are not spoiled by non-perturbative effects in quantum gravity.  It would be interesting to clarify any such implications, even if certain assumptions cannot be proven.  A preliminary (yet challenging) task would be to study whether our bound correctly reproduces the qualitative features of subleading corrections in $\frac{1}{N}$ to OTOCs in the SYK model.

(\emph{6}) At early times, we can explicitly construct a Hamiltonian on a one dimensional lattice for which (for certain commutators) Theorem~\ref{theor3} is exact at early times.  Let $X,Y,Z$ be the Pauli matrices (a basis for $\mathfrak{su}(2)$).  Then consider the Hamiltonian \begin{equation}
H = \sum_{i\in2\mathbb{Z}} \left(X_i X_{i+1} + Y_i Y_{i-1}\right),  \label{eq:theor3tight}
\end{equation}
and the commutator norm $\lVert [Z_i(t),Z_j]\rVert$.  It is straightforward to see that Theorem~\ref{theor3} is a tight bound as $t\rightarrow 0$ for any $i$ and $j$.    More generally, finding necessary and/or sufficient conditions (if they exist) for the optimality of Theorem~\ref{theor3} and Theorem~\ref{theor4} for any $i$ and $j$, and for a fixed duration of time $t>0$, is an important future task.  Perhaps a necessary condition for the optimality of our bounds away from $t\rightarrow 0$ is that the factor graph is locally treelike.  In the SYK model, our bounds do not correctly capture the destructive interference of a growing operator with itself, and we expect such interference to be generic on loopy factor graphs.

(\emph{7}) Similarly, it would be interesting to check if Theorem~\ref{theor3} and/or Theorem~\ref{theor4} ever correctly reproduces the speed of light in lattice discretizations of Lorentz-invariant quantum field theory.  

(\emph{8})  If we replace Hamiltonian quantum evolution in a simple random ensemble with Brownian Hamiltonian evolution \cite{lashkari, lucas1903}, where each coupling constant $J_X(t)$ is an independent Brownian motion, we expect that the genus expansion of Theorem~\ref{theor4} truncates at leading order (genus 0).  It is an important question whether the recent models of quantum dynamics employing random unitary circuits \cite{nahum, tibor} (which we anticipate are similar in most respects to Brownian evolution) therefore miss any qualitatively important aspects of quantum dynamics in simple random Hamiltonian ensembles (as they do in the specialized model of \cite{lucas1903}).  Progress along these lines appears in \cite{chen2021concentration}, which combined theorem~\ref{theor3} with tool from matrix martingales.

(\emph{9}) As this paper is rigorous, the high genus terms of the topological expansion in Theorem~\ref{theor4} made it very challenging to prove (\ref{eq:susskind}) for regular systems.  It is plausible that for irreducible causal tree pairs whose causal graph is of high genus,  summing over the many creeping sequences of couplings lead to destructive interference in operator growth.  It would be worthwhile to check if this interference can be proved in simple models.  These results may have important implications about holographic quantum gravity.

(\emph{10}) It would be interesting if our topological generalization of the interaction picture of quantum mechanics has an elegant interpretation in a path integral formulation of quantum mechanics.

(\emph{11}) A notoriously challenging problem in mathematical physics, which we did not touch on in this paper, is the study of systems with power-law interactions.  The traditional Lieb-Robinson theorem \cite{hastings} says that commutators between operators separated by distance $r$ can grow large in a time $t\sim \log r$. Numerous recent works \cite{Tran_2019_polyLC,alpha_3_chenlucas,strictlylinear_KS,tran2021optimal} have ultimately proved the existence of a linear light cone in operator norm for sufficiently large $\alpha$, and demonstrated the optimality of the resulting bounds.  In particular ~\cite{alpha_3_chenlucas} directly followed the methods proposed in this paper, choosing an alternative set of equivalence classes and irreducible paths to present the first proof that a ``Lieb-Robinson velocity" was finite in one-dimensional models with power-law interactions.   Interestingly, the shape of Frobenius light cone is qualitatively different  from the operator norm~\cite{hierarachy,Kuwahara_OTOC, chen2021optimal}. 

(\emph{12}) Recently, an alternative notion of ``operator complexity growth" has been proposed \cite{altman,avdoshkin}.  In this framework, one defines a basis of operators by starting with (e.g.) a single Pauli matrix $X$, and then defining a basis of non-orthogonalized vectors $X, [H,X], [H,[H,X]]$, etc.   Studying how quickly an operator ``grows" in this basis seems to be able to probe scrambling in an alternative way to the more conventional OTOCs which we have studied in this paper.  It would be interesting to try and relate our formalism to this operator complexity growth picture in future work.

(\emph{13}) Finally, the Lieb-Robinson theorem has had other applications: for example, proving that spatial correlation functions decay exponentially in gapped systems on lattice graphs \cite{hastings}.  It would be interesting if our stronger results lead to qualitatively stronger theorems in quantum many-body physics which at first appear unrelated to operator growth and many-body chaos. 
}
\addcontentsline{toc}{section}{Acknowledgments}
\section*{Acknowledgments}
This work was supported by the Gordon and Betty Moore Foundation's EPiQS Initiative through Grant GBMF4302, by a Research Fellowship from the Alfred P. Sloan Foundation through Grant FG-2020-13795, and by the Air Force Office of Scientific Research through Grant FA9550-21-1-0195.

\bibliographystyle{unsrt}
\addcontentsline{toc}{section}{References}
\bibliography{randomLR}

\begin{thebibliography}{10}

\bibitem{susskind08}
Y.~Sekino and L.~Susskind.
\newblock ``Fast scramblers",
  \href{https://doi.org/10.1088/1126-6708/2008/10/065}{\textsl{Journal of High
  Energy Physics} \textbf{10} \texttt{065} (\textbf{2008})},
  \href{http://arxiv.org/abs/0808.2096}{\texttt{arXiv:0808.2096}}.

\bibitem{hayden07}
P.~Hayden and J.~Preskill.
\newblock ``Black holes as mirrors: quantum information in random subsystems",
  \href{https://doi.org/10.1088/1126-6708/2007/09/120}{\textsl{Journal of High
  Energy Physics} \textbf{09} \texttt{120} (\textbf{2007})},
  \href{http://arxiv.org/abs/0708.4025}{\texttt{arXiv:0708.4025}}.

\bibitem{dieks}
D.~Dieks.
\newblock ``Communication by EPR devices",
  \href{https://doi.org/10.1016/0375-9601(82)90084-6}{\textsl{Physics Letters}
  \textbf{A92} 271 (1982)}.

\bibitem{zurek}
W.~Wootters and W.~Zurek.
\newblock ``A single quantum cannot be cloned",
  \href{https://doi.org/10.1038/299802a0}{\textsl{Nature} \textbf{299} 802
  (1982)}.

\bibitem{lucas1805}
G.~Bentsen, Y.~Gu, and A.~Lucas.
\newblock ``Fast scrambling on sparse graphs",
  \href{https://doi.org/10.1073/pnas.1811033116}{\textsl{Proceedings of the
  National Academy of Sciences} \textbf{116} 6689 (2019)},
  \href{http://arxiv.org/abs/1805.08215}{\texttt{arXiv:1805.08215}}.

\bibitem{Garttner2017}
M.~G\"arttner, J.~G. Bohnet, A.~Safavi-Naini, M.~L. Wall, J.~J. Bollinger, and
  A.~M. Rey.
\newblock ``Measuring out-of-time-order correlations and multiple quantum
  spectra in a trapped-ion quantum magnet",
  \href{https://doi.org/10.1038/nphys4119}{\textsl{Nature Physics} \textbf{13}
  781 (2017)},
  \href{http://arxiv.org/abs/1608.08938}{\texttt{arXiv:1608.08938}}.

\bibitem{Li2017}
J.~Li, R.~Fan, H.~Wan, B.~Ye, B.~Zeng, H.~Zhai, X.~Peng, and J.~Du.
\newblock ``Measuring out-of-time-ordered correlators on a nuclear magnetic
  resonance quantum simulator",
  \href{https://doi.org/10.1103/PhysRevX.7.031011}{\textsl{Physical Review}
  \textbf{X7} \texttt{031011} (2017)},
  \href{http://arxiv.org/abs/1609.01246}{\texttt{arXiv:1609.01246}}.

\bibitem{shenker13}
S.~H. Shenker and D.~Stanford.
\newblock ``Black holes and the butterfly effect",
  \href{https://doi.org/10.1007/JHEP03(2014)067}{\textsl{Journal of High Energy
  Physics} \textbf{03} \texttt{067} (\textbf{2014})},
  \href{http://arxiv.org/abs/1306.0622}{\texttt{arXiv:1306.0622}}.

\bibitem{lashkari}
N.~Lashkari, D.~Stanford, M.~Hastings, T.~Osborne, and P.~Hayden.
\newblock ``Towards the fast scrambling conjecture",
  \href{https://doi.org/10.1007/JHEP04(2013)022}{\textsl{Journal of High Energy
  Physics} \textbf{04} \texttt{022} (\textbf{2013})},
  \href{http://arxiv.org/abs/1111.6580}{\texttt{arXiv:1111.6580}}.

\bibitem{liebrobinson}
E.~H. Lieb and D.~Robinson.
\newblock ``The finite group velocity of quantum spin systems",
  \href{https://link.springer.com/article/10.1007/BF01645779}{\textsl{Communications
  in Mathematical Physics} \textbf{28} 251 (1972)}.

\bibitem{poulin}
D.~Poulin.
\newblock ``Lieb-Robinson bound and locality for general Markovian quantum
  dynamics",
  \href{https://doi.org/10.1103/PhysRevLett.104.190401}{\textsl{Physical Review
  Letters} \textbf{104} \texttt{190401} (2010)},
  \href{http://arxiv.org/abs/1003.3675}{\texttt{arXiv:1003.3675}}.

\bibitem{hastings}
M.~Hastings and T.~Koma.
\newblock ``Spectral gap and exponential decay of correlations",
  \href{https://doi.org/10.1007/s00220-006-0030-4}{\textsl{Communications in
  Mathematical Physics} \textbf{265} 781 (2006)},
  \href{http://arxiv.org/abs/math-ph/0507008}{\texttt{arXiv:math-ph/0507008}}.

\bibitem{dyson}
F.~J. Dyson.
\newblock ``Divergence of perturbation theory in quantum electrodynamics",
  \href{https://doi.org/10.1103/PhysRev.85.631}{\textsl{Physical Review}
  \textbf{85} 531 (1952)}.

\bibitem{sachdevye}
S.~Sachdev and J.~Ye.
\newblock ``Gapless spin-fluid ground state in a random quantum Heisenberg
  magnet",
  \href{http://journals.aps.org/prl/abstract/10.1103/PhysRevLett.70.3339}{\textsl{Physical
  Review Letters} \textbf{70} 3339 (1993)},
  \href{http://arxiv.org/abs/cond-mat/9212030}{\texttt{arXiv:cond-mat/9212030}}.

\bibitem{maldacena2016remarks}
J.~Maldacena and D.~Stanford.
\newblock ``Comments on the Sachdev-Ye-Kitaev model",
  \href{http://journals.aps.org/prd/abstract/10.1103/PhysRevD.94.106002}{\textsl{Physical
  Review} \textbf{D94} \texttt{106002} (2016)},
  \href{http://arxiv.org/abs/1604.07818}{\texttt{arXiv:1604.07818}}.

\bibitem{suh}
A.~Kitaev and S.~J. Suh.
\newblock ``The soft mode in the Sachdev-Ye-Kitaev model and its gravity dual",
  \href{https://doi.org/10.1007/JHEP05(2018)183}{\textsl{Journal of High Energy
  Physics} \textbf{05} \texttt{183} (\textbf{2018})},
  \href{http://arxiv.org/abs/1711.08467}{\texttt{arXiv:1711.08467}}.

\bibitem{stanford1802}
D.~A. Roberts, D.~Stanford, and A.~Streicher.
\newblock ``Operator growth in the SYK model",
  \href{https://doi.org/10.1007/JHEP06(2018)122}{\textsl{Journal of High Energy
  Physics} \textbf{06} \texttt{122} (\textbf{2018})},
  \href{http://arxiv.org/abs/1802.02633}{\texttt{arXiv:1802.02633}}.

\bibitem{riddell}
R.~J.~Riddell Jr.
\newblock ``The number of Feynman diagrams",
  \href{https://doi.org/10.1103/PhysRev.91.1243}{\textsl{Physical Review}
  \textbf{91} 1243 (1953)}.

\bibitem{polchinski}
J.~Polchinski.
\newblock ``Combinatorics of boundaries in string theory",
  \href{https://doi.org/10.1103/PhysRevD.50.R6041}{\textsl{Physical Review}
  \textbf{D50} 6041 (1994)},
  \href{http://arxiv.org/abs/hep-th/9407031}{\texttt{arXiv:hep-th/9407031}}.

\bibitem{stanford1903}
P.~Saad, S.~H. Shenker, and D.~Stanford.
\newblock ``JT gravity as a matrix integral",
  \href{http://arxiv.org/abs/1903.11115}{\texttt{arXiv:1903.11115}}.

\bibitem{maldacena}
J.~M. Maldacena.
\newblock ``The large $N$ limit of superconformal field theories and
  supergravity",
  \href{https://doi.org/10.1023/A:1026654312961}{\textsl{Advances in
  Theoretical and Mathematical Physics} \textbf{2} 231 (1998)},
  \href{http://arxiv.org/abs/hep-th/9711200}{\texttt{arXiv:hep-th/9711200}}.

\bibitem{loeliger}
F.~R. Kschischang, B.~J. Frey, and H-A. Loeliger.
\newblock ``Factor graphs and the sum-product algorithm",
  \href{https://doi.org/10.1109/18.910572}{\textsl{IEEE Transactions on
  Information Theory} \textbf{47} 498 (2001)}.

\bibitem{stallings}
J.~R. Stallings.
\newblock ``Topology of finite graphs",
  \href{https://doi.org/10.1007/BF02095993}{\textsl{Inventiones Mathematicae}
  \textbf{71} 551 (1983)}.

\bibitem{chung}
F.~R.~K. Chung.
\newblock \emph{Spectral Graph Theory},
  \href{https://www.amazon.com/Spectral-Theory-Regional-Conference-Mathematics/dp/0821803158/ref=sr_1_1?ie=UTF8&qid=1520722134&sr=8-1&keywords=spectral+graph+theory}{(American
  Mathematical Society, 1997)}.

\bibitem{hazzard}
Z.~Wang and K.~R.~A. Hazzard.
\newblock ``Tightening the Lieb-Robinson bound in locally interacting systems",
  \href{https://journals.aps.org/prxquantum/abstract/10.1103/PRXQuantum.1.010303}{\textsl{PRX
  Quantum} \textbf{1} \texttt{010303} (2020)},
  \href{http://arxiv.org/abs/1908.03997}{\texttt{arXiv:1908.03997}}.

\bibitem{Konno_2005}
Norio Konno.
\newblock Limit theorem for continuous-time quantum walk on the line.
\newblock {\em Physical Review E}, 72(2), Aug 2005.

\bibitem{Lucas_2020}
A.~Lucas.
\newblock ``Non-perturbative dynamics of the operator size distribution in the
  Sachdev-Ye-Kitaev model",
  \href{https://doi.org/10.1063/1.5133964}{\textsl{Journal of Mathematical
  Physics} \textbf{61} \texttt{081901} (2020)},
  \href{http://arxiv.org/abs/1910.09539}{\texttt{arXiv:1910.09539}}.

\bibitem{nishimori}
T.~Kadowaki and H.~Nishimori.
\newblock ``Quantum annealing in the transverse Ising model,"
  \href{https://doi.org/10.1103/PhysRevE.58.5355}{\textsl{Physical Review}
  \textbf{E58} 5355 (1998)},
  \href{http://arxiv.org/abs/cond-mat/9804280}{\texttt{arXiv:cond-mat/9804280}}.

\bibitem{farhi}
E.~Farhi, J.~Goldstone, S.~Gutmann, J.~Lapan, A.~Lundgren, and D.~Preda.
\newblock ``A quantum adiabatic evolution algorithm applied to random instances
  of an NP-complete problem,"
  \href{https://doi.org/10.1126/science.1057726}{\textsl{Science} \textbf{292}
  472 (2001)},
  \href{http://arxiv.org/abs/quant-ph/0104129}{\texttt{arXiv:quant-ph/0104129}}.

\bibitem{lucas1903}
A.~Lucas.
\newblock ``Quantum many-body dynamics on the star graph",
  \href{http://arxiv.org/abs/1903.01468}{\texttt{arXiv:1903.01468}}.

\bibitem{arrachea}
L.~Arrachea and M.~J. Rozenberg.
\newblock ``The infinite-range quantum random Heisenberg magnet",
  \href{https://doi.org/10.1103/PhysRevB.65.224430}{\textsl{Physical Review}
  \textbf{B65} \texttt{224430} (2002)},
  \href{http://arxiv.org/abs/cond-mat/0203537}{\texttt{arXiv:cond-mat/0203537}}.

\bibitem{sedgewick}
P.~Flajolet and R.~Sedgewick.
\newblock \emph{Analytic Combinatorics},
  \href{https://www.amazon.com/Analytic-Combinatorics-Philippe-Flajolet/dp/0521898064/ref=sr_1_1?crid=TRX2GR853AJA&keywords=analytic+combinatorics&qid=1555695649&s=gateway&sprefix=analytic+co%2Caps%2C186&sr=8-1}{(Cambridge
  University Press, 2009)}.

\bibitem{shamir}
J.~Schmidt-Pruzan and E.~Shamir.
\newblock ``Component structure in the evolution of random hypergraphs",
  \href{https://doi.org/10.1007/BF02579445}{\textsl{Combinatorica} \textbf{5}
  81 (1985)}.

\bibitem{bender}
C.~M. Bender and S.~A. Orszag.
\newblock \emph{Asymptotic Methods and Perturbation Theory},
  \href{https://www.amazon.com/Advanced-Mathematical-Methods-Scientists-Engineers/dp/1441931872/ref=sr_1_1?crid=HW541ZOBYSIV&keywords=asymptotic+methods+and+perturbation+theory&qid=1556742150&s=gateway&sprefix=asymptotic+methods+%2Caps%2C-1&sr=8-1}{(Springer,
  2010)}.

\bibitem{chen2021concentration}
C-F. Chen.
\newblock ``Concentration of OTOC and Lieb-Robinson velocity in random
  Hamiltonians",
  \href{http://arxiv.org/abs/2103.09186}{\texttt{arXiv:2103.09186}}.

\bibitem{witten}
E.~Witten.
\newblock ``An SYK-like model without disorder",
  \href{http://arxiv.org/abs/1610.09758}{\texttt{arXiv:1610.09758}}.

\bibitem{gurau}
R.~Gurau.
\newblock ``The complete $1/N$ expansion of a SYK-like tensor model",
  \href{https://doi.org/10.1016/j.nuclphysb.2017.01.015}{\textsl{Nuclear
  Physics} \textbf{B916} 386 (2017)},
  \href{http://arxiv.org/abs/1611.04032}{\texttt{arXiv:1611.04032}}.

\bibitem{klebanov}
I.~R. Klebanov and G.~Tarnopolsky.
\newblock ``Uncolored random tensors, melon diagrams, and the SYK models",
  \href{https://doi.org/10.1103/PhysRevD.95.046004}{\textsl{Physical Review}
  \textbf{D95} \texttt{046004} (2017)},
  \href{http://arxiv.org/abs/1611.08915}{\texttt{arXiv:1611.08915}}.

\bibitem{gubser}
S.~S. Gubser, C.~Jepsen, Z.~Ji, and B.~Trundy.
\newblock ``Higher melonic theories",
  \href{https://doi.org/10.1007/JHEP09(2018)049}{\textsl{Journal of High Energy
  Physics} \textbf{09} \texttt{049} (\textbf{2018})},
  \href{http://arxiv.org/abs/1806.04800}{\texttt{arXiv:1806.04800}}.

\bibitem{osborne}
A.~Lucas and A.~Osborne.
\newblock ``Operator growth bounds in a cartoon matrix model",
  \href{https://doi.org/10.1063/5.0022177}{\textsl{Journal of Mathematical
  Physics} \textbf{61} \texttt{122301} (2020)},
  \href{http://arxiv.org/abs/2007.07165}{\texttt{arXiv:2007.07165}}.

\bibitem{lucas1809}
A.~Lucas.
\newblock ``Operator size at finite temperature and Planckian bounds on quantum
  dynamics", \href{http://arxiv.org/abs/1809.07769}{\texttt{arXiv:1809.07769}}.

\bibitem{stanfordbound}
J.~Maldacena, S.~H. Shenker, and D.~Stanford.
\newblock ``A bound on chaos",
  \href{http://link.springer.com/article/10.1007/JHEP08(2016)106}{\textsl{Journal
  of High Energy Physics} \textbf{08} \texttt{106} (\textbf{2016})},
  \href{http://arxiv.org/abs/1503.01409}{\texttt{arXiv:1503.01409}}.

\bibitem{alex1811}
X-L. Qi and A.~Streicher.
\newblock ``Quantum epidemiology: operator growth, thermal effects, and SYK",
  \href{http://arxiv.org/abs/1810.11958}{\texttt{arXiv:1810.11958}}.

\bibitem{xizhi}
X.~Han and S.~A. Hartnoll.
\newblock ``Quantum scrambling and state dependence of the butterfly velocity",
  \href{http://arxiv.org/abs/1812.07598}{\texttt{arXiv:1812.07598}}.

\bibitem{xiaochen}
X.~Chen, Y.~Gu, and A.~Lucas.
\newblock ``Many-body quantum dynamics slows down at low density",
  \href{https://doi.org/10.21468/SciPostPhys.9.5.071}{\textsl{SciPost Physics}
  \textbf{9} \texttt{071} (2020)},
  \href{http://arxiv.org/abs/2007.10352}{\texttt{arXiv:2007.10352}}.

\bibitem{yoshida}
Q.~Zhuang, T.~Schuster, B.~Yoshida, and N.~Y. Yao.
\newblock ``Scrambling and complexity in phase space",
  \href{http://arxiv.org/abs/1902.04076}{\texttt{arXiv:1902.04076}}.

\bibitem{sims}
B.~Nachtergaele, H.~Raz, B.~Schlein, and R.~Sims.
\newblock ``Lieb-Robinson bounds for harmonic and anharmonic lattices",
  \href{https://doi.org/10.1007/s00220-008-0630-2}{\textsl{Communications in
  Mathematical Physics} \textbf{286} 1073 (2009)},
  \href{http://arxiv.org/abs/0712.3820}{\texttt{arXiv:0712.3820}}.

\bibitem{bfss}
T.~Banks, W.~Fischler, S.~H. Shenker, and L.~Susskind.
\newblock ``M theory as a matrix model: a conjecture",
  \href{https://doi.org/10.1103/PhysRevD.55.5112}{\textsl{Physical Review}
  \textbf{D55} 5112 (1997)},
  \href{http://arxiv.org/abs/hep-th/9610043}{\texttt{arXiv:hep-th/9610043}}.

\bibitem{nahum}
A.~Nahum, S.~Vijay, and J.~Haah.
\newblock ``Operator spreading in random unitary circuits",
  \href{https://doi.org/10.1103/PhysRevX.8.021014}{\textsl{Physical Review}
  \textbf{X8} \texttt{021014} (2018)},
  \href{http://arxiv.org/abs/1705.08975}{\texttt{arXiv:1705.08975}}.

\bibitem{tibor}
C.~W. von Keyserlingk, T.~Rakovsky, F.~Pollmann, and S.~L. Sondhi.
\newblock ``Operator hydrodynamics, OTOCs, and entanglement growth in systems
  without conservation laws",
  \href{https://doi.org/10.1103/PhysRevX.8.021013}{\textsl{Physical Review}
  \textbf{X8} \texttt{021013} (2018)},
  \href{http://arxiv.org/abs/1705.08910}{\texttt{arXiv:1705.08910}}.

\bibitem{Tran_2019_polyLC}
M.~C. Tran, Guo A, Y, Y.~Su, J.~R. Garrison, Z.~Eldredge, M.~Foss-Feig, A.~M.
  Childs, and A.~V. Gorshkov.
\newblock ``Locality and Digital Quantum Simulation of Power-Law Interactions",
  \href{https://doi.org/10.1103/PhysRevX.9.031006}{\textsl{Physical Review}
  \textbf{X9} \texttt{031006} (2019)},
  \href{http://arxiv.org/abs/1808.05225}{\texttt{arXiv:1808.05225}}.

\bibitem{alpha_3_chenlucas}
Chi-Fang Chen and Andrew Lucas.
\newblock ``Finite speed of quantum scrambling with long range interactions",
  \href{https://doi.org/10.1103/PhysRevLett.123.250605}{\textsl{Physical Review
  Letters} \textbf{123} \texttt{250605} (2019)},
  \href{http://arxiv.org/abs/1907.07637}{\texttt{arXiv:1907.07637}}.

\bibitem{strictlylinear_KS}
T.~Kuwahara and K.~Saito.
\newblock ``Strictly Linear Light Cones in Long-Range Interacting Systems of
  Arbitrary Dimensions",
  \href{https://doi.org/10.1103/PhysRevX.10.031010}{\textsl{Physical Review}
  \textbf{X10} \texttt{031010} (2020)},
  \href{http://arxiv.org/abs/1910.14477}{\texttt{arXiv:1910.14477}}.

\bibitem{tran2021optimal}
M.~C. Tran, A.~Deshpande, A.~Y. Guo, A.~Lucas, and A.~V. Gorshkov.
\newblock ``Optimal state transfer and entanglement generation in power-law
  interacting systems",
  \href{http://arxiv.org/abs/2010.02930}{\texttt{arXiv:2010.02930}}.

\bibitem{hierarachy}
M.~C. Tran, C-F. Chen, A.~Ehrenberg, A.~Y. Guo, A.~Deshpande, Y.~Hong, Z-X.
  Gong, A.~V. Gorshkov, and A.~Lucas.
\newblock ``Hierarchy of Linear Light Cones with Long-Range Interactions",
  \href{https://doi.org/10.1103/PhysRevX.10.031009}{\textsl{Physical Review}
  \textbf{X10} \texttt{031009} (2020)},
  \href{http://arxiv.org/abs/2001.11509}{\texttt{arXiv:2001.11509}}.

\bibitem{Kuwahara_OTOC}
T.~Kuwahara and K.~Saito.
\newblock ``Absence of Fast Scrambling in Thermodynamically Stable Long-Range
  Interacting Systems",
  \href{https://doi.org/10.1103/PhysRevLett.126.030604}{\textsl{Physical Review
  Letters} \textbf{126} \texttt{030604} (2021)},
  \href{http://arxiv.org/abs/2009.10124}{\texttt{arXiv:2009.101244=}}.

\bibitem{chen2021optimal}
C-F. Chen and A.~Lucas.
\newblock ``Optimal Frobenius light cone in spin chains with power-law
  interactions",
  \href{http://arxiv.org/abs/2105.09960}{\texttt{arXiv:2105.09960}}.

\bibitem{altman}
D.~E. Parker, X.~Cao, A.~Avdoshkin, T.~Scaffidi, and E.~Altman.
\newblock ``A Universal Operator Growth Hypothesis",
  \href{https://journals.aps.org/prx/abstract/10.1103/PhysRevX.9.041017}{\textsl{Physical
  Review} \textbf{X9} \texttt{041017} (2019)},
  \href{http://arxiv.org/abs/1812.08657}{\texttt{arXiv:1812.08657}}.

\bibitem{avdoshkin}
A.~Avdoshkin and A.~Dymarsky.
\newblock ``Euclidean operator growth and quantum chaos",
  \href{https://journals.aps.org/prresearch/abstract/10.1103/PhysRevResearch.2.043234}{\textsl{Physical
  Review Research} \textbf{2} \texttt{043234} (2020)},
  \href{http://arxiv.org/abs/1911.09672}{\texttt{arXiv:1911.09672}}.

\end{thebibliography}

\end{document}